\documentclass[10pt,twocolumn,aps,english,pra,superscriptaddress,longbibliography]{revtex4-2} 

\usepackage{xr}
\usepackage{tikz} 
\usepackage{amsmath,amsfonts,amsthm, mathtools}
\usepackage{placeins}
\usepackage{float}
\usepackage{dsfont}
\usepackage{longtable}
\usepackage{csquotes}
\usepackage{amssymb}
\usepackage{verbatim}
\usepackage{bm}
\usepackage{amsmath}
\usepackage{amssymb}
\usepackage{stmaryrd}
\usepackage{amsthm}
\usepackage[caption=false]{subfig}
\usepackage{physics}
\usepackage{bm}  
\usepackage{algorithm}
\usepackage{algorithmic}

\usepackage{booktabs}
\usepackage[columnwise]{lineno}
\usepackage{hyperref}
\definecolor{darkred}  {rgb}{0.5,0,0}
\definecolor{darkblue} {rgb}{0,0,0.5}
\definecolor{darkgreen}{rgb}{0,0.5,0}
\hypersetup{
  colorlinks = true,
  urlcolor  = blue,         
  linkcolor = darkblue,     
  citecolor = darkgreen,    
  filecolor = darkred       
}

\theoremstyle{definition}

\newtheorem{proposition}{Proposition}

\usepackage{multirow}

\definecolor{cool_green}{rgb}{0.0, 0.5, 0.0}

\newcommand{\yk}[1]{#1}
\newcommand{\ykwang}[1]{{#1}}
\newcommand{\sisi}[1]{\textcolor{teal}{[Sisi: #1]}}

\begin{document}


\title{Advancing quantum imaging through learning theory}

\author{Yunkai Wang}
\email{ywang10@perimeterinstitute.ca}
\affiliation{Perimeter Institute for Theoretical Physics, Waterloo, Ontario N2L 2Y5, Canada.}
\affiliation{Department of Applied Mathematics, University of Waterloo, Ontario N2L 3G1, Canada.}
\affiliation{Institute for Quantum Computing, University of Waterloo, Ontario N2L 3G1, Canada.}

\author{Changhun Oh}
\affiliation{Department of Physics, Korea Advanced Institute of Science and Technology, Daejeon 34141, Republic of Korea.}

\author{Junyu Liu}
\affiliation{Department of Computer Science, The University of Pittsburgh, Pittsburgh, PA 15260, USA.}
\affiliation{Pritzker School of Molecular Engineering, The University of Chicago, Chicago, IL 60637, USA.}

\author{Liang Jiang}
\email{liangjiang@uchicago.edu}
\affiliation{Pritzker School of Molecular Engineering, The University of Chicago, Chicago, IL 60637, USA.}

\author{Sisi Zhou}
\email{sisi.zhou26@gmail.com}
\affiliation{Perimeter Institute for Theoretical Physics, Waterloo, Ontario N2L 2Y5, Canada.}
\affiliation{Department of Applied Mathematics, University of Waterloo, Ontario N2L 3G1, Canada.}
\affiliation{Institute for Quantum Computing, University of Waterloo, Ontario N2L 3G1, Canada.}
\affiliation{Department of Physics and Astronomy, University of Waterloo, Ontario N2L 3G1, Canada.}

\begin{abstract}


\yk{We study quantum imaging by applying the resolvable expressive capacity (REC) formalism developed for physical neural networks (PNNs). In this paradigm of quantum learning, the imaging system functions as a physical learning device that maps input parameters to measurable features, while complex practical tasks are handled by training only the output weights, enabled by the systematic identification of well-estimated features (eigentasks) and their corresponding sample thresholds. Using this framework, we analyze both direct imaging and superresolution strategies for compact sources, defined as sources with sizes bounded below the Rayleigh limit. In particular, we introduce the orthogonalized SPADE method—a nontrivial generalization of existing superresolution techniques—that achieves superior performance when multiple compact sources are closely spaced. This method relaxes the earlier superresolution studies' strong assumption that the entire source must lie within the Rayleigh limit, marking an important step toward developing more general and practically applicable approaches. \ykwang{Using the example of face recognition, which involve complex structured sources,} we demonstrate the superior performance of our orthogonalized SPADE method and highlight key advantages of the quantum learning approach—its ability to tackle complex imaging tasks and enhance performance by selectively extracting well-estimated features.
 }

\end{abstract}

\maketitle

\textit{Introduction.---} \yk{The quality of an image formed by a single-lens system depends on several factors, including the lens's resolution, the measurement strategy employed in the image plane, and the number of collected samples.} Notably, it has been shown that by optimizing the measurement design, one can resolve two-point sources within the Rayleigh limit, surpassing the Rayleigh's criterion \cite{tsang2016quantum}. 
\yk{The concept of  superresolution has been extended in various directions~\cite{pirandola2018advances}, focusing on imaging a single compact source—an object much smaller than the Rayleigh limit. These extensions include more careful treatment of the measurement and data analysis for two point sources~\cite{sorelli2021optimal,grace2020approaching}, general sources within the Rayleigh limit~\cite{tsang2019quantum,tsang2017subdiffraction,zhou2019modern}, sources beyond the weak-source limit~\cite{wang2021superresolution,nair2016far,lupo2016ultimate}, and point sources in higher dimensions~\cite{napoli2019towards,yu2018quantum,ang2017quantum}. These theories have been experimentally demonstrated for estimating point sources under various scenarios~\cite{yang2016farfield,tang2016fault,paur2016achieving,tham2017beating,parniak2018beating,santamaria2024single,tan2023quantum,rouviere2024ultra} and for estimating source moments~\cite{tan2023quantum}.}

\ykwang{Superresolution often relies on conventional statistical tools, such as the Fisher-information-matrix (FIM) approach for quantifying parameter-estimation precision \cite{pirandola2018advances,sorelli2021optimal,grace2020approaching,tsang2019quantum,tsang2017subdiffraction,zhou2019modern,wang2021superresolution,nair2016far,lupo2016ultimate,napoli2019towards,yu2018quantum,ang2017quantum,yang2016farfield,tang2016fault,paur2016achieving,tham2017beating,parniak2018beating,santamaria2024single,tan2023quantum,rouviere2024ultra,tan2023quantum}, and on the Chernoff bound or likelihood-ratio method for discrimination tasks \cite{zanforlin2022optical,huang2021quantum,zhang2020super,lu2018quantum,grace2022identifying}. However, these conventional statistical tools face at least two fundamental challenges when applied to complex imaging tasks.
(i) Modeling complexity and ambiguity. It is often unclear which parameterization is most appropriate for a practical imaging task. Image moments, Fourier coefficients, and pixel intensities all provide complete representations but lead to different interpretations and performance. In addition, the likelihood-ratio method requires a full statistical model of the object, which is extremely difficult to construct for complex tasks—for example, in face recognition.
(ii) Finite-sample reliability. With limited data, only a  subset of features can be estimated accurately; many others are dominated by noise.  Effective use of the data therefore requires identifying and retaining only well-estimated features for downstream analysis—an especially difficult step in imaging, where the number of parameters is, in principle, infinite. }

The machine learning approach has achieved tremendous success in imaging applications, enabling the handling of complex practical tasks. In this work, we implement imaging tasks with a paradigm of quantum learning—physical neural networks (PNNs)~\cite{boyd1985fading, tanaka2019recent, mujal2021opportunities, wilson2018quantum, garcia-beni2023scalable, havlicek2019supervised, rowlands2021reservoir, lin2018all-optical, pai2023experimentally, dambre2012information, sheldon2022computational, schuld2021effect, wu2021expressivity, hu2023tackling}—following in particular the resolvable expressive capacity (REC) formalism in Ref.~\cite{hu2023tackling} to address the limitations of the FIM approach. The quantum learning approach adopted here encompasses both model training and inference using the trained model.
PNNs encode inputs—such as the positions of incoherent point sources—into an analog physical system whose evolution is fixed and governed by its underlying physics, mapping the inputs to the measured high-dimensional features.
A key advantage is that practical tasks can be performed by training only the output weight—equivalent to applying linear or logistic regression to the measured features—while leaving the internal structure unchanged. The relationship between the specific system dynamics and the class of tasks it can realize is a fundamental question in PNNs. The REC formalism addresses this question by identifying the resolvable features and quantifying their achievable precision in the presence of finite-sample noise \cite{hu2023tackling}. The imaging system is an analog physical system that naturally functions as a learning device in PNNs, with output weights trained to perform imaging tasks. \ykwang{And applying the REC formalism to imaging identifies well-estimated features through eigentasks that are invariant under parameterization, thereby guiding the formulation of complex practical tasks for a given imaging system and prior information, directly addressing the first challenge faced by conventional statistical tools. }
Moreover, it provides a way to estimate the sample threshold required to detect each eigentask. This makes it particularly useful for determining which measured features can be reliably used in downstream analysis by selecting low-noise eigentasks based on a threshold. \ykwang{This key advantage, highlighted in the original paper~\cite{hu2023tackling}, can enhance the performance of diverse machine learning tasks, such as classification, regression, and clustering, and  addresses the second challenge faced by conventional statistical tools.} A more detailed introduction to the PNNs is provided in Sec.~\ref{SI:PNNs} of the Supplemental Material.



\yk{Besides adopting a quantum learning perspective to study the imaging problem, we extend the existing discussions on superresolution \cite{tsang2016quantum,pirandola2018advances,sorelli2021optimal,grace2020approaching,tsang2019quantum,tsang2017subdiffraction,zhou2019modern,wang2021superresolution,nair2016far,lupo2016ultimate,napoli2019towards,yu2018quantum,ang2017quantum,yang2016farfield,tang2016fault,paur2016achieving,tham2017beating,parniak2018beating,zanforlin2022optical,santamaria2024single,tan2023quantum,rouviere2024ultra,tan2023quantum} to the broader challenge of imaging multiple compact sources, each individually constrained within the Rayleigh limit, or equivalently general sources that exceed the Rayleigh limit while exhibiting clustered substructures.  This motivation arises as follows: prior studies typically assume that the whole source lies within the Rayleigh limit, i.e., it is a single compact source. However, in practical imaging, sources are often larger than the Rayleigh limit, while we aim to resolve fine features within the source that are below the Rayleigh limit. To address this, a natural idea is to partition the source into compact regions and apply superresolution techniques locally, effectively treating the problem as imaging multiple compact sources.
But a straightforward application of superresolution to individual compact sources, referred to as the \emph{separate SPADE method}, fails to offer an advantage over direct imaging. This is because nearby sources, when separated by distances not much greater than the width of the point spread function (PSF), introduce additional noise into the measurement. However, our generalized approach to superresolution, called the \emph{orthogonalized SPADE method}, can achieve superresolution even when the separation between compact sources is as small as the PSF width.
This discussion advances superresolution to imaging multiple reasonably nearby compact sources, constituting a nontrivial generalization of the earlier SPADE method and a step toward making the approach more practical.}

\textit{Preliminary.---} \yk{We now provide a more detailed explanation of how imaging tasks are realized using PNNs. The imaging system can be regarded as an input–output map, where the input is a set of system parameters $\boldsymbol{\theta}$ (e.g., the locations of point sources) and the output is a set of measured features, which are the probabilities $P_j(\boldsymbol{\theta}) = \mathrm{Tr}[\rho(\boldsymbol{\theta}) M_j]$, where $M_j$ denotes the $j$th element of the positive operator-valued measure (POVM). A general learning task—such as classification or regression—can be formulated in terms of a target function $f(\boldsymbol{\theta})$ of the input parameters. Realizing an imaging task is then equivalent to approximating this target function using the measured features, which corresponds to training the output weights of the PNNs \cite{boyd1985fading,  tanaka2019recent, mujal2021opportunities,  wilson2018quantum, garcia-beni2023scalable, havlicek2019supervised, rowlands2021reservoir, lin2018all-optical, pai2023experimentally, dambre2012information, sheldon2022computational, schuld2021effect, wu2021expressivity}.

For a given physical system, it is important to determine the class of functions that can be approximated using the measured features. Moreover, due to sampling noise, the quantities $\bar{P}_j(\boldsymbol{\theta})$ can only be estimated approximately. It is also necessary to quantify the precision with which these functions can be approximated. These questions, concerning the capability of a physical system when regarded as a PNN, are addressed by the recently developed REC formalism \cite{dambre2012information,hu2023tackling}. REC formalism analyzes the resolvable function space of a physical system under sample noise via the concept of REC,
\begin{equation}\label{eq:C_f}
    C[f]=1-\min _W \frac{\mathbb{E}_{\boldsymbol{\theta}} \left[\mathbb{E}_{\mathcal{X}}\left[\left(\sum_{j} W_j\bar{P}_j(\boldsymbol{\theta})-f(\boldsymbol{\theta})\right)^2\right]\right]}{\mathbb{E}_{\boldsymbol{\theta}}[ f(\boldsymbol{\theta})^2 ]},
\end{equation}
where we take the expectation value for the output samples $\mathcal{X}$ and the prior distribution $p(\boldsymbol{\theta})$,  $f(\boldsymbol{\theta})$ is approximated by  a linear combination of measured functions $\sum_{j} W_j\bar{P}_j(\boldsymbol{\theta})$, $W_j$ represents the weight coefficient in front of $\bar{P}_j(\boldsymbol{\theta})$ to be optimized to achieve the optimal linear approximation of $f(\boldsymbol{\theta})$, where the index $j$ corresponds to different measurement outcomes. REC $C[f]$, which takes values between 0 and 1, can be understood as the normalized mean-squared accuracy of approximating $f(\boldsymbol{\theta})$, where $C[f]=1$ represents a perfect approximation using the measured features, and deviation from 1 indicates that the target function cannot be well approximated.

\begin{figure*}[!t]
\begin{center}
\includegraphics[width=1.8\columnwidth]{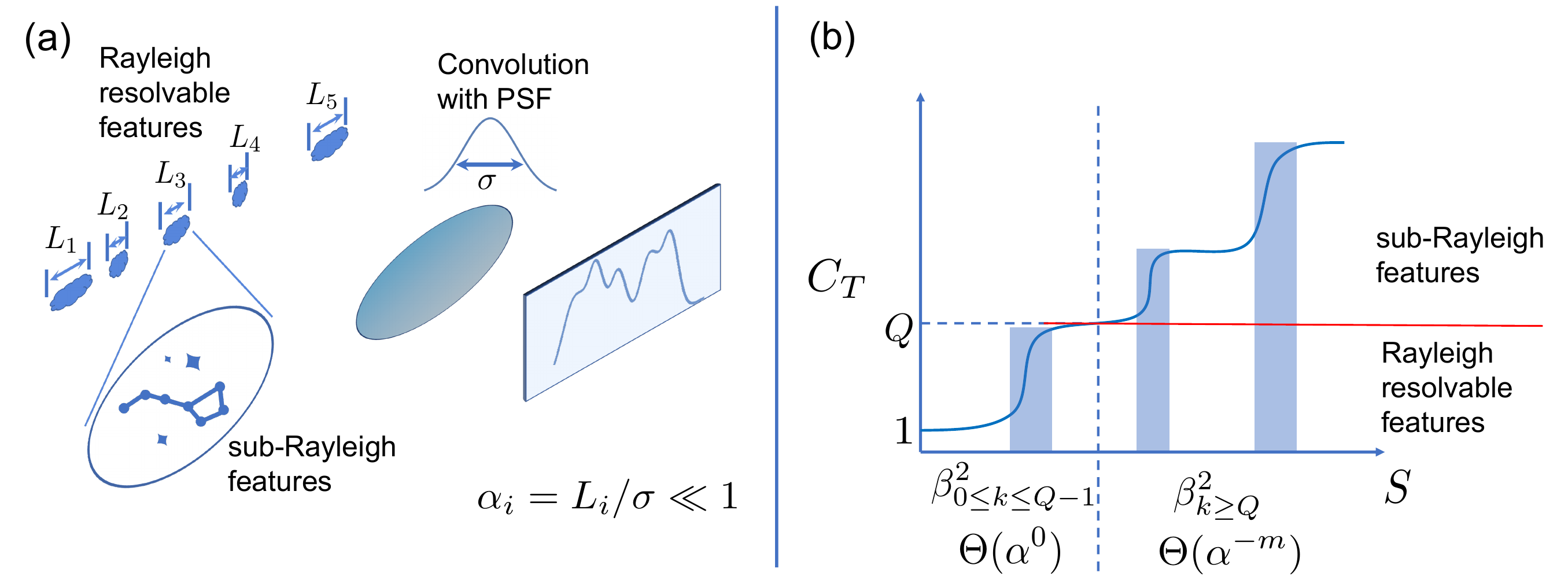}
\caption{ Imaging of both Rayleigh resolvable features and sub-Rayleigh features. (a) Multiple compact sources of size $L_i$ are imaged by a lens with a PSF width of $\sigma$. 
Besides Rayleigh-resolvable features, we would also like to extract information from sub-Rayleigh features, which are associated with the small parameter $\alpha=L/\sigma$. 
(b) The total REC, $C_T$, which shows a stepwise increase, is plotted as a function of the number of samples, $S$. 
The threshold of $S$ for each stepwise increase of $C_T$ in the shaded region is determined by the eigenvalue $\beta_k^2$ associated with the corresponding eigentask in learning. 
Each time $C_T$ increases by 1, there is a corresponding eigenvalue $\beta_k^2$, with the sample number threshold following $S \sim \Theta(\beta_k^2)$.
The intensity of each compact source corresponds to resolvable features, which can be imaged with a constant number of samples, scaling as $S \sim \beta^2_{0\leq k\leq Q-1}=\Theta(\alpha^0)$, independent of the source size. In contrast, sub-Rayleigh features that reveal detailed information about each compact source require a number of samples scaling inversely with the source size,  following $S \sim \beta^2_{ k\geq Q}= \Theta(\alpha^{-m})$, where $\alpha$ is determined by the compact sources, and $m$ depends on the order of  moments. } 
\label{set_up}
\end{center}
\end{figure*}

To quantify the overall performance of an imaging system as PNNs, we are interested in identifying the set of functions $f(\boldsymbol{\theta})$ that can be approximated in this way and in quantifying the effective size of this set, given the physical imaging system, the finite number of samples $S$, and prior knowledge about the input $\boldsymbol{\theta}$. The total REC $C_T := \sum_{k} C[g_k]$~\cite{hu2023tackling} fulfills this purpose, where $\{g_k\}_k$ can be any complete orthonormal basis of functions in the Hilbert space equipped with inner product $\mathbb{E}_{\boldsymbol{\theta}}[g_k(\boldsymbol{\theta})g_\ell(\boldsymbol{\theta})]$. The value of total REC can be obtained from the following eigenvalue problem, 
\allowdisplaybreaks
\begin{gather}\label{eq:D_G_main_text}
D_{k j}=\delta_{kj}\operatorname{Tr}\left\{{M}_k \hat{\rho}^{(1)}\right\}, \, G_{j k}=\operatorname{Tr}\left\{\left({M}_j \otimes {M}_k\right) \hat{\rho}^{(2)}\right\},\\
\hat{\rho}^{(t)}=\mathbb{E}_{\boldsymbol{\theta}} [ {\rho}(\boldsymbol{\theta})^{\otimes t} ],\\
V=D-G,\quad V{r}_{k}=\beta_k^2 G {r}_{k},\label{eq:REC_equations}
\end{gather}
where $\beta_k^2$ and $r_k$ are the $k$th eigenvalue and eigenvector (in increasing order). 
The eigenbasis $r_k$ correspond to a minimal set of  eigentasks $f_k(\boldsymbol{\theta}) :=\sum_m r_{km}P_m(\boldsymbol{\theta})$ that saturate the available REC of the system in the space of all functions of input parameters $\boldsymbol{\theta}$. Then 
\begin{gather}\label{eq:C_T_f}
C_T=\sum_{k}C[f_k]=\sum_{k} \frac{1}{1+\beta_k^2 / S},
\end{gather}
where $S$ is the number of samples. 
Intuitively, $C_T$ quantifies how many independent features of the underlying signal can be captured by the measurement process, and capturing more features improves performance on complex learning tasks by providing greater freedom to approximate the target function $f(\boldsymbol{\theta})$.
Treating the imaging system as a PNN involves first identifying the eigentasks in the REC formalism and then training the output weights—essentially performing logistic or linear regression using the values of eigentasks $f_k(\boldsymbol{\theta})$—while including only the well-estimated eigentasks under a finite sample size and discarding noisy ones to ensure optimal performance. Another important property of this formalism is that reparameterization leaves both the total REC and the eigentasks unchanged, making the strategy independent of any artificial parameterization and determined solely by the physical system and the structure of the problem. 
A more detailed introduction to this formalism, proposed in Ref.~\cite{hu2023tackling}, is provided in Sec.~\ref{SI:REC} of the Supplemental Material.}

\yk{
When we apply this quantum learning approach to the  superresolution setting of imaging multiple compact sources, as shown in Fig.~\ref{set_up}, we observe that the images exhibit two types of features: Rayleigh resolvable features and sub-Rayleigh features. The Rayleigh resolvable features are determined by the intensity of each compact source, which can be measured with a constant number of samples independent of the source size for both direct imaging and superresolution methods. In contrast, the sub-Rayleigh features involve details below the Rayleigh limit, requiring the number of samples to scale inversely with the source size, where carefully designed superresolution methods demonstrate clear advantages. 
Intuitively, the total REC quantifies the number of reliably estimated features—both large and small—and increases as smaller features become resolvable. We first examine the distinctive behavior of the eigenvalues $\beta_k^2$ and the total REC $C_T$ in each case of superresolution, and introduce our orthogonalized SPADE method. We then demonstrate the advantages of both this quantum learning approach and our new orthogonalized SPADE method through a concrete learning task.
}

\textit{Resolving two-point sources.---}As a simplest example, we begin with the imaging 
of two incoherent point sources in one dimension. A single photon received on the image plane can be described as
$\rho(L)=\frac{1}{2}(\ket{\psi_1}\bra{\psi_1}+\ket{\psi_2}\bra{\psi_2})$,
where $\ket{\psi_i}=\int du\psi(u-u_i)\ket{u}$, $\ket{u}=a^\dagger_u\ket{0}$ is the single photon state at position $u$,  and we choose the PSF $\psi(u)={\exp(-u^2/4\sigma^2)}/{(2\pi\sigma^2)^{1/4}}$. Define the separation $L$, which is the input of learning task $\boldsymbol{\theta}=L$ and assume $u_2=L/2$, $u_1=-L/2$.
To enable analytical analysis, we focus on the binary SPADE measurement, which is capable of achieving superresolution in resolving two-point sources as introduced in Ref.~\cite{tsang2016quantum}, \yk{where POVM
$M_0=\ket{\phi_0}\bra{\phi_0}$, $M_1=I-M_0$,} $\ket{\phi_0}=\int du\phi_0(u)\ket{u}$, $\phi_0(u)=\frac{1}{(2\pi\sigma^2)^{1/4}}\exp(-\frac{u^2}{4\sigma^2})$. 
Assuming the prior knowledge about the separation is described as 
$p(L)=\frac{1}{\sqrt{2\pi}\gamma}\exp(-\frac{L^2}{2\gamma^2})
$, we can calculate the total REC, $C_T$, which here represents the total number of linearly independent functions $f(L)$ that can be expressed as a linear combination of the measured probabilities $P_k=\tr(\rho(L)M_k)$. 
 \yk{Assuming $\gamma \ll \sigma$ to exhibit the advantage of superresolution within the Rayleigh limit,} we find that 
\begin{equation}\begin{aligned}\label{eq:two_point_beta}
&\beta_0^2=0,\\
&\beta_1^2=\frac{8}{\alpha^2} + \frac{3}{4} - \frac{1}{64 }\alpha^2+O(\alpha^4), \quad \alpha=\gamma/\sigma,
\end{aligned}
\end{equation}
where $\alpha$ is roughly the ratio between the separation and width of the PSF and $\alpha\ll1$ when the two-point sources are very close to each others. 


We can compare this with the direct imaging case, where we directly project onto each spatial mode ${E_x=\ket{x}\bra{x}}_x$. In this case, we find that $\beta_0^2 = O(1)$, $\beta_1^2 = \Theta(\alpha^{-4})$, and $\beta_2^2 = \Theta(\alpha^{-8})$ after discretizing the spatial coordinate. \yk{More details of the calculations for both the direct imaging and SPADE methods are provided in Sec.~\ref{SI:two_point} of the Supplemental Information.} The much larger eigenvalue $\beta_1^2$ in direct imaging indicates poorer performance compared to binary SPADE, as it requires a larger number of samples $S$ to achieve the same $C_T$.



\yk{\textit{Resolving a single compact source.---}We now consider the problem of imaging a single compact source, defined as a generally distributed source whose spatial extent is bounded well below the Rayleigh limit. This represents the most general setting for applying superresolution in imaging that has been considered in previous works \cite{pirandola2018advances,sorelli2021optimal,grace2020approaching,tsang2019quantum,tsang2017subdiffraction,zhou2019modern,wang2021superresolution,nair2016far,lupo2016ultimate,napoli2019towards,yu2018quantum,ang2017quantum,yang2016farfield,tang2016fault,paur2016achieving,tham2017beating,parniak2018beating,zanforlin2022optical,santamaria2024single,tan2023quantum,rouviere2024ultra,tan2023quantum}.}
 Assume the normalized source intensity $I(u)$ 
is confined within the interval $[-L/2, L/2]$. We can define the moments as $\int du \, I(u) \left(\frac{u-u_0}{L}\right)^{n} = x_{n}$, which completely describe the source and are the input for the learning task $\boldsymbol{\theta}=\vec{x}=[x_0,x_1,x_2,\cdots]$. Within the Rayleigh limit,  $\alpha = L/\sigma \ll 1$, where $\sigma$ represents the width of the PSF, the size of the compact source is significantly smaller than the resolution limit. \yk{For any prior $p(\vec{x})$ and PSF, we find that for the direct imaging}
\begin{equation}\label{beta_direct1}
\beta_0^2=0,\; \beta_1^2=\Theta(\alpha^{-2}), \;\beta_2^2=\Theta(\alpha^{-4}),\; \cdots
 \end{equation}
where $\beta_0^2=0$ is a trivial eigenvalue which corresponds to the fact that $\sum_m P_m=1$.

\begin{figure}[!tb]
\begin{center}
\includegraphics[width=0.8\columnwidth]{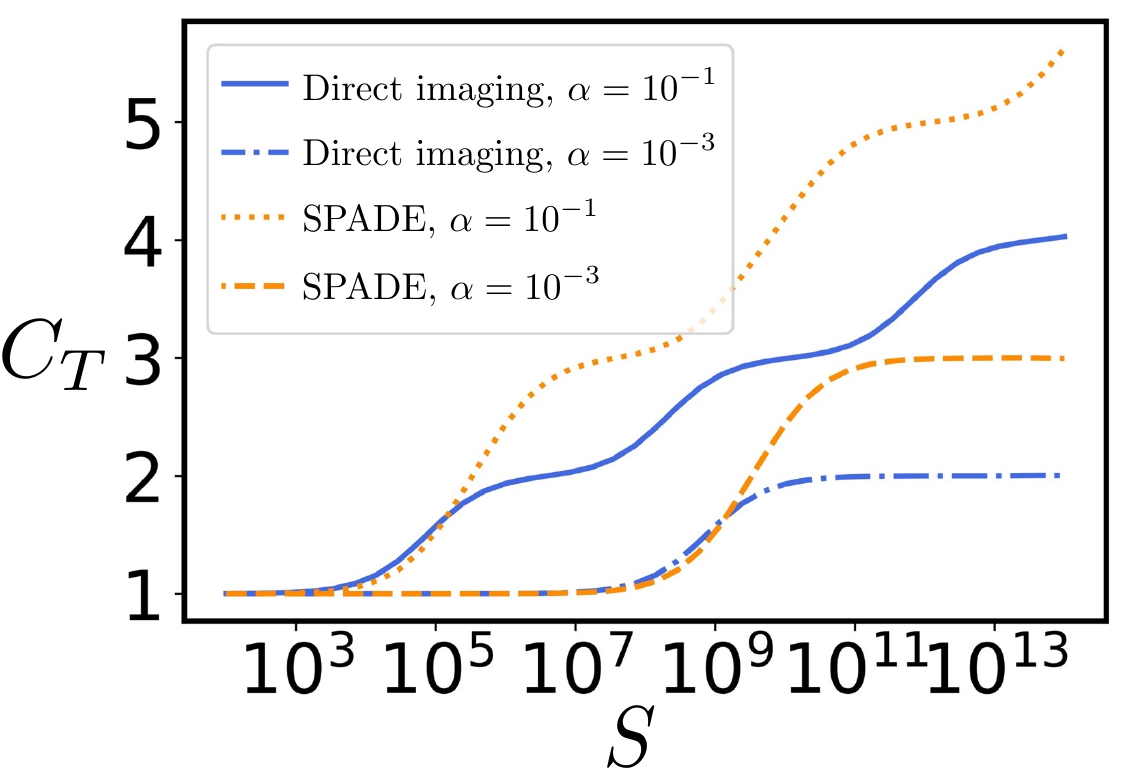}
\caption{ Total REC $C_T$ for direct imaging and the SPADE method as a function of $S$, when imaging one generally distributed compact source with different $\alpha$. 
}
\label{CT_combined}
\end{center}
\end{figure}


\yk{For superresolution, we adopt the measurement construction from Ref. \cite{zhou2019modern}, as reviewed in Sec. \ref{SI:preliminary_superresolution} of the Supplemental Material. For simplicity, we still refer to this method as the SPADE method  throughout the discussion. Intuitively, the SPADE method isolates higher-order moments by constructing probability distributions that start at higher-order terms of $\alpha$, which serves as the signal strength, thereby suppressing lower-order terms of $\alpha$ that act as noise. This yields a much better signal-to-noise ratio for estimating those moments, especially in the weak-signal regime. For any prior $p(\vec{x})$ and PSF, we find that for the SPADE method}
\begin{gather}
\label{beta_superresolution1}
\beta_0^2=0, \;\beta_1^2=\Theta(\alpha^{-2}), \;\beta_2^2=\Theta(\alpha^{-2}),\\
\beta_3^2=\Theta(\alpha^{-4}), \; \beta_4^2=\Theta(\alpha^{-4}),\,\cdots\notag
\end{gather}
Compared to direct imaging, the SPADE method achieves smaller $\beta_k^2$, which significantly reduces the required $S$ to achieve the same $C_T$.


{\color{blue}
}

We demonstrate the significance of $\beta_k^2$ as the threshold for the stepwise increase in the total REC $C_T$, as shown in Fig.~\ref{CT_combined}. \yk{The total REC $C_T = \sum_{k} \frac{1}{1 + \beta_k^2 / S}$ shows that each $\beta_k^2$ sets the sample size at which its eigentask contributes significantly, with contributions near 1 when $S \gg \beta_k^2$ and negligible when $S \ll \beta_k^2$.}
For direct imaging, $C_T$ increases by 1 at each step, while for the SPADE method, $C_T$ increases by 2 per step, as expected. As $\alpha$ decreases, the plateau regions expand. \yk{All numerical calculations in this work assume a Gaussian PSF $\psi(u) = \exp(-u^2/4\sigma^2)/(2\pi\sigma^2)^{1/4}$; however, our method is applicable to any PSF. We choose the prior distribution for the moment vectors $\vec{x}$ by randomly generating a set of images and assuming they occur with equal probability, thereby establishing $p(\vec{x})$ as the empirical distribution of the resulting moment vectors, as detailed in Sec.~\ref{SI:numerics} of the Supplemental Material. For illustration purpose, we plot the data for only one instance of the randomly generated prior distributions (and in all figures below where the prior is picked randomly). }

\yk{Note that the threshold of $S$ is not precisely located at $\alpha^{-2n}$ in each case. This deviation arises from a constant prefactor in $\beta_k^2$. This constant prefactor is independent of $\alpha$ and is approximately $10^2$ in Fig.~\ref{CT_combined}. This is reasonable because, even in the simpler imaging task where sources are extended outside the Rayleigh limit (i.e., $\alpha \gg 1$), hundreds or more samples are still required to effectively image a source. The prefactors depend on the imaging strategy and the prior information. We provide a more detailed discussion of these prefactors in Sec.~\ref{SI:prefactor} of the Supplemental Material. }

\yk{The $n$th eigentask corresponds to $\sum_m r_{nm} P_m$ as a function of $\vec{x}$, where $r_{nm}$ are the components of the eigenvectors obtained by solving Eq.~\ref{eq:REC_equations}, and the corresponding REC is given by $1 / (1 + \beta_n^2 / S)$. For direct imaging of a single compact source, we observe that the eigentasks converge to $x_n$ in the limit of small $\alpha$ to the leading order. For the SPADE method, we show that $r_{nm}$ becomes an triangular matrix as $\alpha \to 0$. The first two leading-order terms of both the $2k$-th and $(2k+1)$-th eigentasks have coefficients $x_{2k}$ and $x_{2k+1}$. 
}

Further details on the derivation of the scaling of $\beta_k^2$ and the eigenvectors $r_n$, based on perturbation theory and confirmed by numerical calculations, are provided in Sec.~\ref{SI:single_compact_source} of the Supplemental Material. 
So far, our discussion has focused on sources where the entire source lies within the Rayleigh limit. In this case, 
the total intensity, trivially equal to $1$, is the Rayleigh resolvable feature that contributes to the total REC when only a constant number of samples (smaller than $1/\alpha^2$) is available. 
\yk{The sub-Rayleigh features contribute to the total REC when $\Omega(\alpha^{-2})$ samples are available.} As we will see later, the Rayleigh resolvable features can become nontrivial when dealing with multiple compact sources.

\textit{New superresolution methods on multiple compact sources.}---We now want to consider the scenario where we have multiple compact sources, each with a size within the Rayleigh limit, but collectively distributed over a region larger than the Rayleigh limit. The quantum state from these multiple compact sources is given by
\begin{equation}\label{eq:rho_multiple}
\rho=\sum_{q=1}^Q\int du du_1du_2I_q(u)\psi(u-u_1)\ket{u_1}\bra{u_2}\psi^*(u-u_2),
\end{equation}
where $Q$ is the number of compact sources, $I_q(u)$ is the intensity distribution for $q$th compact source. We can expand near the centroid $u_q$ of $q$th source and reorganize the state as
\begin{equation}\begin{aligned}\label{main_eq:rho_expansion}
&\rho=\sum_{q=1}^Q \sum_{m,n=0}^\infty x_{m+n,q}\ket{\psi_q^{(m)}}\bra{\psi_q^{(n)}},\\
&\ket{\psi_q^{(m)}}=\int du \psi_q^{(m)}(u)\ket{u}, \\
&\psi^{(n)}_q(v)=\frac{\partial^n \psi(v-u)}{\partial u^n}\bigg|_{u=u_q}\frac{L_q^n}{n!},
\end{aligned}\end{equation}
where $L_q$ is the size (diameter) of $q$th source, $x_{n,q}=\int du I_q(u)\left(\frac{u-u_q}{L_q}\right)^n $ is the $n$th moment for the $q$th source and are the input for the learning task $\boldsymbol{\theta}=[x_{0,1},x_{1,1},x_{2,1},\cdots,x_{0,Q},x_{1,Q},x_{2,Q},\cdots]$.  \yk{We find that for the direct imaging}
\begin{gather}\label{eq:beta_DI_single}
\beta_{0}^2=0,\quad\beta_{1\leq i\leq Q-1}^2=\Theta(1),\\
\beta_{Q\leq i\leq 2Q-1}^2=\Theta(\alpha^{-2}), \,
\beta_{2Q\leq i\leq 3Q-1}^2=\Theta(\alpha^{-4}),\cdots\notag
\end{gather}
where $\alpha_q =\max{L_q}/\sigma$. Here, we assume that $L_q$ does not differ significantly, allowing the different $L_q$ values to be incorporated into the constant coefficients. \yk{It is then clear that there are $Q$ Rayleigh resolvable features corresponding to the intensity of each compact sources, and the number of sub-Rayleigh features also increases by a factor of $Q$.  The scaling is numerically confirmed in Fig.~\ref{beta_new_method}(a) for the first six eigenvalues, and we expect $\beta_k^2=\Theta(\alpha^{-2\lfloor k/Q\rfloor})$ to hold for eigenvalues with higher indices with any prior $p(\vec{x})$. In the numerical calculation, we assume two compact sources with centroids at $-L/4$ and $L/4$, 
with a random prior distribution for the moments (by randomly generating a set of images). For both $L = 2$ and $L = 20$, we observe the same scaling behavior in the direct imaging method. 
}


\begin{figure}[!tb]
\begin{center}
\includegraphics[width=0.5\textwidth]{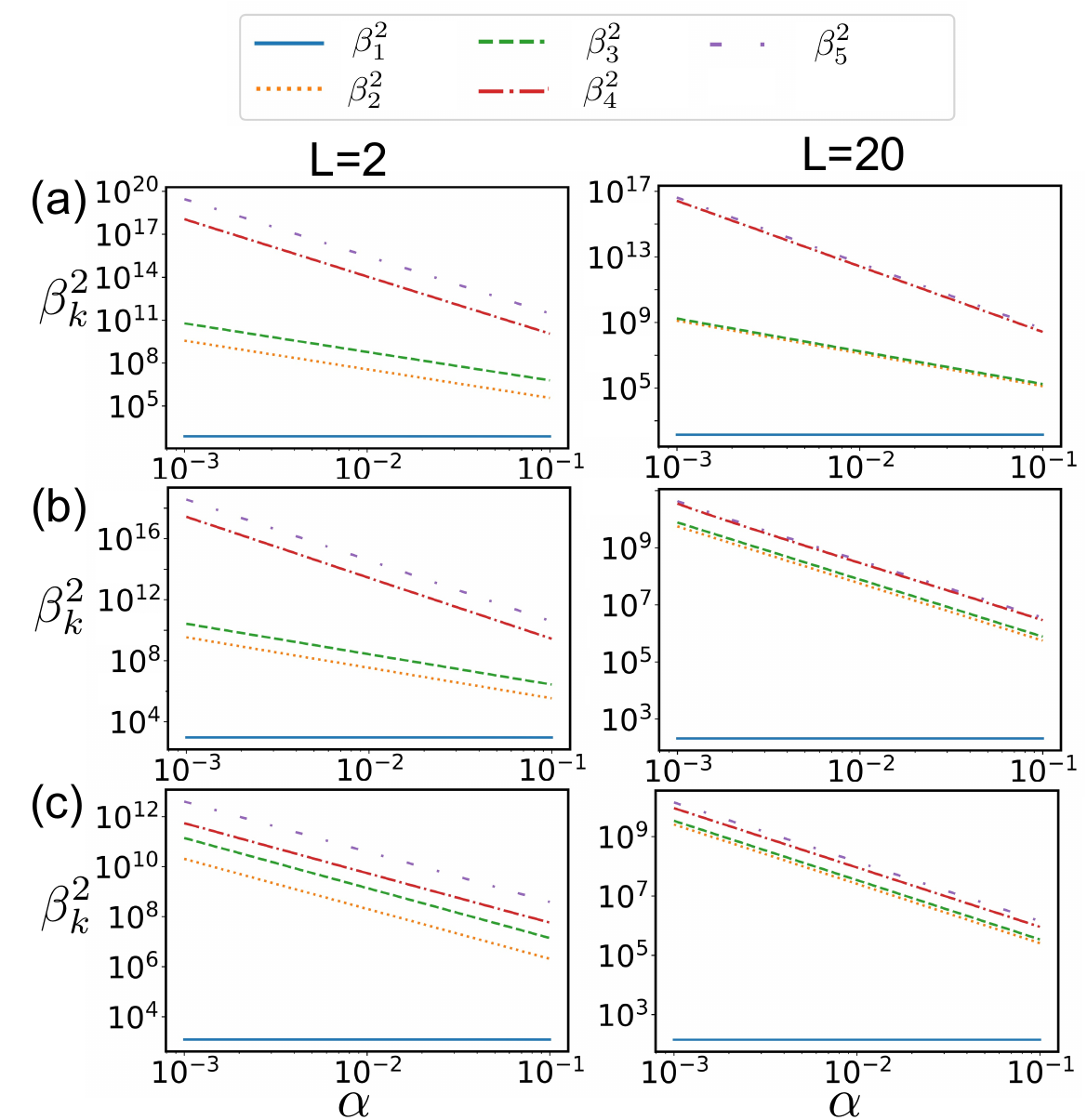}
\caption{Scaling of the $\beta_k^2$ as a function of $\alpha$ for imaging two compact sources with distance $L/2$. We consider three different cases: (a) direct imaging (b) separate SPADE method (c) orthogonalized SPADE method. Width of PSF $\sigma=1$.} 
\label{beta_new_method}
\end{center}
\end{figure}

\begin{figure}[!tb]
\begin{center}
\includegraphics[width=1\columnwidth]{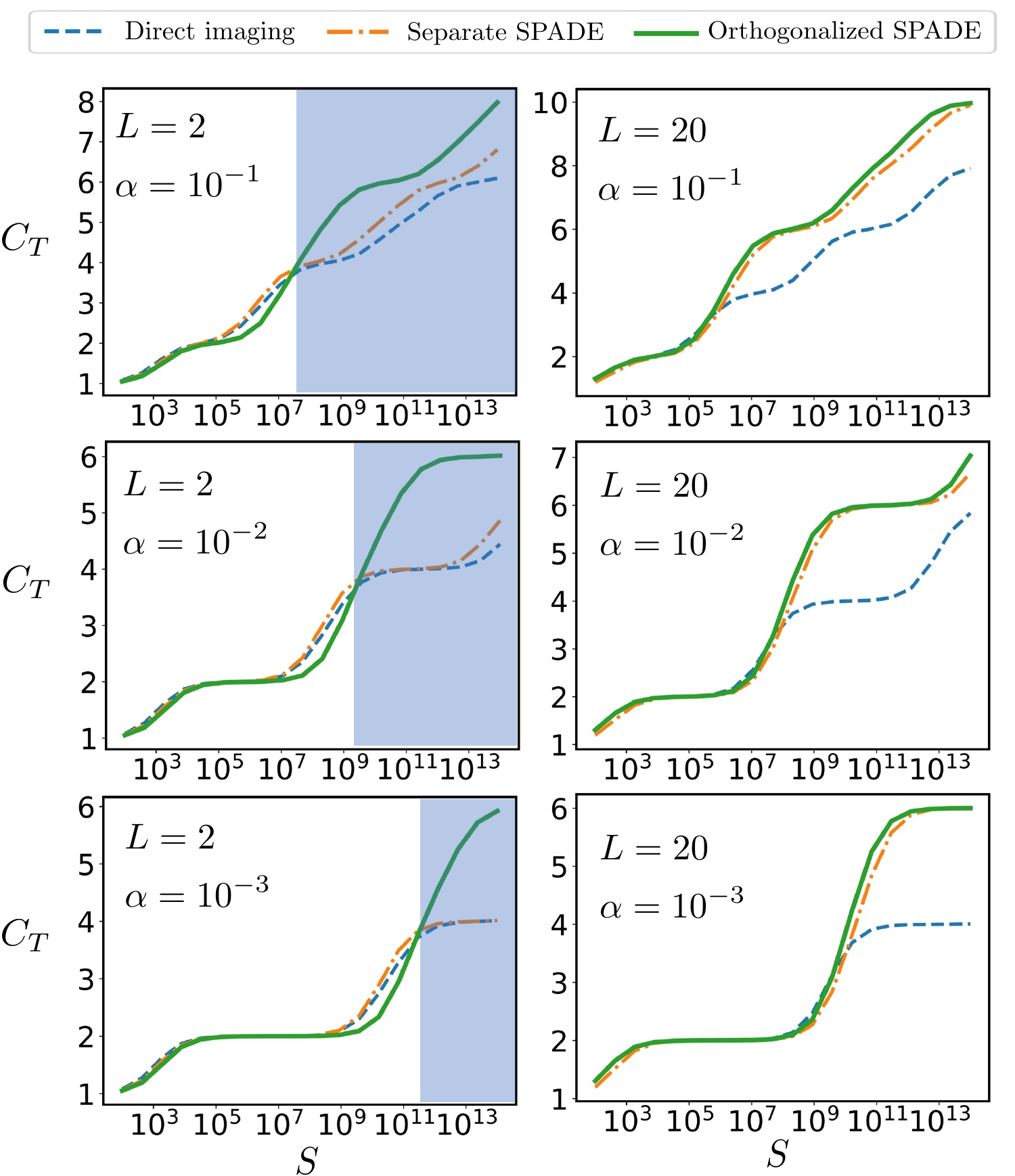}
\caption{Total REC $C_T$ for direct imaging, the separate SPADE method and the orthogonalized SPADE method as a function of $S$ for imaging of two compact sources with $\alpha=10^{-1},10^{-2},10^{-3}$, number of compact source $Q=2$, $\sigma=1$. The distance between the centroid of the two compact sources is $L/2$. When $L=2$, the orthogonalized SPADE method demonstrates a clear advantage over both the separate SPADE method and direct imaging in the shaded region. 
When $L=20$, the performance of SPADE and orthogonalized SPADE is comparable (essentially because they become equivalent when the two compact sources are sufficiently far apart) and both outperform direct imaging. Overall, the orthogonalized SPADE protocol demonstrates excellent performance, achieving a high $C_T$ (compared to the best from direct and SPADE protocols) for various choices of $L = 2$ and $L = 20$, as well as across a wide range of sample sizes $S$.}
\label{C_T_new_method}
\end{center}
\end{figure}

\begin{figure*}[!tb]
\begin{center}
\includegraphics[width=2\columnwidth]{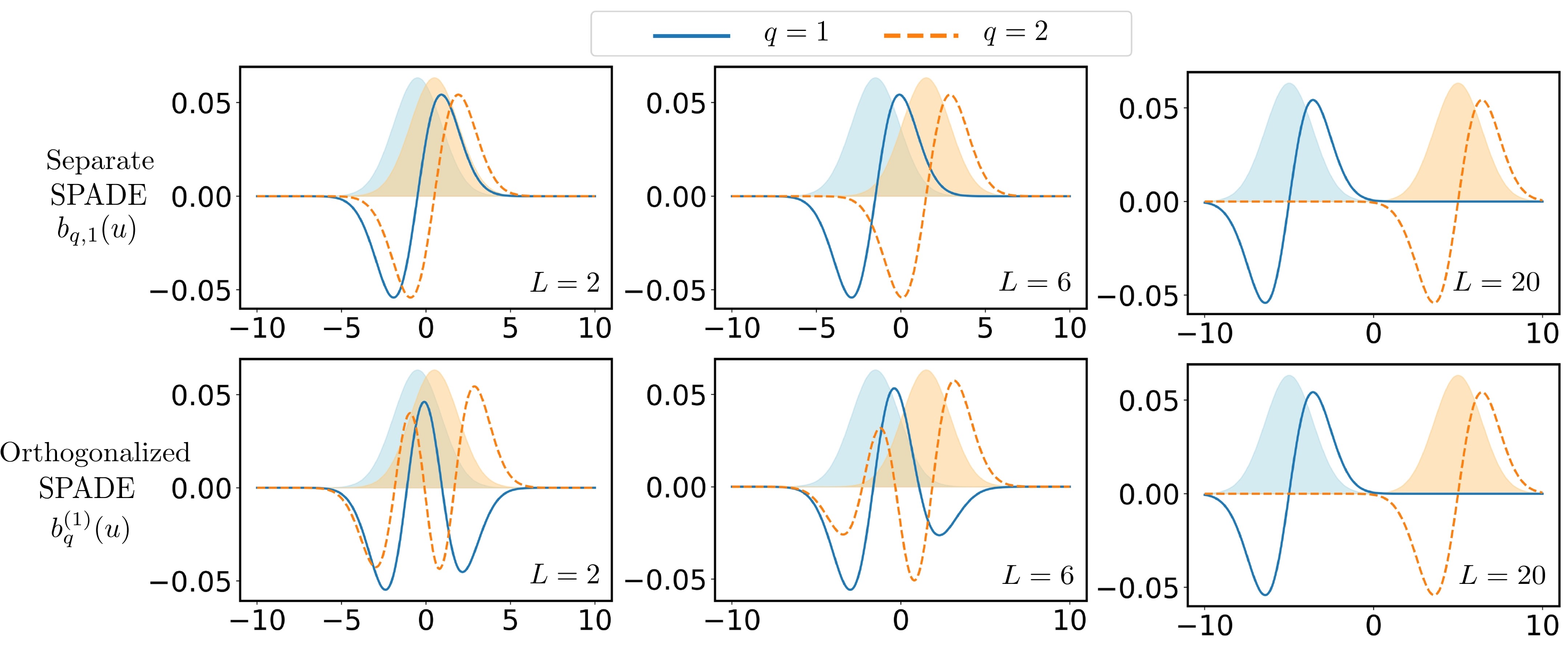}
\caption{We compare the basis states constructed in Eq.~\ref{eq:bjl} for the orthogonalized SPADE measurement, $\ket{b_{q}^{(l=1)}} = \int du \, b_{q}^{(1)}(u)\ket{u}$, with the corresponding basis states for the separate SPADE measurement, $\ket{b_{q,1}} = \int du \, b_{q,1}(u)\ket{u}$. 
\yk{Here, $q = 1, 2$ corresponds to the case of two compact sources ($Q = 2$) with centroids located at $\pm L/4$, so that the separation between the two sources is $L/2$. } We examine cases where $L = 2, 6, 20$. The shaded regions show two Gaussian PSFs of width $\sigma = 1$ located at the centers of two compact sources, that are used to construct the basis states $\ket{b_q^{(l)}}$ and $\ket{b_{q,m}}$.
}
\label{povm}
\end{center}
\end{figure*}

To improve imaging performance, one could apply the SPADE method to each compact source individually—a technique referred to here as the \emph{separate SPADE method}. Unfortunately, it only achieves the same scaling as direct imaging when the sources are not sufficiently spaced apart. This is because the proximity of other compact sources introduces significant noise when estimating higher-order moments. Alternatively, we can construct the orthonormal basis $\ket{b_j^{(l)}}$ using the Gram-Schmidt procedure, such that
\begin{equation}\label{eq:bjl}
\bra{\psi^{(m)}_k}\ket{b^{(l)}_j}\left\{
\begin{array}{ccc}
=0 & \quad   &  m\leq l-1\\
=0 &\quad   & m=l \,\,\& \,\, k\leq j-1\\
\neq 0  & \quad   &  \text{otherwise}
\end{array}\right.
\end{equation}
Choose POVM as the projection onto
\begin{equation}
\ket{\phi_{j\pm}^{(l)}}=\frac{1}{\sqrt{2}}\left(\ket{b_j^{(l)}}\pm\ket{b_j^{(l+1)}}\right),
\end{equation}
where $\quad j=1,2,3,\cdots,Q$, $l=0,1,2,\cdots,\infty$. 
The key intuition behind this construction is to ensure that when estimating the $x_{n,q}$ term in the $\Theta(\alpha^n)$ order, lower-order terms must vanish in the probability distribution, particularly those contributions from nearby compact sources. 
We refer to this new approach as the \emph{orthogonalized SPADE method}, as it projects onto a basis that is an orthogonalization of the separate SPADE method. This construction applies analogously to any PSF beyond Gaussian PSF. Note that for a single compact source, the separate SPADE and orthogonalized SPADE methods are identical, both referred to as the SPADE method. \yk{We find that for the orthogonalized SPADE method} 
\begin{gather}\label{eq:beta_superresolution_single}
\beta_{0}^2=0,\quad\beta_{1\leq i\leq Q-1}^2=\Theta(1),\\
\beta_{Q\leq i\leq 3Q-1}^2=\Theta(\alpha^{-2}),\,
\beta_{3Q\leq i\leq 5Q-1}^2=\Theta(\alpha^{-4}),\cdots\notag
\end{gather}
\yk{The scaling is numerically confirmed in Fig.~\ref{beta_new_method} for the first six eigenvalues, and we expect $\beta_k^2=\Theta(\alpha^{-2\lceil\lfloor k/Q\rfloor/2\rceil})$ to hold for eigenvalues with higher indices and any prior $p(\vec{x})$.  In the numerical calculation, we again consider two compact sources with centroids at $-L/4$ and $L/4$ and a random prior obtained by randomly generating a set of images. We examine the separate SPADE method in Fig.~\ref{beta_new_method}(b) and the orthogonalized SPADE method in Fig.~\ref{beta_new_method}(c). When the sources are well separated ($L \gg \sigma$), both methods achieve the expected scaling, with four $\beta_k^2$ terms scaling as $\Theta(\alpha^{-2})$, compared to two for direct imaging. The doubling of the number of eigenvalues scaling as $\Theta(\alpha^{-2n})$ for each $n$ aligns with expectations for two compact sources.
However, when the sources are closer ($L=2$, $\sigma=1$), the performance of the separate SPADE method is strongly degraded, reducing the scaling to that of direct imaging, with only two $\beta_k^2$ scaling as $\Theta(\alpha^{-2})$. In contrast, our orthogonalized SPADE method retains four eigenvalues $\beta_k^2$ with scaling $\Theta(\alpha^{-2})$.}

In Fig.~\ref{C_T_new_method}, we demonstrate the role of $\beta_k^2$ as the thresholds for stepwise increases in total REC $C_T$ for two compact sources. When $L=2$ (sources close together), direct imaging and separate SPADE show that $C_T$ increases by 2 at each step after the initial two $\Theta(1)$ eigenvalues. In contrast, for orthogonalized SPADE, $C_T$ is increased by 4 after the initial two $\Theta(1)$ eigenvalues, highlighting the advantage of our method for two close compact sources. 
When $L=20$ (with sources well separated), both separate SPADE and orthogonalized SPADE yield a $C_T$ increase of 4 at each step. Note that in certain regions of the sample number when $\alpha = 10^{-2}$, the orthogonalized SPADE method may perform slightly worse than the separate SPADE method. This difference arises from the different constant prefactors in $\beta_k^2$. 
To ensure optimal performance, we can adopt an adaptive approach: for a given sample size, we select either the orthogonalized or separate SPADE method based on the total REC of each method, choosing the one that offers superior performance.


Note that in the imaging of multiple compact sources, the eigentasks identified from the eigenvectors $r_k$ generally do not have the simple form seen in the single compact source case. The structure of the eigentasks is strongly influenced by the practical imaging model and the prior knowledge, such as the positions of the individual sources. We explicitly demonstrate this in Sec. \ref{SI:eigentask_multiple} of the Supplemental Material. This observation suggests that, in practical applications, the quantum learning approach offers nontrivial guidance on which features should be incorporated into the downstream analysis. 

In Fig.~\ref{povm}, we numerically illustrate the shape of the constructed basis  for the separate SPADE method and the orthogonalized SPADE method, \yk{considering different distances between the centroids of the two compact sources. Note that the construction of the basis states   defined in Eq.~\ref{eq:bjl} does not depend on the size of each source $L_q$.} It is evident that when the two compact sources are close to each other, the basis for the separate and orthogonalized SPADE methods differ significantly. However, when the two compact sources are sufficiently far from each other, the basis for the separate and orthogonalized SPADE methods become nearly identical. Given the complicated form of the basis, the spatial light modulator could serve as a practical tool for its implementation, as previously discussed in the context of superresolution \cite{ozer2022reconfigurable}. In Sec. \ref{SI:multiple_compact_source} of the Supplemental Material, we demonstrate that the Hermite-Gaussian mode sorter \cite{lavery2012refractive,beijersbergen1993astigmatic,ionicioiu2016sorting,zhou2017sorting,zhou2018hermite} can be used to implement the orthogonalized SPADE method with some additional steps for a Gaussian PSF, and we also provide more details on the derivation of the scaling of $\beta_k^2$ and the eigenvectors $r_{k}$ based on numerical calculations.


\begin{figure*}[!tb]
\begin{center}
\includegraphics[width=2\columnwidth]{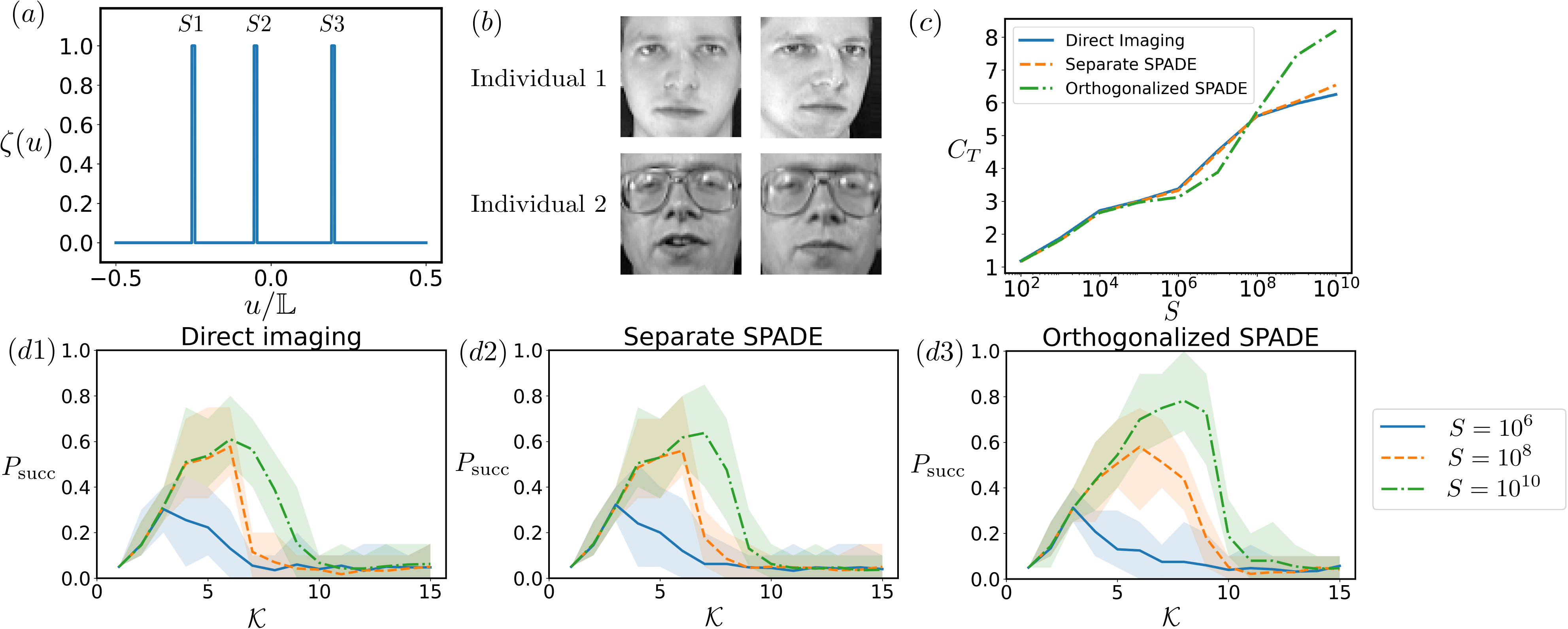}
\caption{\ykwang{Simulation for face recognition. (a) The position function $\zeta(u)$ specifies the positions of three compact sources, labeled $S1$, $S2$, and $S3$. All three compact sources are confined within the interval $[-\mathbb{L}/2, \mathbb{L}/2]$ with $\mathbb{L} = 10$, and each source has a size of $0.1$, which is smaller than the PSF width $\sigma = 1$. (b) Example images of two individuals from the Olivetti Faces dataset.
(c) The total REC $C_T$ as a function of the sample number $S$ for the three approaches.
(d1)–(d3) The success probability $P_{\text{succ}}$ as a function of the highest order $\mathcal{K}$ of the eigentasks included in the training and testing procedures for face recognition, evaluated at different sample sizes $S$. The lines represent the mean success probability, and the shaded region shows the maximum and minimum values across the 20 repetitions. We set $\alpha = 0.1$ in this figure.
} }
\label{fig:face_recognition_main}
\end{center}
\end{figure*}

\ykwang{\textit{Demonstrative example.---}We now present a demonstrative example to illustrate the advantage of the quantum learning approach and our orthogonalized SPADE method in a face-recognition imaging task.
The face images used in our simulation are taken from the Olivetti Faces dataset provided by AT\&T Laboratories Cambridge and distributed through \texttt{scikit-learn}, with examples shown in Fig.~\ref{fig:face_recognition_main}(b). Each individual has $64\times 64$ grayscale images with different facial expressions and small variations in pose. These images are converted into one-dimensional images by rasterization, a standard preprocessing technique in machine learning, and the resulting images are partitioned into $M=3$ segments that serve as compact sources placed over the interval $[-\mathbb{L}/2, \mathbb{L}/2]$ according to the position function $\zeta(u)$ in Fig.~\ref{fig:face_recognition_main}(a).
Our goal is to determine the identity of each given face image. In the following, we treat the imaging system itself as PNNs—the quantum learning devices—and use it to perform this face-recognition task. We emphasize that this approach can address a
wide range of imaging tasks beyond this one, such as
regression and clustering.

We use face images from $N_{\text{person}} = 20$ individuals, which form the training set whose statistics are used to compute the prior information. Using the REC formalism in Eq.~\ref{eq:REC_equations}, we then calculate the eigenvectors $r_k$.
For the $m$th image, the measurement yields a probability distribution $P_m(j)$, where $j$ labels the measurement outcome. The $k$th eigentask is obtained as a linear combination of these distributions with coefficients given by the $k$th eigenvector $r_k$, yielding $\xi_{km} = \sum_j r_{kj} P_m(j)$. This defines the eigentask vector for the $m$th image $\vec{\xi}_m = [\xi_{0m}, \xi_{1m}, \xi_{2m}, \dots, \xi_{\mathcal{K}m}]$, where the eigentasks are truncated at order $\mathcal{K}$.
The truncation is introduced because higher-order eigentasks are noisy and poorly estimated, and including them in downstream analysis, namely training and inference, can degrade the performance.
The shapes of the eigenvectors are generally complex, reflecting the complexity of the imaging task, which is common in practical applications. PNNs perform classification by training only the output weights using logistic regression. 
For each individual, we obtain eigentask vectors $\vec{\xi}_m$, forming datasets that serve as the training inputs for the logistic regression classifier.
Before training, each component of $\vec{\xi}_m$ is normalized by dividing by its mean absolute value across the training set to ensure balanced feature scaling.
During training, we use logistic regression from the standard Python package \texttt{scikit-learn} to fit a  multi-class classifier to the labeled eigentask vectors, modeling the probabilities $P(y \mid \vec{\xi})$ with a multinomial logistic (softmax) model.
During testing, we use the new face images for each individual. 
For each image, we fix the total number of detected photons (sample number) to be $S$. By counting the number of getting each outcome, we obtain an empirical  eigentask vector $\hat{\vec{\xi}}_m = [\hat{\xi}_{0m}, \hat{\xi}_{1m}, \hat{\xi}_{2m}, \dots, \hat{\xi}_{\mathcal{K}m}]$ for the $m$th face image. These empirical eigentask vectors are then used for inference with the trained model.

In the context of imaging multiple compact sources, we show the performance of face recognition  using the three approaches in Fig.~\ref{fig:face_recognition_main}(d1)-(d3).  We observe that the success probability $P_{\text{succ}}$ first increases with $\mathcal{K}$ and then decreases. This behavior is intuitive: increasing $\mathcal{K}$ captures more information and improves classification, but beyond a point, higher-order tasks become noisy due to limited sample size $S$, degrading performance and reducing success probability.
Comparing the success probability $P_{\text{succ}}$ of the three approaches, since the source distance is comparable to the PSF width $\sigma$, separate SPADE and direct imaging perform similarly, whereas our orthogonalized SPADE achieves higher performance, with its peak $P_{\text{succ}}$ exceeding the others at $S=10^{10}$.

This simulation highlights the operational meaning of the total REC $C_T$. First, it estimates the number of eigentasks that can be reliably included in the downstream logistic regression. Since $\beta_k^2$ sets the sample threshold for estimating the $k$th eigentask and $C_T = \sum_k 1/(1+\beta_k^2/S)$, with each term approaching 1 when $S \gg \beta_k^2$, $C_T$ roughly counts the reliably estimated eigentasks. At $S=10^6$ and $S=10^8$, $C_T$ is similar across all approaches, matching the $P_{\text{succ}}$ behavior. At $S=10^{10}$, $C_T$ for orthogonalized SPADE rises to about 8, aligning with its peak $P_{\text{succ}}$, while direct imaging and separate SPADE reach about 6, consistent with their peaks. Second, a larger $C_T$ allows more well-estimated eigentasks to be included, capturing more information and thus improving the success probability in this example; therefore, a larger $C_T$ indicates better imaging performance. 
Further details of this simulation, as well as additional examples beyond face recognition, are provided in Sec.~\ref{SI:demonstrative_example} of the Supplemental Material.

In the broader context, we here present an example that directly addresses the imaging task using PNNs.  PNNs employ analog systems with a fixed internal structure—here, the imaging system—and performs practical tasks by training only the output weights, which in this example corresponds to logistic regression, though alternatives like linear regression can be used depending on the task. In this sense, the PNNs framework provides a systematic approach to modeling and solving complex imaging problems, which is especially valuable given the infinite degrees of freedom inherent to imaging. The REC formalism, developed for the PNNs paradigm, further guides the training step by identifying eigentasks and selecting those that are well-estimated and low-noise. This formalism has also proven highly effective in superresolution problems involving complex source structures, as demonstrated in our simulation.

Many machine-learning tasks, including face recognition, can be viewed as discrimination problems that are in principle solvable by the likelihood-ratio method and whose performance is bounded by the Chernoff bound. While simple discrimination settings have been analyzed using the Chernoff bound in superresolution \cite{zanforlin2022optical,huang2021quantum,zhang2020super,lu2018quantum,grace2022identifying}, the likelihood-ratio method requires an accurate statistical model of the objects being imaged, whereas learning methods can handle far more complex structures. In our face-recognition example involving tens of individuals, it is infeasible to write down a closed-form likelihood function for an exact likelihood-ratio calculation.
Moreover, the Chernoff bound is asymptotic and may appear to be reached even when the success probability is already above $99\%$ in some cases, offering little guidance on the practically relevant regime where we care about reaching moderate performance levels such as $70\%$. It also does not suggest which measurement strategy should be used when the likelihood-ratio method is inapplicable.
By contrast, our learning-based approach yields the total REC as a meaningful figure of merit and, crucially, provides a principled strategy for tackling discrimination tasks  in the finite-sample regime. Simulated examples illustrating the above discussion are provided in Sec.~\ref{SI:Chernoff} of the Supplemental Material.
}

\textit{Imaging general sources beyond the Rayleigh limit.---}
An intriguing question is whether, when a general source cannot be split into multiple compact sources, we can still split the source into $Q$ small pieces and apply our orthogonalized SPADE method, improving the imaging performance. Unfortunately, there is a constraint.  Each source generates a set of states $\{\ket{\psi_q^{(m)}}\}_{m=0,1,2,\dots}$, which are used to construct $\ket{b_j^{(l)}}$ via the Gram-Schmidt procedure. When $Q \gg L/\sigma$, the differences between $\ket{\psi_q^{(m)}}$ for different $q$ can be vanishingly small, leading to the potentially suboptimal performance of orthogonalized SPADE method due to large prefactors of the eigenvalues $\beta_k^2$. 
In Sec.~\ref{SI:new_method_superresolution} of the Supplemental, we demonstrate that when $Q$ exceeds roughly $L/\sigma$, the stepwise increase in $C_T$ is smoothed out, and numerically, we find that both the separate SPADE and the orthogonalized SPADE no longer offer advantages over direct imaging.  
In conclusion,  the superresolution methods from Ref.~\cite{tsang2016quantum,zhou2019modern} can be properly generalized to resolve multiple compact sources but may fail for generally distributed sources. \yk{We also present a discussion of imaging such a general source using direct imaging, showing that $C_T$ is approximately related to the ratio between the source size and the PSF width, as detailed in Sec. \ref{SI:general_source} of the Supplemental Material.}

\textit{Conclusion.---}In this work, we treat the imaging systems in superresolution as PNNs, i.e., quantum learning devices, thereby  providing a systematic framework for addressing practical imaging tasks with complex structures and operating in the finite-sample regime. Based on the REC formalism, the measurable features of the source and their corresponding sample thresholds can be identified with eigentasks, while the total REC serves as a principled metric for selecting relevant features for downstream analysis and quantifying the performance of an imaging method. We further extend the superresolution framework to handle multiple compact sources and propose the orthogonalized SPADE method—a nontrivial generalization that relaxes the strong assumptions of earlier superresolution studies, thereby improving practical applicability. We show that superresolution exhibits a stepwise increase in total REC, with thresholds determined by the ratio of source size to PSF width. \ykwang{The advantages of this quantum learning approach and our orthogonalized SPADE method are demonstrated through total REC calculations and concrete examples, including face recognition.}

It would also be worthwhile to investigate other potentially advantageous imaging protocols, e.g., entangled measurements on multiple copies of photon states, which exhibit advantages over separable measurements in tomography~\cite{o2016efficient, haah2016sample}. 
We may also explore the potential applications of quantum computing schemes \cite{metger2024pseudorandom,ma2024construct,schuster2024random}, where the quantum advantage inherent in these schemes could benefit specific imaging tasks. \yk{Note that the total REC depends on the POVM used in the detection. It may be interesting to explore whether a closed-form expression of the total REC optimized over all possible POVMs can be found, at least in some special cases—for example, when the measurable features require compatible measurements.}

\textit{Acknowledgements.---}We would like to thank Hakan E. Tureci for helpful discussion.  Y.W. and S.Z. acknowledges funding provided by Perimeter Institute for Theoretical Physics, a research institute supported in part by the Government of Canada through the Department of Innovation, Science and Economic Development Canada and by the Province of Ontario through the Ministry of Colleges and Universities. Y.W. also acknowledges funding from the Canada First Research Excellence Fund. L.J. acknowledges support from the the ARO(W911NF-23-1-0077), ARO MURI (W911NF-21-1-0325), AFOSR MURI (FA9550-19-1-0399, FA9550-21-1-0209, FA9550-23-1-0338), DARPA (HR0011-24-9-0359, HR0011-24-9-0361), NSF (OMA-1936118, ERC-1941583, OMA-2137642, OSI-2326767, CCF-2312755), NTT Research, Packard Foundation (2020-71479), and the Marshall and Arlene Bennett Family Research Program. J.L. acknowledges startup funds provided by the Department of Computer Science, School of Computing and Information, the University of Pittsburgh, Pittsburgh Quantum Institute (PQI), funding from IBM Quantum through the Chicago Quantum Exchange, and AFOSR MURI (FA9550-21-1-0209). C.O. acknowledges support from Quantum Technology R\&D Leading Program~(Quantum Computing) (RS-2024-00431768) through the National Research Foundation of Korea~(NRF) funded by the Korean government (Ministry of Science and ICT~(MSIT))

\textit{Author contributions} - Y.W. carried out the analytical calculation and the numerical simulation. L.J. conceived the project. S.Z., J.L., and L.J. supervised the project. All authors contributed to the development of ideas and the writing of the manuscript.

\textit{Competing interests} - The authors declare no competing interests.

\textit{Data availability} - No data sets were
generated or analyzed during the current study.

\textit{Code availability} - All codes used in this paper have been deposited in GitHub at 
https://github.com/ykwang-phys/quantum-learning-imaging

\textit{Correspondence} and requests for materials should be addressed to
Yunkai Wang, Sisi Zhou or Liang Jiang.


\bibliography{arxiv}

\appendix

\onecolumngrid

\newpage

\begin{longtable}{p{0.27\linewidth} p{0.69\linewidth}}
\caption{\yk{Summary of notations.}}\\
\toprule
\textbf{Notation} & \textbf{Definition / Meaning} \\
\midrule
\endfirsthead
\toprule
\textbf{Notation} & \textbf{Definition / Meaning} \\
\midrule
\endhead
\bottomrule
\endfoot

\multicolumn{2}{l}{\textit{General setup and probabilities}}\\
$\boldsymbol{\theta}$ & Input parameters of the imaging task (e.g., separation $L$, or moment vector $\vec{x}$). \\
$p(\boldsymbol{\theta})$ & Prior distribution over parameters $\boldsymbol{\theta}$. \\
$\rho(\boldsymbol{\theta})$ & Quantum state on the image plane encoding $\boldsymbol{\theta}$. \\
$M_j$ & $j$th POVM element (measurement operator). \\
$P_j(\boldsymbol{\theta}) = \mathrm{Tr}[\rho(\boldsymbol{\theta}) M_j]$ & Probability of measurement outcome $j$. \\
$S$ & Number of detected samples/photons. \\

\addlinespace
\multicolumn{2}{l}{\textit{REC (Resolvable Expressive Capacity) formalism}}\\
$f(\boldsymbol{\theta})$ & Target function of $\boldsymbol{\theta}$. \\
$C[f]$ & REC of $f(\boldsymbol{\theta})$, measuring accuracy of approximating $f$ with measured features. \\
$C_T = \sum_k \frac{1}{1+\beta_k^2/S}$ & Total REC, effective number of resolvable features at sample size $S$. \\
$f_k(\boldsymbol{\theta}) = \sum_m r_{km} P_m(\boldsymbol{\theta})$ & $k$th eigentask. \\
$\vec{r}_k=[r_{k0},r_{k1},r_{k2},\cdots]$ & $k$th eigenvector defining $k$th eigentask $f_k$. \\
$\beta_k^2$ & Eigenvalue from REC eigenproblem, sets sample threshold for $f_k$. \\
$\hat{\rho}^{(t)} = \mathbb{E}_\theta[\rho(\boldsymbol{\theta})^{\otimes t}]$ & Prior-averaged $t$-copy state. \\
$D_{kj} = \delta_{kj}\,\mathrm{Tr}[M_k \hat{\rho}^{(1)}]$ & Diagonal expectation matrix. \\
$G_{jk} = \mathrm{Tr}[(M_j\!\otimes\! M_k)\, \hat{\rho}^{(2)}]$ & Covariance matrix of outcomes. \\
$V = D - G$ & Matrix defining REC eigenproblem. \\

\addlinespace
\multicolumn{2}{l}{\textit{Imaging model and PSF}}\\
$|u\rangle = a_u^\dagger |0\rangle$ & Single-photon position eigenstate at $u$. \\
$\psi(u)$ & Point spread function (PSF). Gaussian example: $\psi(u)=\exp(-u^2/4\sigma^2)/(2\pi\sigma^2)^{1/4}$. \\
$\sigma$ & Width of the PSF. \\
$I(u)$ & Normalized source intensity distribution. \\

\addlinespace
\multicolumn{2}{l}{\textit{Two incoherent point sources}}\\
$L$ & Separation between two point sources. \\
$u_1=-L/2,\ u_2=+L/2$ & Coordinates of the sources. \\
$\rho(L)=\tfrac12(|\psi_1\rangle\langle\psi_1|+|\psi_2\rangle\langle\psi_2|)$ & State of two incoherent points. \\
$p(L) = \tfrac{1}{\sqrt{2\pi}\gamma} e^{-L^2/2\gamma^2}$ & Gaussian prior for separation $L$. \\
$\alpha = \gamma/\sigma$ & Small parameter in two-point source analysis. \\
$E_x = |x\rangle\langle x|$ & Projector onto position $x$. \\
$M_0 = |\phi_0\rangle\langle \phi_0|,\ M_1 = I-M_0$ & Binary SPADE POVM. \\
$\phi_0(u) = (2\pi\sigma^2)^{-1/4} e^{-u^2/4\sigma^2}$ & Fundamental Gaussian mode. \\

\addlinespace
\multicolumn{2}{l}{\textit{Single compact source (moments)}}\\
$L$ & Source size. \\
$\alpha = L/\sigma$ & Small parameter for compact source. \\
$x_n = \int du\, I(u)\left(\tfrac{u-u_0}{L}\right)^n$ & $n$th normalized moment. \\
$\vec{x}=[x_0,x_1,\dots]$ & Moment vector. \\
$\rho = \sum_{k=0}^\infty x_k \rho^{(k)}$ & Expansion of state in PSF derivatives. \\
$|\psi^{(m)}\rangle$ & Basis vector defined by the $m$th derivative of the PSF. \\
$|b_{l}\rangle$ & Basis obtained via Gram–Schmidt orthogonalization of the set $\{|\psi^{(m)}\rangle\}_{m=0,1,2,\cdots}$. \\
$|\phi_{n,\pm}\rangle = (|b_n\rangle \pm |b_{n+1}\rangle)/\sqrt{2}$ & Basis vectors of the projective measurement used in SPADE.  \\
$P_{n,\pm}$ & Outcome probabilities. \\
$c_{nm},\, c^{\pm}_{nm}$ & Expansion coefficients in probabilities. \\

\addlinespace
\multicolumn{2}{l}{\textit{Multiple compact sources}}\\
$Q$ & Number of compact sources. \\
$I_q(u)$ & Intensity of the $q$th source. \\
$u_q$ & Centroid of the $q$th source. \\
$L_q$ & Size of the $q$th source. \\
$x_{n,q} = \int du\, I_q(u)\left(\tfrac{u-u_q}{L_q}\right)^n$ & $n$th moment of the $q$th source. \\
$|\psi_q^{(m)}\rangle$ & Basis vector defined by the $m$th derivative of the PSF centered at $u_q$. \\
$\ket{b_{q,m}}$ & Basis obtained via Gram-Schmidt orthogonalization of set $\{\ket{\psi_q^{(m)}}\}_{m=0,1,2,\cdots,\infty}$\\
$\ket{\phi_{q,m,\pm}}=\frac{1}{\sqrt{2}}(\ket{b_{q,m}}\pm\ket{b_{q,m+1}})$ & Basis vectors of the projective measurement used in separate SPADE method.\\
$|b^{(l)}_j\rangle$ & Basis obtained via Gram–Schmidt orthogonalization of the set $\{|\psi^{(m)}_q\rangle\}_{m=0,1,2,\cdots,\,q=1,2,\cdots,Q}$.\\
$\ket{\phi_{j\pm}^{(l)}}=\frac{1}{\sqrt{2}}\left(\ket{b_j^{(l)}}\pm\ket{b_j^{(l+1)}}\right)$ & Basis vectors of the projective measurement used in orthogonalized SPADE method.\\

\addlinespace
\multicolumn{2}{l}{\textit{Demonstrative examples}}\\
$\eta(u)$ & Indicator function used to randomly generate intensity distributions. \\
$\zeta(u)$ & The function that specifies the positions of  compact sources.\\
$d$ & Region size containing all compact sources. \\
$\alpha=0.1,\ \sigma=1$ & Example numerical parameters used in demonstrations. \\
$P_m(j)$ & Probability of the $j$th measurement outcome for the $m$th intensity distribution. \\
$\vec{\xi}_m=[\xi_{0m},\xi_{1m},\dots,\xi_{\mathcal{K}m}]$ & Eigentask vector for the $m$th distribution, truncated at order $\mathcal{K}$. \\
$\vec{r}_k=[r_{k0},r_{k1},r_{k2},\cdots]$ & $k$th eigenvector defining $k$th eigentask. \\
$\xi_{km}=\sum_j r_{kj}P_m(j)$ & $k$th eigentask for the $m$th distribution. \\
$\mathcal{K}$ & Truncation order of the eigentask vector. \\

\addlinespace
\multicolumn{2}{l}{\textit{Auxiliary / supplemental symbols}}\\
$F$ & Fisher information matrix. \\
$\lambda_k$ & Generalized eigenvalue, related to $\beta_k^2$. \\
$W = D^{-1/2} G D^{-1/2}$ & Symmetrized matrix used in perturbation analysis. \\
$C_n$ & Coefficient vectors in $\alpha$-expansions. \\
$d_k$,\quad $g_{ij}$ & Prior-dependent quantities: $d_k=\int d\vec{x}\,p(\vec{x})\,x_k$, $g_{ij}=\int d\vec{x}\,p(\vec{x})\,x_i x_j$. \\
$z_n$ & Orthonormal basis vectors introduced for convenience in perturbative analysis.  \\
$\Pi_m$ & Projectors for block decomposition in perturbation theory. \\
$\vec{y}_k$ & Auxiliary vectors appearing as eigenvectors of an equivalent  form of the REC eigenproblem. \\
 \\

\end{longtable}

\section{Preliminary}
\subsection{Review of quantum-metrology-inspired superresolution}
\label{SI:preliminary_superresolution}
The Fisher information matrix $F$ provides a lower bound on the variance $\Sigma_{\vec{\theta}}$ associated with estimating unknown parameters $\vec{\theta} = [\theta_1, \theta_2, \cdots, \theta_n]$, such that $\Sigma_{\vec{\theta}} \geq F^{-1}$ \cite{braunstein1994statistical,paris2009quantum}. Quantum-metrology-inspired superresolution leverages the insight that a more carefully designed measurement can yield significantly higher Fisher information compared to direct imaging, when the size of the source is much smaller than the Rayleigh limit of the imaging system. Ref.~\cite{tsang2016quantum} highlighted that, when estimating the separation between two point sources, the Fisher information from direct imaging approaches zero as the separation between the sources decreases to zero, which indicates the Rayleigh't limit. However, the Fisher information for a spatial mode demultiplexing (SPADE) measurement onto Hermite-Gaussian modes remains constant, suggesting that it is possible to circumvent Rayleigh's criterion by carefully designing the measurement strategy.


Later on, this discussion is extended to consider a general source within the Rayleigh limit. We now review the discussion in Ref. \cite{zhou2019modern} using our notation in more detail for later use. A general incoherent source can be modeled as
\begin{equation}\label{eq:rho_basic}
\rho=\int du du_1 du_2 I(u)\psi(u-u_1)\ket{u_1}\bra{u_2}\psi(u-u_2),
\end{equation}
where $\ket{u}=a^\dagger_u\ket{0}$ is the single photon state at position $u$, $\psi(u)$ is the PSF whose width is quantified as $\sigma$, such as $\psi(u)={\exp(-u^2/4\sigma^2)}/{(2\pi\sigma^2)^{1/4}}$.  \yk{Assume the normalized source intensity $I(u)$ is confined within the interval $[-L/2, L/2]$}, define $\alpha=L/\sigma\ll 1$, we expand the PSF 
\begin{equation}
\psi(u_1-u)=\sum_{n=0}^\infty \left(\frac{\partial^n \psi(u_1-u)}{\partial u^n}\bigg|_{u=u_0}\frac{L^n}{n!}\right)\left(\frac{u-u_0}{L}\right)^n=\sum_{n=0}^\infty \psi^{(n)}(u_1)\left(\frac{u-u_0}{L}\right)^n.
\end{equation}
We can then get
\begin{equation}\begin{aligned}
\rho&=\sum_{m,n=0}^\infty\left(\int du I(u)\left(\frac{u-u_0}{L}\right)^{m+n}\right)\left(\int du_1\psi^{(m)}(u_1)\ket{u_1}\right)\left(\int du_2\psi^{(m)}(u_2)\bra{u_2}\right)\\
&=\sum_{m,n=0}^\infty x_{m+n}\ket{\psi^{(m)}}\bra{\psi^{(n)}}=\sum_{k=0}^\infty x_k\rho^{(k)},
\end{aligned}\end{equation}
where $\rho^{(k)}=\sum_{m+n=k}\ket{\psi^{(m)}}\bra{\psi^{(n)}}$. We aim to construct a measurement such that, when estimating higher-order moments, the lower-order terms of $\alpha$, which introduce significant noise, can be canceled. This can be achieved by constructing an orthonormal measurement basis ${\ket{b_l}}_l$ through the Gram-Schmidt procedure such that
\begin{equation}
a_{ml}=\bra{\psi^{(m)}}\ket{b_l}\left\{
\begin{array}{cc}
=0    &  m\leq l-1\\
\neq 0     &  m\geq l
\end{array}\right.
\end{equation}
The measurement is then constructed as 
\begin{equation}\label{ref5_POVM}
\left\{\frac{1}{2}\ket{\phi_{i,\pm}}\bra{\phi_{i,\pm}},\frac{1}{2}\ket{b_0}\bra{b_0}\right\}, 
\end{equation}
\begin{equation}
\ket{\phi_{i,\pm}}=(\ket{b_i}\pm\ket{b_{i+1}})/\sqrt{2},\quad i=0,1,2,\cdots
\end{equation}
The probability of projecting onto state $\ket{\phi_{n,\pm}}$ is 
\begin{equation}\begin{aligned}\label{eq:P_superresolution2}
&P_0=\frac{1}{2}\sum_{m=0}^\infty x_m \rho_{0,m,0}=\sum_{m=0}^\infty c_m x_m\alpha^m\\
&P_{n,\pm}=\frac{1}{4}\sum_{m=2n}^\infty  x_m \rho_{n,m,n}+\frac{1}{4}\sum_{m=2n+2}^\infty x_m \rho_{(n+1),m,(n+1)}\pm \frac{1}{4}\sum_{m=2n+1}^\infty  x_m (\rho_{n,m,(n+1)}+\rho_{(n+1),m,n})\\
&\quad\quad=\sum_{m\geq 2n}c_{nm}^{\pm}x_m\alpha^m,\quad n=0,1,2,3,\cdots\\
&\rho_{n,m,l}=\bra{b_n}\rho^{(m)}\ket{b_l}\sim \Theta(\alpha^m),
\end{aligned}\end{equation}
where we can find each coefficients $c_{nm}$. 

\yk{As suggested above, the key feature in the construction of the superresolution method is that we can have the probability distribution $P_{n,\pm}$ starting from the $\Theta(\alpha^{2n})$ order. This is crucial because, when estimating higher-order moments, these moments appear as the coefficients of higher-order terms in the probability distribution. For example, when we want to estimate the moment $x_2$, this information is encoded in $P_0=c_0x_0+c_1x_1\alpha+c_2x_2\alpha^2+\cdots$, $P_{0,\pm}=c_{00}^{\pm}x_0+c_{01}^{\pm}x_1\alpha+c_{02}^{\pm}x_2\alpha^2+\cdots$ as well, but these contain lower-order terms of $\alpha$ such as $c_0x_0$, $c_1x_1\alpha$, etc.
In contrast, the probability distribution $P_{1,\pm}=c_{12}^{\pm}x_2\alpha^2+c_{13}^{\pm}x_3\alpha^3+\cdots$ has $x_2$ in the leading-order term. Intuitively, the lower-order terms of $\alpha$ serve as noise in estimating higher-order moments, while the term containing $x_2$ constitutes the signal. As $\alpha \rightarrow 0$, the signal-to-noise ratio in $P_0$ and $P_{0,\pm}$ becomes very poor due to the dominance of these noise terms. In contrast, in $P_{1,\pm}$, where lower-order terms are absent, the signal-to-noise ratio is significantly improved. }

This improvement is explicitly quantified by the Fisher information.
The Fisher information  contributed by the basis $\ket{\phi_{n,\pm}}$ is
\begin{equation}
F_{x_ix_j}^{(n)}=\left\{
\begin{aligned}
&\Theta(\alpha^{i+j-2n})&, \quad& i,j \geq 2n \\
&0&, \quad & i<2n \quad\text{or} \quad j<2n
\end{aligned}
\right.
\end{equation}
where $i,j,n=0,1,2,\cdots$.  Adding the contribution of all basis $\ket{\phi_{n,\pm}}$, the total Fisher information $F=\sum_n F^{(n)}$ is given by
\begin{equation}\label{Fij_suboptimal}
F_{x_ix_j}=\Theta(\alpha^{i+j-2\lfloor \min\{i,j\}/2 \rfloor}).
\end{equation}
In particular, the diagonal elements $F_{x_ix_i}$ scale as $\Theta(\alpha^{0}), \Theta(\alpha^{2}), \Theta(\alpha^{2}), \Theta(\alpha^{4}), \Theta(\alpha^{4}), \cdots$ for $i = 0, 1, 2, 3, 4, \cdots$. Note that the scaling of the Fisher information in the above equation differs from that in Ref.\cite{zhou2019modern} by a factor of $\alpha^2$. This discrepancy arises because we define the moment as $x_k = \int du  I(u)  (u-u_0)^k/L^k$, while in Ref.\cite{zhou2019modern}, the moment was defined as $M_k=(\int du I(u)(u-u_0)^k)^{1/k}$. Here, $x_k$ is independent of $\alpha$, whereas $M_k \propto \alpha$, leading to the factor of $\alpha^2$ difference in the Fisher information.


For a direct imaging approach, we simply project onto the state $\ket{u}\bra{u}$ to estimate the intensity at each position $u$ on the image plane, with the probability given by
\begin{equation}\begin{aligned}\label{P_direct}
&P_{n}=\sum_{m=0}^\infty c_{nm}x_m\alpha^m,\quad n=0,1,2,3,\cdots
\end{aligned}\end{equation}
Note that in practical detection scenarios, it is not feasible to estimate each spatial mode individually as $\ket{u}\bra{u}$. Instead, we measure discretized pixels, represented as $M_n = \int_{X_n - l/2}^{X_n + l/2} dx \ket{x}\bra{x}$, where $l$ is the size of each pixel in the direct imaging, $X_n$ is the centroid of each pixel. Consequently, we will use a discretized version of direct imaging for both analytical and numerical discussions.
And we can calculate the FI of estimating $x_n$ as
\begin{equation}
F_{x_ix_j}=\Theta(\alpha^{i+j}),
\end{equation}
which is much lower compared to the measurement constructed in Eq. \ref{Fij_suboptimal}. Specifically, the diagonal elements scale as $F_{x_ix_i} = \Theta(\alpha^{0}), \Theta(\alpha^{2}), \Theta(\alpha^{4}), \Theta(\alpha^{6}), \Theta(\alpha^{8}), \cdots$ for $i = 0, 1, 2, 3, 4, \cdots$. This indicates that sensitivity can be improved by carefully designing the measurement strategy. 

The measurement scheme from Ref.~\cite{zhou2019modern} works universally for any PSF. For the Gaussian PSF $\psi(u) = \exp(-u^2/4\sigma^2)/(2\pi\sigma^2)^{1/4}$, the scheme from Ref.~\cite{zhou2019modern} is the projective measurement  $\left\{\frac{1}{4}(\ket{b_i}\pm\ket{b_{i+1}})(\bra{b_i}\pm\bra{b_{i+1}}),\frac{1}{2}\ket{b_0}\bra{b_0}\right\}_{i=0,1,2,\cdots}$ as in Eq.~\ref{ref5_POVM}, where $\ket{b_q} = \int du \, \phi_q(u) \ket{u}$ and $\phi_q(u) = (1 / 2 \pi \sigma^2)^{1/4} (1 / \sqrt{2^q q!}) H_q(\frac{x}{\sqrt{2}\sigma}) \exp(-\frac{u^2}{4 \sigma^2})$, with $H_q$ as the Hermite polynomial. 
Note that the SPADE measurement first used in Ref.~\cite{tsang2016quantum} is the projective measurement  $\{\ket{b_q}\bra{b_q}\}_{q=0,1,2,3,\cdots}$ when PSF is Gaussian, the measurement scheme from Ref.~\cite{zhou2019modern} is the projective measurement 
$\left\{\frac{1}{4}(\ket{b_i}\pm\ket{b_{i+1}})(\bra{b_i}\pm\bra{b_{i+1}}),\frac{1}{2}\ket{b_0}\bra{b_0}\right\}_{i=0,1,2,\cdots}$. 
For simplicity, however, we still refer to the measurement from Ref.~\cite{zhou2019modern} for any PSF as the SPADE method  throughout the discussion. However, all the analytical results presented in this work apply to any PSF and to the measurement approach from Ref.~\cite{zhou2019modern}.

\yk{
We briefly summarize related theoretical works for comparison with our results.
Sorelli et al. \cite{sorelli2021optimal} propose a useful estimator for the superresolution problem. However, their work focuses on a simple scenario—two point sources—and remains within the conventional parameter-estimation framework based on Fisher information, targeting a single well-defined parameter: the separation distance.
In contrast, we address the imaging of multiple general compact sources with more complex structures using a quantum learning approach. Rather than estimating specific parameters like moments, our method identifies which features can be reliably estimated, given prior knowledge and the measurement model in the finite-sample regime. The resulting eigentasks are problem-dependent. Our framework can tackle complex tasks by informing downstream analysis strategies, making it more versatile and practically relevant than the approach of Sorelli et al.
Grace et al. \cite{grace2020approaching} focus on estimating the separation and centroid of two point sources of equal strength, framed as a simple parameter estimation problem. In contrast, our approach adopts a quantum learning framework and addresses more general imaging scenarios, where the source can have structures far more complex than just two point sources.
}


\yk{
\subsection{Physical neural networks}\label{SI:PNNs}

In parametric quantum circuit learning \cite{benedetti2019parameterized,cerezo2021variational,schuld2021effect,schuld2020circuit,farhi2018classification,farhi2014quantum,holmes2022connecting,abbas2021power,schuld2021effect,larose2020robust,sim2019expressibility,cong2019quantum}, a widely studied paradigm in quantum machine learning, one prepares an initial state and applies a sequence of parameterized unitaries $U(\vec{\omega})$, where the gate operations depend on tunable parameters $\vec{\omega}$. The input parameters $\boldsymbol{\theta}$ are first embedded through a data-encoding unitary $U_{\text{enc}}(\boldsymbol{\theta})$, so that the overall state is $|\psi(\boldsymbol{\theta},\vec{\omega})\rangle = U(\vec{\omega}) U_{\text{enc}}(\boldsymbol{\theta})|0\rangle$. Measuring an observable $O$ at the output yields an expectation value $f(\boldsymbol{\theta},\vec{\omega}) = \langle 0| U_{\text{enc}}^\dagger(\boldsymbol{\theta}) U^\dagger(\vec{\omega}) O U(\vec{\omega}) U_{\text{enc}}(\boldsymbol{\theta}) |0\rangle$, which defines a family of functions parameterized jointly by the input $\boldsymbol{\theta}$ and the trainable parameters $\vec{\omega}$. By iteratively updating $\vec{\omega}$ during training, the expectation values can be steered toward desired outcomes, enabling the circuit to approximate target functions of $\boldsymbol{\theta}$ and thereby address a wide range of tasks in optimization, learning, and quantum simulation. The expressive power of the model is governed by the architecture of $U(\vec{\omega})$—including its depth and connectivity—as well as by the choice of data-encoding map $U_{\text{enc}}(\boldsymbol{\theta})$ and measurement observables.

Physical neural networks (PNNs) can be regarded as special instances of parametric quantum circuit learning \cite{boyd1985fading, tanaka2019recent, mujal2021opportunities, wilson2018quantum, garcia-beni2023scalable, havlicek2019supervised, rowlands2021reservoir, lin2018all-optical, pai2023experimentally, dambre2012information, sheldon2022computational, wu2021expressivity, hu2023tackling}. Unlike variational quantum circuits, PNNs contain no tunable gates within the physical evolution itself. Once the physical setup is specified, the dynamics are entirely fixed, so the parameters $\vec{\omega}$  and the corresponding unitary $U(\vec{\omega})$ that would normally be optimized in a parametric quantum circuit are predetermined. The evolution then produces a fixed set of measurement features $\eta_j(\boldsymbol{\theta})=\langle 0| U_{\text{enc}}^\dagger(\boldsymbol{\theta}) U^\dagger(\vec{\omega}) O_j U(\vec{\omega}) U_{\text{enc}}(\boldsymbol{\theta}) |0\rangle$, where $\boldsymbol{\theta}$ denotes the input parameters. Learning takes place only at the output stage, where the features ${\eta_j(\boldsymbol{\theta})}$ are linearly combined with adjustable weights $W_j$ to realize the desired task, $f(\boldsymbol{\theta})=\sum_j W_j \eta_j(\boldsymbol{\theta})$. In this sense, PNNs represent a limiting case of parametric quantum circuits: the mapping from input to features is dictated entirely by the fixed physics of the system, while optimization is restricted to the output weights ${W_j}$.

An important advantage of the PNN approach is that it avoids the need to optimize over a large number of variational parameters within the circuit. In standard parametric quantum circuit learning, training can be hindered by barren plateaus, slow convergence, and the heavy computational cost of navigating a high-dimensional parameter space \cite{cerezo2023does,cerezo2021cost,mcclean2018barren}. PNNs circumvent these challenges entirely: the physical evolution is fixed by the setup, and learning reduces to an  optimization over only the output weights. This makes training efficient and scalable with classical resources, while the underlying physics directly supplies a structured and interpretable feature map. Moreover, owing to the intrinsic complexity of the physical dynamics, PNNs may naturally generate highly expressive feature maps, offering the potential for computational advantages over classical learning models. Taken together, these features make PNNs a robust and data-efficient framework for quantum learning, with particular promise for near-term applications where experimental control is limited and simple, reliable training procedures are essential.

In the imaging context, the role of the PNN is realized by the optical system together with its measurement scheme. Information about the source is encoded into a quantum state that arrives at the imaging plane, evolves through the imaging channel, and is subsequently probed by measurements that generate the feature functions. Imaging tasks such as object recognition, classification, or image reconstruction can then be modeled as learning problems within the PNN framework.

PNNs share the same high-level structure as parametric quantum circuit learning—consisting of state encoding, physical evolution, and measurement—but they differ in a crucial respect: the internal dynamics are fixed once the physical setup is chosen, so the expressive capacity of the feature map is determined entirely by the physics of the system. This raises the central question of which functions, or equivalently which learning tasks, can be supported by a given PNN architecture. The resolvable expressive capacity  (REC) formalism established in Ref. \cite{hu2023tackling} provides a natural framework to address this question, offering a quantitative characterization of the expressive capacity and revealing how the approximability of target functions and the associated sample complexity are governed by the eigentasks of the system.

}

\subsection{Resolvable expressive capacity }\label{SI:REC}

\yk{We will use the tools of REC formalism introduced in Ref.~\cite{hu2023tackling}. } If our goal is approximate a function $f(\boldsymbol{\theta})$, define the capacity
\begin{equation}\label{SI_eq:C_f}
C[f]=1-\min _W\frac{ \mathbb{E}_{\boldsymbol{\theta}} \left[\mathbb{E}_{\mathcal{X}}\left[\left(\sum_{k=0}^{K-1} W_k \bar{\eta}_k(\boldsymbol{\theta})-f(\boldsymbol{\theta})\right)^2\right]\right] }{\int f^2(\boldsymbol{\theta}) p(\boldsymbol{\theta}) d \boldsymbol{\theta}},
\end{equation}
where we take the expectation value for the output sample $\mathcal{X}$ and the prior distribution $p(\boldsymbol{\theta})$, where $f(\boldsymbol{\theta})$ is approximated by a linear combination of measured functions $\bar{\eta}_k(\boldsymbol{\theta})$. 
\yk{$C[f]$ takes values between 0 and 1, representing the relative error of approximating the target function $f(\boldsymbol{\theta})$ as a linear combination of the measured features $\bar{\eta}_j(\boldsymbol{\theta})$. Here, $\bar{\eta}_j(\boldsymbol{\theta})$ denotes the measured features, which approximate the actual features ${\eta}_j(\boldsymbol{\theta})$.

The value of $C[f]$ depends on two aspects. Firstly, whether $f(\boldsymbol{\theta})$ can be well approximated by the features ${\eta}_j(\boldsymbol{\theta})$. Note not all functions of the input parameters $\boldsymbol{\theta}$ can be approximated as a linear combination of the features ${\eta}_j(\boldsymbol{\theta})$. The spanned linear subspace by all the features ${\eta}_j(\boldsymbol{\theta})$ actually quantifies the capability of the imaging system. A larger subspace spanned by ${\eta}_j(\boldsymbol{\theta})$ means we can approximate more functions of the input parameters $\boldsymbol{\theta}$. And which functions lie in this subspace gives the information of what can be obtained using this imaging system.
In the case that $f(\boldsymbol{\theta})$ is not in this subspace spanned by ${\eta}_j(\boldsymbol{\theta})$, the value of $C[f]$ will always significantly deviate from 1 even if we increase $S\rightarrow\infty$. In the worst case, we can at least choose $W_j=0$ for $\forall j$, which gives $C[f]=0$. So, if we cannot well approximate the target function, $C[f]$ takes values not close to 1.
Secondly, in the case that $f(\boldsymbol{\theta})$ is in this subspace spanned by ${\eta}_j(\boldsymbol{\theta})$, this means that with the actual value of features ${\eta}_j(\boldsymbol{\theta})$, we have $\sum_{j} W_j\bar{\eta}_j(\boldsymbol{\theta})-f(\boldsymbol{\theta})\approx 0$.  Furthermore, since $\bar{\eta}_j(\boldsymbol{\theta})$ is the measured feature as an approximation of the actual feature ${\eta}_j(\boldsymbol{\theta})$, when the sample number $S$ is not sufficiently large, the deviation of $\bar{\eta}_j(\boldsymbol{\theta})$ from ${\eta}_j(\boldsymbol{\theta})$ will also cause some error. But when $S\rightarrow\infty$, we have a very good approximation. In this case, $C[f]\approx 1$.
}
The mean squared error loss serves as an effective approximation of the first term in the Taylor expansion for a wide range of  nonlinear post-processing and non-quadratic loss functions \cite{hu2023tackling}.

\yk{The total REC is defined as $C_T=\sum_{l=0}^\infty C[f_l]$, where $\{f_l\}_l$ represents a complete and orthonormal set of basis functions in the space of all functions of input parameters $f(\boldsymbol{\theta})$. Note that this set of basis functions has infinitely many elements, as the space of functions of $\boldsymbol{\theta}$ is infinite-dimensional. However, the calculation of $C_T=\sum_{l=0}^\infty C[f_l]$ can be reduced to just the eigentasks defined based on the following eigenvalue problem. }
\begin{equation}\label{SI_eq:C_T_trace}
C_T=\operatorname{Tr}\left(\left(G+\frac{1}{S} V\right)^{-1}G\right)=\sum_{k=0}^{K-1} \frac{1}{1+\beta_k^2 / S},
\end{equation}
\begin{equation}
\begin{gathered}\label{SI_eq:eigenproblem}
D_{k k}=\operatorname{Tr}\left\{\hat{M}_k \hat{\rho}^{(1)}\right\}, \quad G_{j k}=\operatorname{Tr}\left\{\left(\hat{M}_j \otimes \hat{M}_k\right) \hat{\rho}^{(2)}\right\},\quad \hat{\rho}^{(t)}=\int {\rho}(\boldsymbol{\theta})^{\otimes t} p(\boldsymbol{\theta}) d \boldsymbol{\theta},\\
V=D-G,\quad Vr_{k}=\beta_k^2 G r_{k}.
\end{gathered}
\end{equation}
where $p(\boldsymbol{\theta})$ describes our prior knowledge about the input, $\rho(\boldsymbol{\theta})$ is the quantum state encoded with input data $\boldsymbol{\theta}$, $S$ is the number of samples, $\{M_k\}$ is a set of $K$ POVM elements. The total REC $C_T$ describes the capability of approximating arbitrary functions $f(\boldsymbol{\theta})$ using $\bar{\eta}_k(\boldsymbol{\theta})$. In this work, $\bar{\eta}_k(\boldsymbol{\theta})$ corresponds to the observed frequency of occurrence of the $k$th measurement outcome and $\mathbb{E}_{\mathcal{X}}[\overline{{\eta}}_k(\boldsymbol{\theta})] = \eta_k(\boldsymbol{\theta}) = P_k(\boldsymbol{\theta}) = \tr(\rho(\boldsymbol{\theta}) M_k)$. 
A larger $C_T$, or smaller $\beta_k^2$, indicates better learning performance. Note that this theoretical framework incorporates only $\hat{\rho}^{(t=1,2)}$. This is because $C[f]$ is defined based on the square error, so when deriving the equation, only the expectation and covariance of $\bar{\eta}_k(\boldsymbol{\theta})$ are involved. Since we only need the expectation and covariance of the probability distribution, $\hat{\rho}^{(t=1,2)}$ is sufficient.



The eigenvectors $r_k$ determines the eigentasks 
\begin{equation}\label{SI_eq:f_k}
f_k(\boldsymbol{\theta})=\sum_j r_{kj} \eta_j(\boldsymbol{\theta}),
\end{equation}
where $\eta_j(\boldsymbol{\theta})=\mathbb{E}_{\mathcal{X}}[\overline{{\eta}}_j(\boldsymbol{\theta})]$ is probability of obtaining the $k$th measurement outcome. \yk{The REC of each eigentask is given by $C[f_k(\boldsymbol{\theta})] = 1/(1 + \beta_k^2 / S)$. Intuitively, the total REC $C_T$ accounts for all possible basis functions. In practice, however, most basis functions contribute negligibly to $C_T$, and it suffices to consider only the eigentasks $f_k$ constructed from the vectors $r_k$ that have significant contributions. Importantly, in the finite-sample regime, the relevance of each eigentask is well quantified by its associated eigenvalue $\beta_k^2$. } In the actual experiment, the direct measurement results yield ${\overline{\eta}_j(\boldsymbol{\theta})}$. 
In general, the information obtained from each $\eta_j$ may not be independent. On the other hand, the eigentasks $f_k$ are orthonormal in the sense of 
\begin{equation}
\mathbb{E}_{\boldsymbol{\theta}}[f_k(\boldsymbol{\theta})f_j(\boldsymbol{\theta})]=\delta_{kj}. 
\end{equation} 
Intuitively, these can be interpreted as the ``independent'' tasks with learning precision quantified by $\beta_k^2$. The calculation of eigentasks can give us a sense of what function is hard to be estimated for our measurement device. 

\yk{
We emphasize that a key advantage of the REC formalism is that the identified eigentasks and the total REC are invariant under reparameterization. These eigentasks are determined by the prior information, the physical model of the learning device, and the structure of the learning tasks. They are then used as the actual feature vectors for downstream analysis, including model training and inference with the trained model. This ensures that the results are not artificially influenced by the choice of parameterization. The invariance under reparameterization is formally established in the following proposition.

\begin{proposition}[Reparameterization invariance of total REC and eigentasks]\label{prop:reparameterization}
Let $\rho(\boldsymbol{\theta})$ be a family of states parameterized by $\boldsymbol{\theta}$ with prior $p(\boldsymbol{\theta})$, a fixed POVM $\{M_i\}_{i=0}^{K-1}$, sample size $S$, and measured features $\eta_i(\boldsymbol{\theta})=\operatorname{Tr}[\rho(\boldsymbol{\theta})M_i]$.
In the parameterization of $\boldsymbol{\theta}$, the total REC at sample size $S$ is $C_T(S)$, with eigenvalues $\{\beta_k^2\}$ and eigentasks defined by solving Eq.~\ref{SI_eq:eigenproblem}, yielding $f_k(\boldsymbol{\theta})=\sum_j r_{kj}\,\eta_j(\boldsymbol{\theta})$ as in Eq.~\ref{SI_eq:f_k}.
Let $\boldsymbol{\phi}=h(\boldsymbol{\theta})$ be a bijective differentiable reparameterization, with pushforward prior $p_\Phi(\boldsymbol{\phi})=p(h^{-1}(\boldsymbol{\phi}))\,|\det J_{h^{-1}}(\boldsymbol{\phi})|$ so that $p_\Phi(\boldsymbol{\phi})\,d\boldsymbol{\phi}=p(\boldsymbol{\theta})\,d\boldsymbol{\theta}$, and reparameterized family $\tilde\rho(\boldsymbol{\phi})=\rho(h^{-1}(\boldsymbol{\phi}))$. As in Eq.~\ref{SI_eq:eigenproblem}, form $\tilde\rho^{(t)}$, $\tilde D$, $\tilde G$, $\tilde V$, the generalized spectrum $\{\tilde\beta_k^2\}$, eigentasks $\tilde f_k(\boldsymbol{\phi})=\sum_j \tilde r_{kj}\,\tilde\eta_j(\boldsymbol{\phi})$ with $\tilde\eta_j(\boldsymbol{\phi})=\eta_j(h^{-1}(\boldsymbol{\phi}))=\eta_j(\boldsymbol{\theta})=\operatorname{Tr}[\rho(\boldsymbol{\theta})M_j]$, and total REC $\tilde C_T(S)$. Then:
\begin{enumerate}
\item (Total REC invariance) $\;\tilde C_T(S)=C_T(S)$.
\item (Spectral invariance) $\;\{\tilde\beta_k^2\}=\{\beta_k^2\}$.
\item (Eigentask invariance) $\;\tilde f_k(\boldsymbol{\phi})=f_k\!\big(h^{-1}(\boldsymbol{\phi})\big)=f_k(\boldsymbol{\theta})$.
\end{enumerate}
\end{proposition}

\begin{proof}
By definition and the pushforward prior,
\[
\tilde\rho^{(t)}
=\int_{\Phi}\tilde\rho(\boldsymbol{\phi})^{\otimes t}\,p_\Phi(\boldsymbol{\phi})\,d\boldsymbol{\phi}
=\int_{\Phi}\rho\!\big(h^{-1}(\boldsymbol{\phi})\big)^{\otimes t}\,
p\!\big(h^{-1}(\boldsymbol{\phi})\big)\,\bigl|\det J_{h^{-1}}(\boldsymbol{\phi})\bigr|\,d\boldsymbol{\phi}.
\]
Set $\boldsymbol{\theta}=h^{-1}(\boldsymbol{\phi})$. Then $d\boldsymbol{\phi}=\bigl|\det J_{h}(\boldsymbol{\theta})\bigr|\,d\boldsymbol{\theta}$ and $\bigl|\det J_{h^{-1}}(\boldsymbol{\phi})\bigr|=1/\bigl|\det J_{h}(\boldsymbol{\theta})\bigr|$, so the Jacobians cancel:
\[
\tilde\rho^{(t)}
=\int_{\Theta}\rho(\boldsymbol{\theta})^{\otimes t}\,p(\boldsymbol{\theta})\,d\boldsymbol{\theta}
=\rho^{(t)}\qquad (t=1,2).
\]
This relation implies
\[
\tilde D_{ii}=\operatorname{Tr}[M_i\,\tilde\rho^{(1)}]=\operatorname{Tr}[M_i\,\rho^{(1)}]=D_{ii},\qquad
\tilde G_{ij}=\operatorname{Tr}[(M_i\!\otimes\! M_j)\,\tilde\rho^{(2)}]
=\operatorname{Tr}[(M_i\!\otimes\! M_j)\,\rho^{(2)}]=G_{ij}.
\]
Hence $\tilde V=\tilde D-\tilde G=D-G=V$.
With $\tilde G=G$ and $\tilde V=V$, we obviously have
\[
\tilde C_T(S)=\operatorname{Tr}\!\big((\tilde G+\tilde V/S)^{-1}\tilde G\big)
=\operatorname{Tr}\!\big((G+V/S)^{-1}G\big)=C_T(S).
\]
The generalized eigenproblems
\[
V\,r_k=\beta_k^2\,G\,r_k,\qquad
\tilde V\,\tilde r_k=\tilde\beta_k^2\,\tilde G\,\tilde r_k
\]
are also identical, so $\{\tilde\beta_k^2\}=\{\beta_k^2\}$ and  $\tilde r_k=r_k$. Using $\tilde\eta(\boldsymbol{\phi})=\eta(h^{-1}(\boldsymbol{\phi}))$,
\[
\tilde f_k(\boldsymbol{\phi})=\sum_j r_{kj}\,\tilde\eta_j(\boldsymbol{\phi})
=\sum_j r_{kj}\,\eta_j\!\big(h^{-1}(\boldsymbol{\phi})\big)
=f_k\!\big(h^{-1}(\boldsymbol{\phi})\big)=f_k(\boldsymbol{\theta}).
\]
\end{proof}
}

\section{Two-point sources}\label{SI:two_point}

In this section, we give more details about the calculation of imaging two-point sources using direct imaging and the binary SPADE method proposed in Ref.~\cite{tsang2016quantum}.
For the binary SPADE method, the POVM can be written as $M_0=\ket{\phi_0}\bra{\phi_0}$, $M_1=I-M_0$, $\ket{\phi_0}=\int du\phi_0(u)\ket{u}$, $\phi_0(u)=\frac{1}{(2\pi\xi^2)^{1/4}}\exp(-\frac{u}{4\xi^2})$. We can then calculate the  matrix elements of $D,G$ for the eigenvalue problem $V=D-G$, $Vr_{k}=\beta_k^2 G r_{k}$.
\begin{equation}
D_{00}=\tr(M_0\rho^{(1)})=\frac{4\xi\sigma}{
\sqrt{(\xi^2+\sigma^2)(\gamma^2+4\xi^2+4\sigma^2)}},\quad D_{11}=1-D_{00},
\end{equation}
\begin{equation}
\begin{aligned}
&G_{00}=\tr((E_0\otimes E_0)\rho^{(2)})=\frac{4\sqrt{2}\sigma^2\xi^2}{(\sigma^2+\xi^2)^{3/2}\sqrt{\gamma^2+2\xi^2+2\sigma^2}},\\
&G_{01}=G_{10}=\tr((E_0\otimes E_1)\rho^{(2)})=\tr((E_0\otimes I)\rho^{(2)})-\tr((E_0\otimes E_0)\rho^{(2)})=D_{00}-G_{00},\\
&G_{11}=\tr((E_1\otimes E_1)\rho^{(2)})=\tr(\rho^{(2)})-\tr((I\otimes E_0)\rho^{(2)})-\tr((E_0\otimes I)\rho^{(2)})+\tr((E_0\otimes E_0)\rho^{(2)})=1-2D_{00}+G_{00},
\end{aligned}
\end{equation}
\begin{equation}
D=\left[\begin{matrix}
D_{00} & 0\\
0 & D_{11}
\end{matrix}\right],\quad G=\left[\begin{matrix}
G_{00} & G_{01}\\
G_{10} & G_{11}
\end{matrix}\right],
\end{equation}
which then gives us the eigenvalues
\begin{equation}\begin{aligned}
&\beta_0^2=0,\\
&\beta_1^2=\frac{-\sqrt{2} \xi \sigma (\gamma^2 + 4 (\xi^2 + \sigma^2)) + 
   \sqrt{\xi^2 + \sigma^2} \sqrt{\gamma^2 + 2 (\xi^2 + \sigma^2)}
     \sqrt{(\xi^2 + \sigma^2) (\gamma^2 + 
       4 (\xi^2 + \sigma^2))}}{\xi \sigma (4 \sqrt{2} \xi^2 + 
     \sqrt{2} (\gamma^2 + 4 \sigma^2) - 
     4 \sqrt{\xi^2 + \sigma^2} \sqrt{\gamma^2 + 2 (\xi^2 + \sigma^2)})}.    
\end{aligned}
\end{equation}
If we assume $\gamma\ll \sigma,\xi$ and define $\alpha=\gamma/\sigma$, we have 
\begin{equation}
\beta_1^2=\frac{16 (\xi - \sigma)^2 (\xi^2 + \sigma^2)^2}{\xi \sigma \gamma^4}+ \frac{
 2 (\xi^2 + \sigma^2) (-8 \xi \sigma + 5 (\xi^2 + \sigma^2))}{
 \xi \sigma \gamma^2} + 
 \frac{3}{8} \left(-4 + \frac{3 \xi}{\sigma} + \frac{3 \sigma}{\xi}\right) - \frac{\gamma^2}{
 64 (\xi \sigma)}+O(\alpha^4).
\end{equation}
The first-order term $\Theta(\alpha^{-4})$ can be made to vanish if we choose $\xi = \sigma$. With this choice, the leading order of $\beta_1^2$ becomes proportional to $\alpha^{-2}$.

\begin{figure}[!tb]
\begin{center}
\includegraphics[width=0.5\columnwidth]{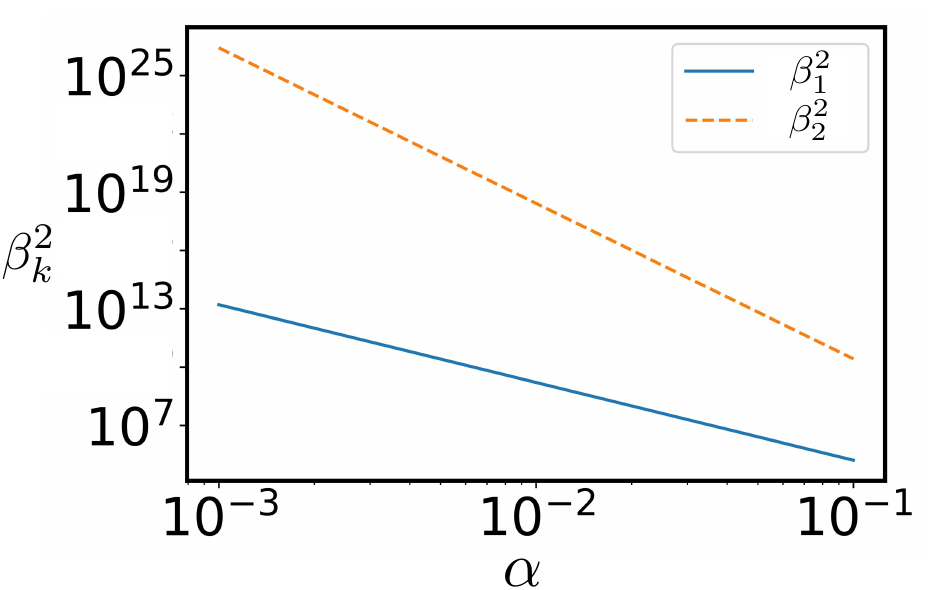}
\caption{Scaling of the eigenvalue $\beta_{1,2,3}^2$ as a function of $\alpha$ for direct imaging of two-point sources.} 
\label{beta_alpha}
\end{center}
\end{figure}

We may want to compare this with the direct imaging case where we directly project onto each spatial mode $\{E_x=\ket{x}\bra{x}\}_x$. We again calculate the matrix elements as
\begin{equation}
D_{xx}=\tr(E_x\rho^{(1)})=\frac{\sqrt{2/\pi}}{\sqrt{\gamma^2+4\sigma^2}}\exp[-2x^2/(\gamma^2+4\sigma^2)],
\end{equation}
\begin{equation}
G_{xy}=\tr((E_x\otimes E_y)\rho^{(2)})=\frac{\exp(-\frac{x^2+y^2}{2\sigma^2})}{2\sqrt{2}\pi\gamma\sigma^2\sqrt{\frac{2}{\gamma^2}+\frac{1}{\sigma^2}}}\left[\exp\left(\frac{\gamma^2(x-y)^2}{4\sigma^2(\gamma^2+2\sigma^2)}\right)+\exp\left(\frac{\gamma^2(x+y)^2}{4\sigma^2(\gamma^2+2\sigma^2)}\right)\right].
\end{equation}
It is possible to numerically find $\beta_k^2$ as shown in Fig.~\ref{beta_alpha}. We again choose $\alpha=\gamma/\sigma$.  In this case, the smallest eigenvalue $\beta_0^2 = \Theta(1)$, followed by $\beta_1^2 = \Theta(\alpha^{-4})$ and $\beta_2^2 = \Theta(\alpha^{-8})$. 
\yk{Values of $\beta_k^2$ exceeding $\sim 10^{15}$ in Fig.~\ref{beta_alpha} are interpolated using a straight-line fit at larger $\alpha$ to mitigate numerical precision loss; this interpolation is applied to all $\beta_k^2$--$\alpha$ plots.}
More details about the numerical calculation can be found in Sec.~\ref{SI:numerics}. 

\yk{
We use the prior distribution $p(L) = \frac{1}{\sqrt{2\pi}\gamma} \exp(-\frac{L^2}{2\gamma^2})$, which reflects the assumption that the separation $L$ is much smaller than the PSF width $\sigma$. This choice is motivated by the fact that the advantage of superresolution over direct imaging becomes significant only in the Rayleigh limit where $L \ll \sigma$. For a fair comparison, both methods are given the same prior. The prior serves as a probability distribution over possible separations $L$, making the overall precision—as quantified by the total REC—a weighted average of local precisions at each value of $L$. This is particularly relevant in subdiffraction imaging, where precision varies strongly with $L$: direct imaging performs poorly for small $L$, while superresolution maintains high precision. Although superresolution still performs well for larger $L$, its relative advantage over direct imaging diminishes.
Conventional parameter estimation frameworks based on Fisher information assess precision locally at the true parameter value, which is typically unknown. In contrast, the quantum learning approach incorporates the prior directly and defines precision globally across the parameter space. This makes it more suitable when the parameter may vary over a range of values.
Finally, the prior alone does not guarantee accurate estimation. The REC formalism, which measures relative error. While the prior implies $L \ll \sigma$, this does not ensure low relative error. In practice, it is easy to recognize that two sources are too close to be resolved by direct imaging, but extracting precise subdiffraction information is much more challenging—and better accomplished with superresolution. Therefore, we believe this prior is a well-motivated choice for the present discussion.

}

\section{A single compact source (generally distributed source within the Rayleigh limit)}\label{SI:single_compact_source}


\subsection{Direct imaging of a single compact source}\label{SI:one_source_direct}

We model a general incoherent source as
\begin{equation}\begin{aligned}\label{rho_k}
\rho&=\sum_{k=0}^\infty x_k\rho^{(k)} ,  
\end{aligned}
\end{equation}
where $\rho^{(k)}=\Theta(\alpha^k)$, which depends on the choice of PSF as detailed in Sec.~\ref{SI:preliminary_superresolution} of the supplemental and moments vector $\vec{x}=[x_0,x_1,\cdots,]^T$ is our input data, i.e. $\boldsymbol{\theta}=\vec{x}$. For direct imaging, we directly project onto each position $\ket{u}$, and the probability distribution for the direct imaging is given in the form of 

\begin{equation}\begin{aligned}
&P_{n}=\sum_{m=0}^\infty c_{nm}x_m\alpha^m,\quad n=0,1,2,3,\cdots
\end{aligned}\end{equation}
Again, we measure $M_n = \int_{X_n - l/2}^{X_n + l/2} du \ket{u}\bra{u}$, a discretized version of direct imaging for both analytical and numerical discussions. 
Assume we have the prior knowledge about the moments $\vec{x}=[x_0,x_1,x_2,\cdots]$ described by $p(\vec{x})$. The prior distribution can be derived from the knowledge that the image is selected from all possible images $I_m(u)$, weighted by their probabilities $p(I_m(u))$. Each image has a corresponding moment vector $\vec{x}(m)$. Therefore, the prior distribution $p(\vec{x})$ for each moment vector $\vec{x}$ can be calculated by weighting it according to the probabilities $p(I_m(u))$ of all the images $I_m(u)$ associated with the moments $\vec{x}$. For any moment vector $\vec{x}$ that does not correspond to a valid image $I_m(u)$, or if all corresponding images $I_m(u)$ have $p(I_m(u)) = 0$, we simply assume that $p(\vec{x}) = 0$.
We define
\begin{equation}\label{SI_eq:d_g}
d_k=\int d\vec{x}p(\vec{x}) x_k,\quad g_{k_1k_2}=\int d\vec{x}p(\vec{x}) x_{k_1}x_{k_2}.
\end{equation}
We can rewrite the eigenvalue problem as
\begin{equation}\label{SI_eq:eigentask}
Gr_k=\lambda_kDr_k,
\end{equation}
where $\lambda_k=(\beta_k^2+1)^{-1}$. Notice that $D_{kk}=\int d\vec{x} p(\vec{x}) P_k(\vec{x})$ and $G_{jk}=\int d\vec{x} p(\vec{x})P_j(\vec{x})P_k(\vec{x})$. We can expand $D$ and $G$ as a series of $\alpha$
\begin{equation}
D=\sum_{n=0}^\infty D^{(n)}=\sum_{n=0}^\infty d_n\alpha^n \mathcal{C}_n, \quad\mathcal{C}_n=\left[\begin{matrix}
c_{0n} & 0 & 0 & \cdots\\
0 & c_{1n} & 0 & \cdots\\
0 & 0 & c_{2n} & \cdots\\
\cdots & \cdots & \cdots & \cdots
\end{matrix}\right],
\end{equation}
\begin{equation}\label{SI_eq:D_G_DI_single}
G=\sum_{n=0}^\infty G^{(n)}=\sum_{n=0}^\infty\alpha^n\sum_{i+j=n}g_{ij}  C_i C_j^T, \quad C_i= [c_{0i},c_{1i},c_{2i},\cdots]^T.
\end{equation}

We can reformulate the eigenvector problem as $W y_k = \lambda_k y_k$, where $W = D^{-1/2} G D^{-1/2}$. In the basis $\{z_n\}_{n=0,1,2,\cdots}$, which is the orthonormal basis constructed from the set of vectors $\{D^{-1/2} C_n\}_{n=0,1,2,\cdots}$ using the Gram-Schmidt procedure, the elements of $W$ are given by $z_m^\dagger W z_n = \Theta(\alpha^{m+n})$. We can intuitively predict the scaling of eigenvalues from its character polynomial 
\begin{equation}
|\lambda I-W|=\sum_{k=0}^M(-1)^k e_k\lambda^{M-k},
\end{equation}
where
\begin{equation}\begin{aligned}
&e_0=1,\quad e_1=\sum\lambda_i,\quad e_2=\sum\lambda_i\lambda_j,\quad e_3=\sum\lambda_i\lambda_j\lambda_k,\quad \cdots,\quad e_M=\prod_{i=1}^M\lambda_i.
\end{aligned}\end{equation}
To determine the scaling of $\lambda_i$, we only need to know the leading-order scaling of $e_i$ with respect to $\alpha$. We observe that selecting the largest $k$ diagonal elements $z_n^\dagger W z_n$ provides the leading-order scaling for each $e_k$. This is because the contribution to each $e_k$ from off-diagonal elements $z_m^\dagger W z_n$ can at most share the same order of scaling with respect to $\alpha$. However, there is a possibility of cancellation among the contributions from the diagonal and off-diagonal elements, making this reasoning somewhat heuristic and not entirely rigorous. Essentially, we are trying to determine the scaling of $e_k$ for a polynomial
\begin{equation}
(\lambda-\Theta(\alpha^0))(\lambda-\Theta(\alpha^2))(\lambda-\Theta(\alpha^4))(\lambda-\Theta(\alpha^6))\cdots
\end{equation}
We can read from it 
\begin{equation}
e_1=\Theta(\alpha^0), \quad e_2=\Theta(\alpha^{0+2}),\quad  e_3=\Theta(\alpha^{0+2+4}),\quad\cdots
\end{equation}
So, the scaling of eigenvalue 
\begin{equation}\label{scaling_direct}
\lambda=\Theta(\alpha^0),\quad \Theta(\alpha^{2}),\quad \Theta(\alpha^{4}),\quad \Theta(\alpha^{6}), \quad \Theta(\alpha^{8}),\quad\cdots
\end{equation}
which implies
\begin{equation}
\beta_0^2=\Theta(1),\, \beta_1^2=\Theta(\alpha^{-2}),\, \beta_2^2=\Theta(\alpha^{-4}),\, \beta_3^2 = \Theta(\alpha^{-6}),\cdots
\end{equation}

We now try to find the scaling of eigenvalues and the eigenbasis by  perturbation theory up to the $O(\alpha^{2})$ order for $\lambda$ below,
\begin{equation}
\left(\sum_{n=0}^{\infty} G^{(n)}\right)\left(\sum_{n=0}^{\infty} r_k^{(n)}\right)=\left(\sum_{n=0}^{\infty} \lambda_k^{(n)}\right)\left(\sum_{n=0}^{\infty} D^{(n)}\right)\left(\sum_{n=0}^{\infty} r_k^{(n)}\right).
\end{equation}
To the zeroth order 
\begin{equation}
G^{(0)} r_k^{(0)}=\lambda_k^{(0)} D^{(0)} r_k^{(0)},
\end{equation}
\begin{equation}
\lambda_0^{(0)}=O(1),\quad \lambda_{j\neq0}^{(0)}=0,\quad r_0^{(0)}\propto[1,1,1,\cdots,1]^T.
\end{equation}
We easily prove that the properly normalized eigenvector $r_j^{(0)}$ satisfies
\begin{equation}
r_i^{(0)\dagger }D^{(0)}r_j^{(0)}=\delta_{ij},
\end{equation}
We now want to do degenerate perturbation theory in the subspace $\text{span}\{r_{k\neq 0}^{(0)}\}$.
\begin{equation}
\sum_{i+j=1}G^{(i)}r^{(j)}_k=\sum_{p+q+m=1}\lambda_k^{(p)}D^{(q)}r^{(m)}_k,
\end{equation}
where $k\geq 1$. Notice that $\lambda^{(0)}_{k\geq 1}=0$, $r.h.s=\lambda_k^{(1)}D^{(0)}r^{(0)}_k$. Notice that $G^{(0)}$ is a zero matrix in the subspace $\text{span}\{r_{j\neq 0}^{(0)}\}$, $l.h.s=G^{(1)}r_k^{(0)}$. Multiplying $r^{(0)\dagger}_l$ in both side
\begin{equation}\label{first_order_perturbation}
r^{(0)\dagger}_lG^{(1)}r_k^{(0)}=\lambda_k^{(1)}r^{(0)\dagger}_lD^{(0)}r^{(0)}_k,
\end{equation}
where $k,l\geq 1$. Furthermore, notice that $r^{(0)\dagger}_{l\geq 1}D^{(0)}r^{(0)}_0\propto r^{(0)\dagger}_{l\geq 1}C_0^T=0$ and $G^{(1)}=\alpha g_{01}C_0C_1^T+\alpha g_{10}C_1C_0^T$. It is clear that the $l.h.s$ of Eq.~\ref{first_order_perturbation} vanishes. We thus have
\begin{equation}
\lambda_{k\geq 1}^{(1)}=0.
\end{equation}
Similarly, we can have the second order perturbation theory.
\begin{equation}
\sum_{i+j=2}G^{(i)}r^{(j)}_k=\sum_{p+q+m=2}\lambda_k^{(p)}D^{(q)}r^{(m)}_k.
\end{equation}
Notice that $\lambda^{(0,1)}_{k\geq 1}=0$, $r.h.s=\lambda_k^{(2)}D^{(0)}r^{(0)}_k$. Notice that $G^{(0,1)}$ is a zero matrix in the subspace $\text{span}\{r_{j\neq 0}^{(0)}\}$, $l.h.s=G^{(2)}r_k^{(0)}$. Multiplying $r^{(0)\dagger}_l$ in both side
\begin{equation}
r^{(0)\dagger}_lG^{(2)}r_k^{(0)}=\lambda_k^{(2)}r^{(0)\dagger}_lD^{(0)}r^{(0)}_k=\lambda_k^{(2)}\delta_{kl},
\end{equation}
\begin{equation}
r^{(0)\dagger}_lG^{(2)}r_k^{(0)}=\alpha^2g_{11}r^{(0)\dagger}_lC_1C_1^Tr_k^{(0)}.
\end{equation}
Since $r^{(0)\dagger}_lC_1C_1^Tr_k^{(0)}$ is rank 1, we have proved that $\lambda_1^{(2)}=\Theta(\alpha^2)$ and $\lambda_{j\geq 2}^{(2)}=0$. Eigenvector $r_1^{(0)}$ is given by orthonogalizing $\mathcal{C}_0^{-1}C_1$ with respect to $r_0^{(0)}\propto\mathcal{C}_0^{-1}C_0$. Note this also means that $z_n$ is the zeroth order eigenbasis for $Wy_k=\lambda_k y_k$ for the first two eigenvectors because $\{z_n\}_{n=0,1,2,\cdots}$  is the orthogonalization of $\{D^{-1/2}\mathcal{C}_n\}_{n=0,1,2,\cdots}$.
For higher order perturbation, we need to include the higher order correction of eigenbasis $r_k^{(j\geq 1)}$, which becomes more involved. Therefore, we choose to stop at this point.

\subsection{Image a single compact source using SPADE method}\label{SI:one_source_superresolution}
As reviewed in  Sec.~\ref{SI:preliminary_superresolution}, the probability for SPADE method takes the form of
\begin{equation}\label{eq:P_one_source_superresolution}
P_0=\sum_{m=0}^\infty c_m x_m\alpha^m,\quad P_{n\pm}=\sum_{m\geq 2n}c_{nm}^{\pm}x_m\alpha^m,\quad n=0,1,2,3,\cdots
\end{equation}
where the low order terms of $O(\alpha^{m\leq 2n-1})$ in $P_{n\pm}$ is cancelled out  by construction. 
We can define

\begin{equation}\begin{aligned}\label{SI_eq:eigentask_SPADE}
&\yk{Gr_{k}=\lambda_kDr_{k},\quad \lambda_k=(1+\beta_k^2)^{-1},}\\
&D=\sum_{n=0}^\infty D^{(n)}=\sum_{n=0}^\infty d_n\alpha^n \mathcal{C}_n,\quad \mathcal{C}_n=\Pi_{\lfloor n/2\rfloor}\,\mathcal{C}_n^0\,\Pi_{\lfloor n/2\rfloor},\\
&\mathcal{C}_n^0=\left[\begin{matrix}
c_{n} & 0 & 0 & 0 & 0 &\cdots\\
0 & c_{0n}^+ & 0 & 0 & 0 &\cdots\\
0 & 0 & c_{0n}^- & 0 & 0 &\cdots\\
0 & 0 & 0 & c_{1n}^+ & 0 &\cdots\\
0 & 0 & 0 & 0 & c_{1n}^- & \cdots\\
\cdots & \cdots & \cdots & \cdots & \cdots & \cdots
\end{matrix}\right],
\end{aligned}
\end{equation}
\begin{equation}\begin{aligned}\label{SI_eq:D_G_superresolution_single}
&G=\sum_{n=0}^\infty G^{(n)}=\sum_{n=0}^\infty\alpha^n\sum_{i+j=n}g_{ij}  C_i C_j^T, \quad C_i=\Pi_{\lfloor i/2\rfloor}C^0_i,\\
&C^0_i= [c_i,c_{0i}^+,c_{0i}^-,c_{1i}^+,c_{1i}^-,\cdots]^T, \quad \Pi_m=\left[\begin{matrix}
I_{2m+3} & 0\\
0 & 0
\end{matrix}\right].
\end{aligned}\end{equation}

\yk{
In the following, we analyze the spectrum of the SPADE method using perturbation theory. To begin, we provide some intuition about the structure of the derived eigenbasis. Observing the form of $G$, we see that truncating its expansion at order $O(\alpha^n)$ makes the higher-order terms negligible, resulting in a rank-deficient matrix. For example, retaining terms only up to order $\alpha^2$ gives
\begin{equation}
G = g_{00} C_0 C_0^T + \alpha (g_{01} C_0 C_1^T + g_{10} C_1 C_0^T) + \alpha^2 (g_{20} C_2 C_0^T + g_{11} C_1 C_1^T + g_{02} C_0 C_2^T) + o(\alpha^2).
\end{equation}
Such a truncated $G$ captures the leading eigenvalues, as the neglected higher-order terms contribute only weakly to the spectrum. A key difference from the direct imaging case is that the structure $C_i = \Pi_{\lfloor i/2 \rfloor} C^0_i$ forces certain components to vanish, effectively restricting the support of each $C_i$. When we examine the eigenvectors within the subspace spanned by $C_0, C_1, C_2$, we find that the matrix $r_{nm}$ takes an upper-triangular form. The higher-order eigentasks are dominated by higher-order moments in the case of imaging a single compact source. This intuition is rigorously validated by the perturbative analysis presented below.

}

We again reformulate the eigenvalue problem  as
$Wy_k=\lambda_k y_k$, $W=D^{-1/2}GD^{-1/2}$. Note the lower order $D$ is not full rank. Rewriting the problem as standard eigenvalue problem avoids this difficulty. To further simplify the calculation, we will divide the matrix into blocks, the $(m,n)$th block is given by
\begin{equation}\begin{aligned}
&W_{mn}=[D^{-1/2}GD^{-1/2}]_{mn}\\
=&P_m\left(\frac{\alpha^{-m}}{\sqrt{d_{2m}}}\mathcal{C}_{2m}^{-1/2}\left[I-\frac{\alpha}{2}\frac{d_{2m+1}}{d_{2m}}\mathcal{C}_{2m+1}\mathcal{C}_{2m}^{-1}+\frac{3}{8}\alpha^2\frac{d_{2m+1}^2}{d_{2m}^2}\mathcal{C}_{2m+1}^2\mathcal{C}_{2m}^{-2}-\frac{1}{2}\alpha^2\frac{d_{2m+2}}{d_{2m}}\mathcal{C}_{2m+2}\mathcal{C}_{2m}^{-1}+\cdots\right]\right)\\
&\times\sum_{i\geq 2m,j\geq 2n}\alpha^{i+j}g_{ij}C_iC_j^T\\
&\times  \left(\frac{\alpha^{-n}}{\sqrt{d_{2n}}}\mathcal{C}_{2n}^{-1/2}\left[I-\frac{\alpha}{2}\frac{d_{2n+1}}{d_{2n}}\mathcal{C}_{2n+1}\mathcal{C}_{2n}^{-1}+\frac{3}{8}\alpha^2\frac{d_{2n+1}^2}{d_{2n}^2}\mathcal{C}_{2n+1}^2\mathcal{C}_{2n}^{-2}-\frac{1}{2}\alpha^2\frac{d_{2n+2}}{d_{2n}}\mathcal{C}_{2n+2}\mathcal{C}_{2n}^{-1}+\cdots\right]\right)  P_n,
\end{aligned}
\end{equation}
where $P_n=\Pi_n-\Pi_{n-1}$, $\Pi_{m<0}=0$. Note the leading order of block $W_{mn}$ has lowest order $\Theta(\alpha^{m+n})$.

We will define the following set of orthonormal basis for convenience
\begin{equation}\label{y0}
\begin{aligned}
&z_{2k}\propto P_k\mathcal{C}_{2k}^{1/2}t,\quad z_{2k+1}\bot z_{2k},\quad z_{2k+1}\in P_k,\\
\end{aligned}\end{equation}
where $k=0,1,2,3,\cdots$, and $z_{2k+1}$ is the vector constructed by orthogonalize $P_k\mathcal{C}_{2k}^{-1/2}\mathcal{C}_{2k+1}t$ over $z_{2k}$. And it should be clear that in this basis, for  the ($m$,$n$)th block $W_{mn}=P_mWP_n$, the scaling of  over $\alpha$ is given by
\begin{equation}
\left[
\begin{matrix}
\Theta(\alpha^{m+n}) & \Theta(\alpha^{m+n+1})\\
\Theta(\alpha^{m+n+1}) & \Theta(\alpha^{m+n+2})
\end{matrix}\right],
\end{equation}
To find the scaling of eigenvalues, we again try to observe the structure of the character polynomial $|\lambda I-W|=\sum_{k=0}^M(-1)^k e_k\lambda^{M-k}$, its coefficients now has the form of 
\begin{equation}\begin{aligned}
&e_1=\Theta(\alpha^0), \quad e_2=\Theta(\alpha^{0+2}),\quad  e_3=\Theta(\alpha^{0+2+2}),\\
&e_4=\Theta(\alpha^{0+2+2+4}),\quad  e_5=\Theta(\alpha^{0+2+2+4+4}),\cdots    
\end{aligned}
\end{equation}
So, the scaling of eigenvalue is
\begin{equation}
\beta_0^2=\Theta(1), \beta_1^2=\Theta(\alpha^{-2}), \beta_2^2=\Theta(\alpha^{-2}), \beta_3^2=\Theta(\alpha^{-4}), \beta_4^2=\Theta(\alpha^{-4}),\cdots
\end{equation}

We now calculate the eigenvalues and eigenbasis based on perturbation theory. We can write $W$ as a series of $\alpha$,
\begin{equation}
W=\sum_{n=0}^\infty W^{(n)}\alpha^n.    
\end{equation}
The zeroth order term 
\begin{equation}
W^{(0)}\propto P_0\mathcal{C}_0^{1/2}tt^T\mathcal{C}_0^{1/2}P_0,
\end{equation}
where $t=[1,1,1,\cdots]^T$. Using perturbation theory,
\begin{equation}
W^{(0)}y_k^{(0)}=\lambda_k^{(0)}y_k^{(0)},
\end{equation}
it is clear that 
\begin{equation}
\lambda_0=\Theta(1),\quad y_0^{(0)}\propto P_0\mathcal{C}_0^{1/2}t\propto z_0.
\end{equation}
Since we have $r_i=D^{-1/2}y_i$, we find that 
\begin{equation}
r_0^{(0)}\propto P_0t.
\end{equation}
Obviously, we have
\begin{equation}
\lambda_{i\neq 0}^{(0)}=0.
\end{equation}
Furthermore, since $W^{(0)}$ is rank one, we  will find the proper first order approximation to eigenvector $y_{i\geq 1}^{(0)}$ as a linear combination of $z_{i\geq 1}$ in the higher order perturbation discussion.

With this choice of $z_k$, we can continue the calculation for $\Theta(\alpha)$ terms.
\begin{equation}
W^{(1)}y_k^{(0)}+W^{(0)}y_k^{(1)}=\lambda_k^{(1)}y_k^{(0)},\quad k\geq 1.
\end{equation}
If we apply $(z_{j\geq 1})^T$ from the left and insert $y_k^{(0)}=\sum_{m=1}^\infty a_{km}^{(0)}z_m$, we get
\begin{equation}
\sum_{m=1}^\infty\left((z_j)^TW^{(1)}z_m-\delta_{mj}\lambda_k^{(1)}\right)a_{km}^{(0)}=0,
\end{equation}
where we have used the fact that $W^{(0)}z_{i\geq 1}=0$, which is because of our construction in Eq.~\ref{y0} and $W^{(0)}\sim P_0\mathcal{C}_0^{1/2}tt^T\mathcal{C}_0^{1/2}P_0$. Since $W^{(1)}$ has terms proportional to
\begin{equation}
P_0\mathcal{C}_0^{1/2}tt^T\mathcal{C}_1\mathcal{C}_0^{-1/2}P_0, \quad P_0\mathcal{C}_0^{-1/2}\mathcal{C}_1tt^T\mathcal{C}_0^{1/2}P_0,\quad P_0\mathcal{C}_0^{1/2}tt^T\mathcal{C}_2\mathcal{C}_0^{-1/2}P_1, \quad P_1\mathcal{C}_0^{-1/2}\mathcal{C}_2tt^T\mathcal{C}_0^{1/2}P_0,
\end{equation}
and $z_{k\geq 1}$ to orthogonal to $y_0^{(0)}\propto P_0\mathcal{C}_0^{1/2}t$, we have $(z_j)^TW^{(1)}z_m=0$ for $j,m\geq 1$. 
\begin{equation}
\lambda^{(1)}_k=0\quad \text{for} \quad k\geq 1.
\end{equation}


We then consider the $\Theta(\alpha^2)$ terms
\begin{equation}
W^{(2)}y_k^{(0)}+W^{(1)}y_k^{(1)}+W^{(0)}y_k^{(2)}=\lambda_k^{(2)}y_k^{(0)},\quad k\geq 1.
\end{equation}
Apply $z_{j\geq 1}$ from left,
\begin{equation}
\sum_{m=1}^\infty\left((z_j)^TW^{(2)}z_m-\delta_{mj}\lambda_k^{(2)}\right)a_{km}^{(0)}=0,
\end{equation}
where we have used the fact that $z_{i\geq 1}W^{(0)}z_{j\geq 1}=z_{i\geq 1}W^{(1)}z_{j\geq 1}=0$ as argued above. Since $W^{(2)}$ has terms proportional to
\begin{equation}
P_0\mathcal{C}_1\mathcal{C}_0^{-1/2}tt^T\mathcal{C}_1\mathcal{C}_0^{-1/2}P_0, \quad P_1\mathcal{C}_2^{1/2}tt^T\mathcal{C}_2^{1/2}P_1, \quad P_0\mathcal{C}_1\mathcal{C}_0^{-1/2}tt^T\mathcal{C}_2^{1/2}P_1, \quad P_0\mathcal{C}_1\mathcal{C}_0^{-1/2}tt^T\mathcal{C}_2^{1/2}P_1,
\end{equation}
This shows 
\begin{equation}
\lambda_{1,2}^{(2)}\neq 0,\quad \lambda_{k\geq 3}^{(2)}= 0,
\end{equation}
and also provides the leading order approximation to $y_{1,2}^{(0)}$ as a linear combination of $z_{1,2}$. By the orthogonality condition, we must have $y_{i\geq 3}^{(0)} \perp z_{j\leq 2}$.


We can then consider the $\Theta(\alpha^3)$ terms 
\begin{equation}
W^{(3)}y_k^{(0)}+W^{(2)}y_k^{(1)}+W^{(1)}y_k^{(2)}+W^{(0)}y_k^{(3)}=\lambda_k^{(3)}y_k^{(0)},\quad k\geq 3.
\end{equation}
We apply $(z_{j\geq 3})^T$ from the left
\begin{equation}
\sum_{m=3}^\infty a_{km}^{(0)}(z_j)^TW^{(3)}z_m+\sum_{m=1}^\infty a_{km}^{(1)}(z_j)^TW^{(2)}z_m=\lambda_k^{(3)}a_{kj}^{(0)}.
\end{equation}
Analyzing each terms of $W^{(2)}$ and $W^{(3)}$, we conclude the left hand side vanishes and  this shows
\begin{equation}
\lambda_{k\geq 3}^{(3)}=0.
\end{equation}
We demonstrate that $\beta_k^2$ acts as the threshold for the stepwise increase in the total REC, $C_T$, as illustrated in Fig.~\ref{step_C_T} for different values of $\alpha$ separately. For direct imaging, $C_T$ increases by 1 at each step, while for the SPADE method, $C_T$ increases by 2 per step, as expected. Moreover, as $\alpha$ decreases, the plateau regions become more pronounced.

\begin{figure}[!tb]
\begin{center}
\includegraphics[width=1\columnwidth]{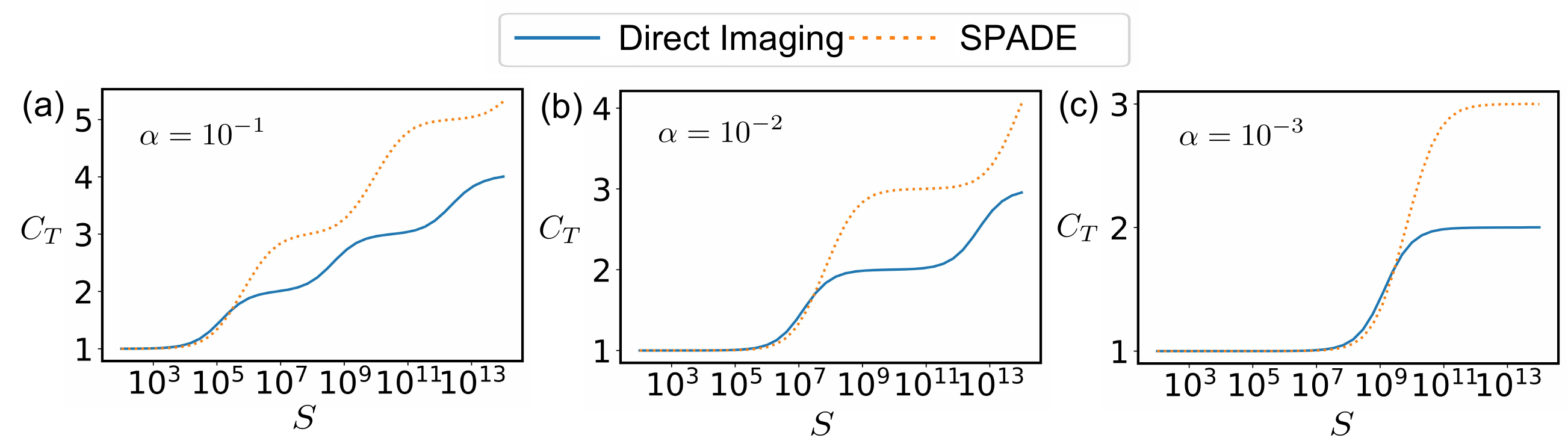}
\caption{Total REC $C_T$ for direct imaging and the SPADE method as a function of $S$, when imaging one generally distributed compact source with different $\alpha$.} 
\label{step_C_T}
\end{center}
\end{figure}

\yk{
\subsection{Numerical results for eigenvalues $\beta_k^2$}\label{SI:Numerical results for eigenvalues}

\begin{figure}[!tb]
\begin{center}
\includegraphics[width=0.4\columnwidth]{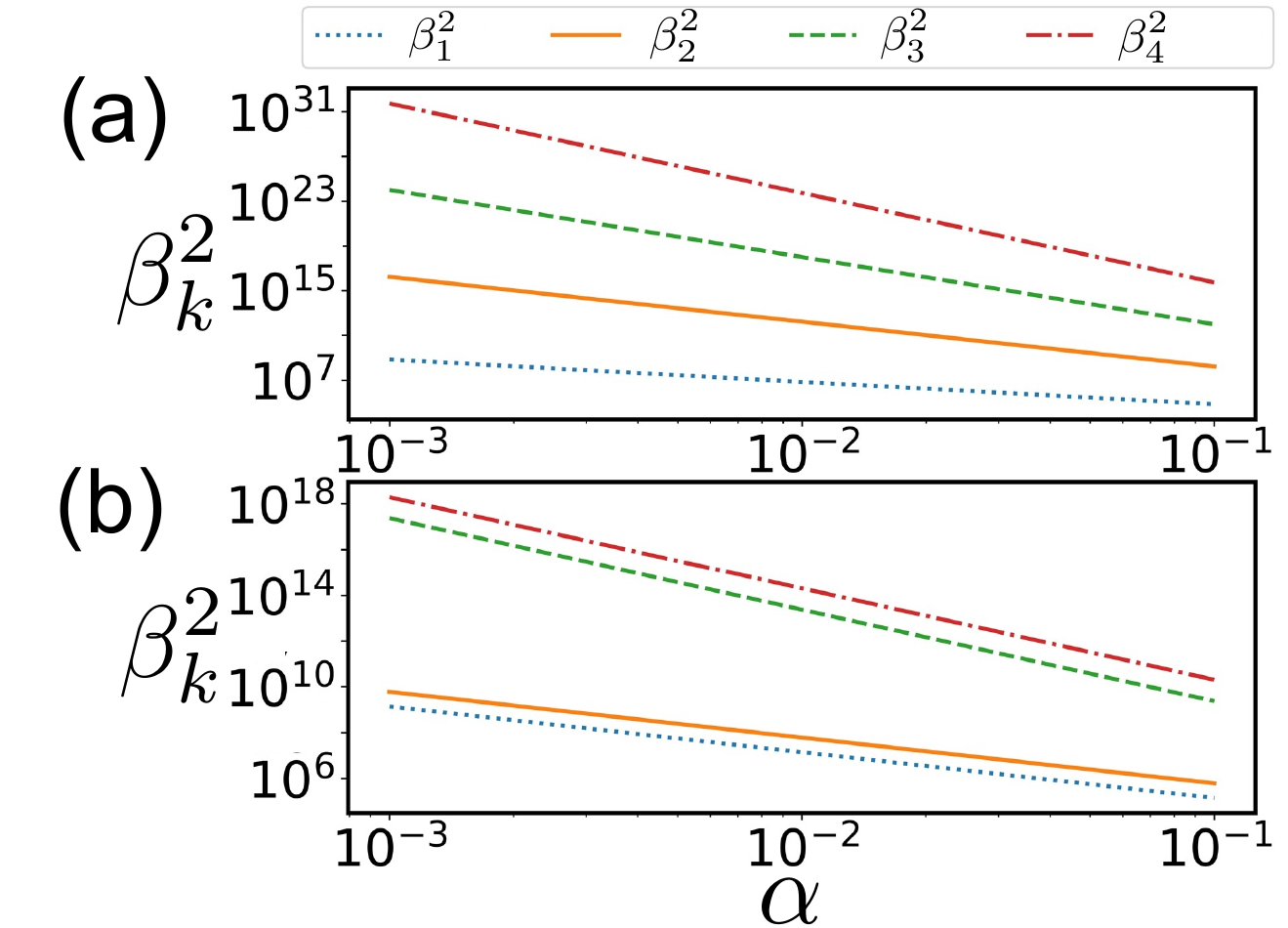}
\caption{ Scaling of the  $\beta_k^2$ as a function of $\alpha$ for (a) direct imaging and (b) SPADE method, when imaging a generally distributed compact source. $\beta_0^2 = 0$ is omitted from the plot. Note that $\beta_{3,4}^2$ in Fig.~\ref{CT_combined}(a) is interpolated using values calculated for cases where $\alpha > 10^{-1}$ that are not explicitly shown in the plot.
}
\label{CT_combined_betak}
\end{center}
\end{figure}

In the previous subsections, we provided a rigorous analytical proof based on perturbation theory that applies to any prior distribution $p(\vec{x})$ and PSF. For direct imaging, we showed that $\beta_0^2 = 0$, $\beta_1^2 = \Theta(\alpha^{-2})$, and $\beta_{k \geq 2}^2 = \omega(\alpha^{-2})$. More generally, we expect the scaling $\beta_k^2 = \Theta(\alpha^{-2k})$ to hold for higher-index eigenvalues under any prior. For the superresolution case, we demonstrated that $\beta_0^2 = 0$, $\beta_1^2 = \Theta(\alpha^{-2})$, $\beta_2^2 = \Theta(\alpha^{-2})$, and $\beta_{k \geq 3}^2 = \omega(\alpha^{-2})$, and we conjecture that the scaling $\beta_k^2 = \Theta(\alpha^{-2\lceil k/2 \rceil})$ persists for larger $k$.

The scalings of the first five eigenvalues for the direct imaging case are confirmed by numerical calculations, as shown in Fig. \ref{CT_combined_betak} (a). In the superresolution case, Fig. \ref{CT_combined_betak} (b) shows that the first five eigenvalues, computed numerically for a random prior, form pairs $\{\beta^2_{2k-1}, \beta^2_{2k}\}_{k \geq 1}$ with identical scaling starting from $k=1$, also consistent with our analytical proof. Values of $\beta_k^2$ exceeding $\sim 10^{15}$ in Fig.~\ref{CT_combined_betak} are interpolated using a straight-line fit at larger $\alpha$ to mitigate numerical precision loss; this interpolation is applied to all $\beta_k^2$--$\alpha$ plots.

All numerical calculations in this work assume a Gaussian PSF, $\psi(u) = \exp(-u^2/4\sigma^2)/(2\pi\sigma^2)^{1/4}$, although our method is applicable to arbitrary PSFs. The prior distribution $p(\vec{x})$ is constructed by randomly generating a set of images and assigning them equal probability, so that $p(\vec{x})$ corresponds to the empirical distribution of the resulting moment vectors. Further details are provided in Sec.~\ref{SI:numerics} of the Supplemental Material. For illustration, we show results for a single instance of the randomly generated prior (as is the case for all figures below that involve a randomly chosen prior).

Note that the above analysis is limited to the first few eigenvalues $\beta_k^2$, although we expect the same scaling behavior to extend to higher indices $k$. To avoid potential caveats and ensure clarity, Table~\ref{Table:beta_k} summarizes the scaling behavior that has been explicitly calculated and confirmed through both analytical and numerical methods.

\begin{table}[h]
\centering
\begin{tabular}{|c|c|c|c|c|}
\hline
\textbf{Imaging Type} & \textbf{Eigenvalue Index $k$} & \textbf{Approach} & \textbf{PSF Used} & \textbf{Prior Used} \\
\hline
\multicolumn{5}{|c|}{\textbf{Analytical Results}} \\
\hline
\multirow{4}{*}{Direct Imaging} 
& 0 & $ 0$ & Arbitrary & Arbitrary \\
& 1 & $\Theta(\alpha^{-2})$ & Arbitrary & Arbitrary \\
& 2 & $\omega(\alpha^{-2})$ & Arbitrary & Arbitrary \\
& 3 & $\omega(\alpha^{-2})$ & Arbitrary & Arbitrary \\
\hline
\multirow{5}{*}{Superresolution} 
& 0 & $ 0$ & Arbitrary & Arbitrary \\
& 1 & $\Theta(\alpha^{-2})$ & Arbitrary & Arbitrary \\
& 2 & $\Theta(\alpha^{-2})$ & Arbitrary & Arbitrary \\
& 3 & $\omega(\alpha^{-2})$ & Arbitrary & Arbitrary \\
& 4 & $\omega(\alpha^{-2})$ & Arbitrary & Arbitrary \\
\hline
\multicolumn{5}{|c|}{\textbf{Numerical Results}} \\
\hline
\multirow{4}{*}{Direct Imaging} 
& 0 & $ 0$ & Gaussian & Random empirical \\
& 1 & $\Theta(\alpha^{-2})$ & Gaussian & Random empirical \\
& 2 & $\Theta(\alpha^{-4})$ & Gaussian & Random empirical \\
& 3 & $\Theta(\alpha^{-6})$ & Gaussian & Random empirical \\
& 4 & $\Theta(\alpha^{-8})$ & Gaussian & Random empirical \\
\hline
\multirow{5}{*}{Superresolution} 
& 0 & $ 0$ & Gaussian & Random empirical \\
& 1 & $\Theta(\alpha^{-2})$ & Gaussian & Random empirical \\
& 2 & $\Theta(\alpha^{-2})$ & Gaussian & Random empirical \\
& 3 & $\Theta(\alpha^{-4})$ & Gaussian & Random empirical \\
& 4 & $\Theta(\alpha^{-4})$ & Gaussian & Random empirical \\
\hline
\end{tabular}
\caption{\yk{Summary of the asymptotic scaling behavior of eigenvalues $\beta_k^2$ from both analytical and numerical approaches. Analytical results apply to arbitrary PSFs and priors, while numerical results are obtained for a Gaussian PSF and a randomly generated empirical prior.}}\label{Table:beta_k}
\end{table}

}

\subsection{Deviation of the approximated eigenbasis from the actual eigenbasis}\label{SI:deviation_eigenbasis}

\begin{figure}[!tb]
\begin{center}
\includegraphics[width=1\columnwidth]{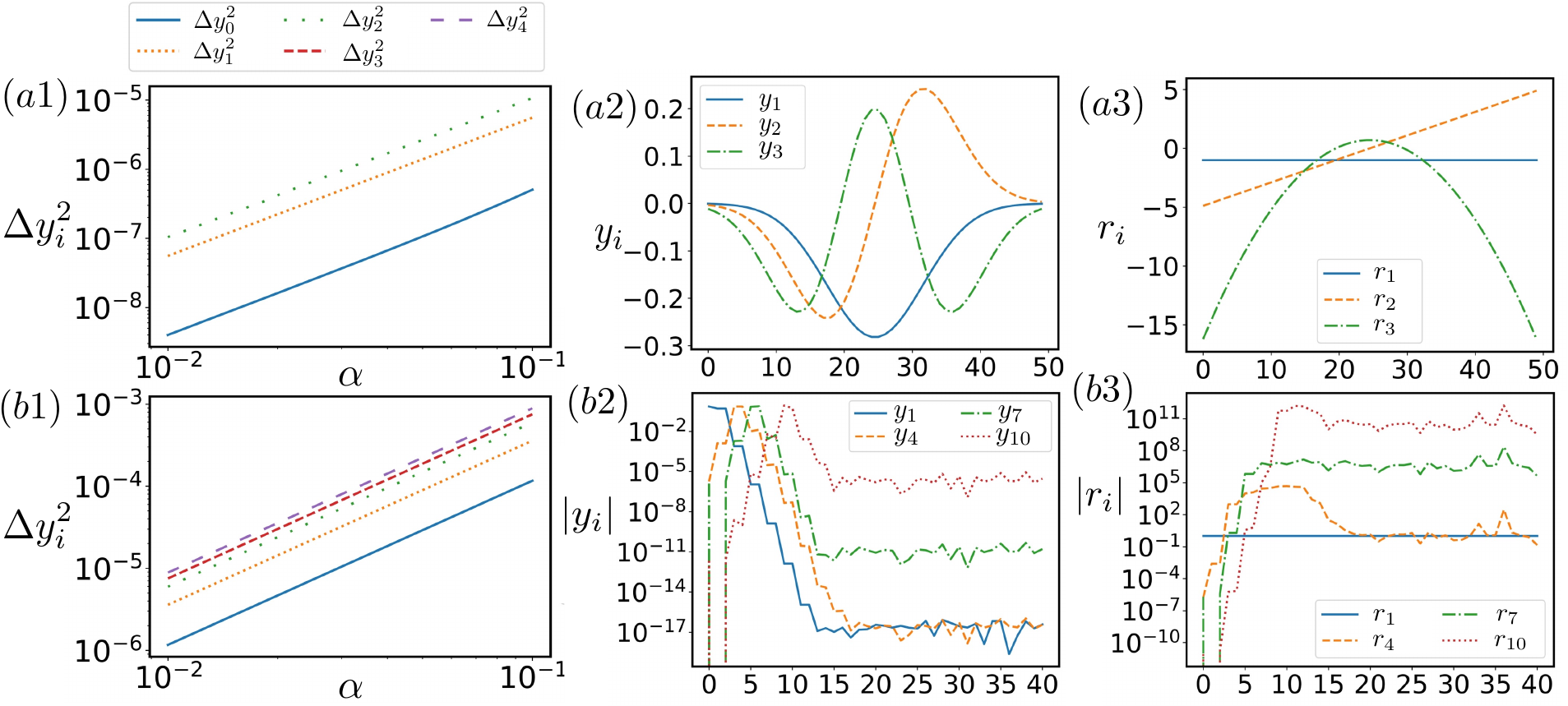}
\caption{\yk{We plot the deviation of the approximated eigenvector $\hat{y}_k$ from the actual eigenbasis $y_k$ for imaging a single compact source using (a1) direct imaging and (b1) the SPADE method. The deviation is defined as the squared norm of the difference between two normalized vectors, $\Delta y_k^2 = |\hat{y}_k - y_k|^2$, for the eigenvalue problem $W y_k = \lambda_k y_k$, where $W = D^{-1/2} G D^{-1/2}$. We also show the actual eigenbasis $y_k$ in (a2) for direct imaging and (b2) for the SPADE method. In (b2), we plot the absolute value to display it on a logarithmic scale. Additionally, we present the vectors $r_k$ in (a3) for direct imaging and (b3) for the SPADE method. The $x$-axis in (a2), (b2), (a3), (b3) represents the different measurement outcomes, and we set $\alpha = 10^{-2}$ for these figures.
 } }
\label{deviation}
\end{center}
\end{figure}


In this section, we want to verify the deviation of the approximated  eigenbasis $\hat{y}_k$ from the actual eigenbasis $y_k$ vanishes as $\alpha\rightarrow0$. In the direct imaging case, we predicted above that $\{z_n\}_{n=0,1,2,\cdots}$, which is the orthonormal basis constructed from the set of vectors $\{D^{-1/2} C_n\}_{n=0,1,2,\cdots}$ using the Gram-Schmidt procedure, are the approximated eigenvector $\hat{y}_k=z_k$ for actual eigenvectors $y_k$, which is confirmed as shown in Fig.~\ref{deviation}(a1). For the SPADE method, we predicted above that the approximated eigenvector $\hat{y}_0=z_0$, each $\hat{y}_{k \geq 1}$ is the linear combination of $z_{2\lceil k/2\rceil-1}$ and $z_{2\lceil k/2\rceil}$  in Eq.~\ref{y0}.
The coefficients in the linear combination of the two bases depend on the specific values of $g_{ij}$ and $d_i$. We choose $\hat{y}_0=z_0$ and employ a least squares approach to fit each $\hat{y}_{k \geq 1}$ with the two bases $z_{2\lceil k/2\rceil-1}$ and $z_{2\lceil k/2\rceil}$  in Eq.~\ref{y0}.  The deviation is illustrated in Fig.~\ref{deviation}(b1). In both cases, as $\alpha \rightarrow 0$, the actual eigenbasis $y_k$ increasingly aligns with the eigenvectors $\hat{y}_k$ described above. 
For the numerical calculations, we select the PSF as $\psi(u) = {\exp(-u^2/4\sigma^2)}/{(2\pi\sigma^2)^{1/4}}$. The coefficients $d$ and $g$ are chosen by randomly generating a set of images and assuming that these images appear with equal probability, which then can be used to $d$ and $g$ with the corresponding moments vector of each image as detailed in Sec.~\ref{SI:numerics}.

\yk{We also plot the actual eigenbasis $y_k$ and $r_k$ in Fig.~\ref{deviation}(a2)(a3) for the direct imaging method and in Fig.~\ref{deviation}(b2)(b3) for the SPADE method. The $x$-axis corresponds to different measurement outcomes. For direct imaging [Fig.~\ref{deviation}(a2)(a3)], the measurement outcome corresponds to detecting a single photon at each pixel on the image plane. For the SPADE method [Fig.~\ref{deviation}(b2)(b3)], we plot the absolute value $|y_k|,|r_k|$, with the outcomes labeled as $P_0, P_{1+}, P_{1-}, P_{2+}, P_{2-}, \dots$. It is evident that $|y_k|$ is concentrated on just a few outcomes and $|r_k|$ is almost vanishing for the first few components. Since for the SPADE method,
$P_{n\pm} = c^{\pm}_{n,2n} x_{2n} \alpha^{2n} + c^{\pm}_{n,2n+1} x_{2n+1} \alpha^{2n+1} + o(\alpha^{2n+1})$,
it follows that the eigentasks are dominated by only a few moments. Note that in the figure we plot $y_k$, the eigenvector of $D^{-1/2} G D^{-1/2}$ defined by $D^{-1/2} G D^{-1/2} y_k = \lambda_k y_k$, which is related to $r_k$ by $y_k = D^{1/2} r_k$. We use $y_k$ because it is normalized without the weight matrix, i.e., $y_k^\dagger y_k = 1$, whereas $r_k^\dagger D r_k = 1$. Using $y_k$ makes it easier to see which components become dominant. }

\section{Imaging general source extended outside the Rayleigh limit}\label{SI:general_source}

In the above discussion, we have consistently assumed the size of one compact source is much smaller than the Rayleigh limit of our lens. However, in practice, the source size can exceed the Rayleigh limit, and we still need to quantify imaging performance in such cases. In this section, we focus on quantifying imaging performance using the REC under these conditions.

\yk{

A natural question is what the total REC would be for a general imaging task using direct imaging and how to interpret its meaning. Intuitively, when imaging a source of size $L$ with a lens where the PSF has a width $\sigma$, the resulting image is a blurred version of the original, with the blur size being $\sigma$. The total REC represents the number of degrees of freedom we can extract from the measurement, and since the image is blurred, this is roughly of order $L/\sigma$; however, this relation is not exact, as it arises from heuristic reasoning rather than formal analysis. To confirm this, we perform a numerical calculation   as follows:
we randomly generate a general source of size $L$ consisting of $N_{\text{max}}$ discrete pixels,
with each pixel assigned a random value between 0 and 1. We then normalize the intensity distribution $I(u)$ such that $\sum_u I(u) = 1$. By repeatedly generating the intensity distribution $I(u)$ in this manner for $W$ times, we simulate the case where we do not have any prior knowledge of the source. For all the numerical calculation for imaging general source outside the Rayleigh limit, we use $N_{\text{max}}=200$, $W=200$. Assuming the PSF is given by $\psi(u) = (2\pi\sigma^2)^{-1/4} \exp(-u^2/(4\sigma^2))$.

We plot the total REC as a function of sample number $S$ for different values of $L$, as shown in Fig.~\ref{general_source_S_log}(a). As expected, the total $C_T$ soon reaches a (slowly increasing) plateau as $S \rightarrow \infty$. 
We also present in Fig.~\ref{general_source_S_log}(b) a plot of $C_T$ as a function of $S$ on a logarithmic scale. As discussed in the main text, when moving beyond the regime of resolving multiple compact sources—where the sources are no longer confined to a few localized regions—the increase of $C_T$ no longer exhibits stepwise behavior.
Note that as $S \rightarrow \infty$, the total REC $C_T$ continues to increase slowly and eventually diverges. This implies that, with infinitely many photons, even direct imaging can resolve the image to an arbitrarily fine degree as a number of investigations over the last fifty years have showed \cite{goodman2005introduction,lucy1992statistical,wang2023fundamental}. However, to reach this infinite $C_T$, the required number of samples also diverges. In particular, when $S$ becomes sufficiently large, the increase of $C_T$ becomes very slow. This behavior is  illustrated in Fig.~\ref{general_source_S_log}(c), which shows that the slope $(\partial C_T) / (\partial \log S)$ reaches a peak at some sample size and then gradually decreases.
In this sense, $C_T$ exhibits a logarithmic-type growth with respect to $S$, characterized by a relatively small coefficient when $S$ becomes sufficiently large. Since $\log S$ is used as the $x$-axis in the slope plot, advancing by a fixed interval along the $\log S$ axis requires an exponentially larger increase in the number of samples. Meanwhile, $C_T$ increases only modestly over the same interval of $\log S$. Therefore, the curve approaches a slowly rising plateau.

}

\yk{We also demonstrate that the total REC $C_T$ increases approximately linearly with the source size $L$, as shown in Fig.~\ref{general_source_CT_L}(a) for different choices of $\sigma$. The ratio $C_T/(L/\sigma)$ is plotted in Fig.~\ref{general_source_CT_L}(b), which approaches 1 but shows a slight decreasing trend as $L$ increases in this model across various values of $\sigma$.
In Fig.~\ref{general_source_sigma}, we plot the total REC $C_T$ as a function of the PSF width $\sigma$ for different source sizes $L$. Our numerical results indicate that, for the model considered here, $C_T$ approximately follows a power-law scaling: $C_T \propto \sigma^{-0.9}$. We have thus observed that the total REC $C_T$ is approximately determined by the ratio $L/\sigma$, but this relation is not exact—rather, it provides a rough estimate that serves only as an intuitive argument rather than a rigorous treatment.}
All the above calculations are performed for a single instance of randomly generated prior information. However, in Fig.~\ref{general_source_robust}, we present results from three consecutive executions of the calculation with $N_{\text{max}} = 200$ and $W = 200$, where each execution randomly generates a different instance of prior information. Despite the variation in prior information, we observe negligible differences in the calculation results.

\begin{figure}[!tb]
\begin{center}
\includegraphics[width=1\columnwidth]{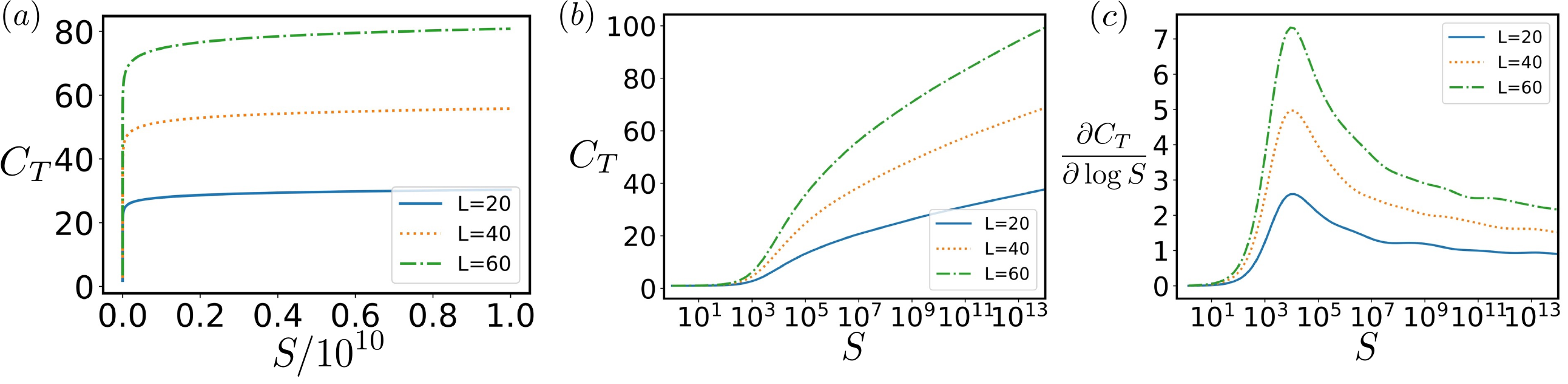}
\caption{\yk{(a) Total REC for direct imaging  as a function of $S$ for a general source with size $L=20,40,60$ respectively and width of PSF $\sigma=1$. (b) Total REC $C_T$ for direct imaging of general source extended outside the Rayleigh limit as a function of $S$ in log scale for source with size $L=20,40,60$. (c) Slope $(\partial C_T)/(\partial \log  S)$ as a function of $S$,  for source with size $L=20,40,60$. $\sigma=1$.}} 
\label{general_source_S_log}
\end{center}
\end{figure}

\begin{figure}[!tb]
\begin{center}
\includegraphics[width=0.8\columnwidth]{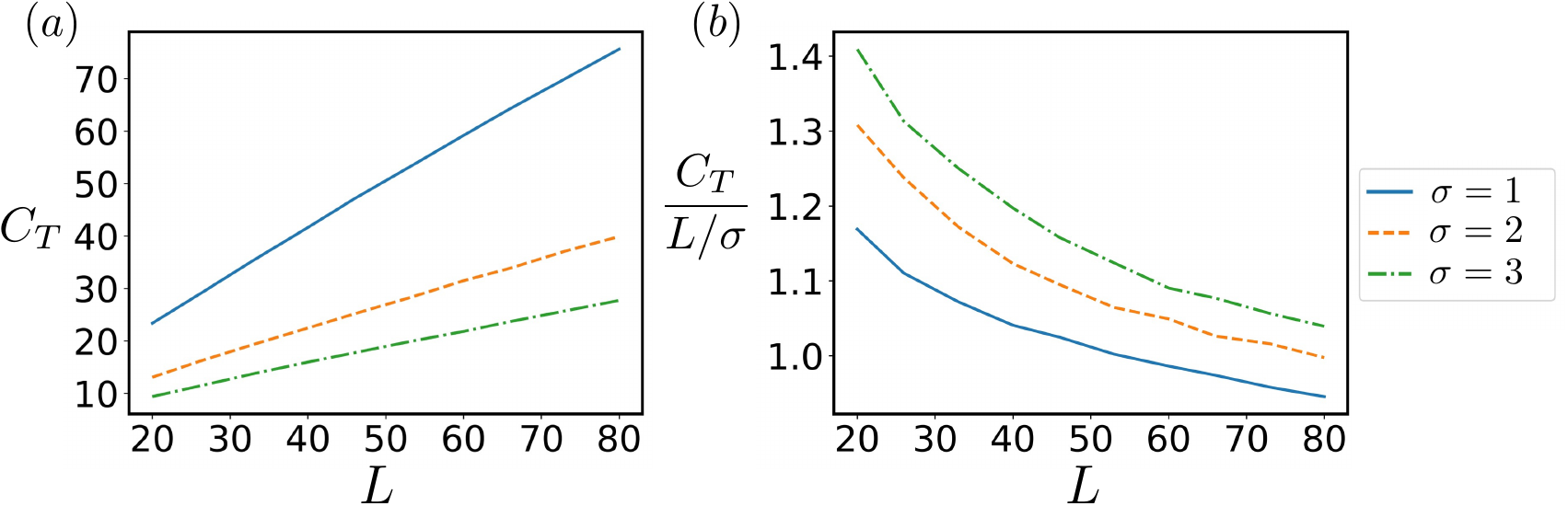}
\caption{\yk{(a) Total REC $C_T$ for direct imaging of a general source that extends outside the Rayleigh limit, plotted as a function of the source size $L$; (b) the ratio $C_T / (L/\sigma)$ as a function of $L$, for different values of the PSF width $\sigma =  1, 2,3$. In both panels, the total sample number is fixed at $S = 10^7$.} }
\label{general_source_CT_L}
\end{center}
\end{figure}

\begin{figure}[!tb]
\begin{center}
\includegraphics[width=0.4\columnwidth]{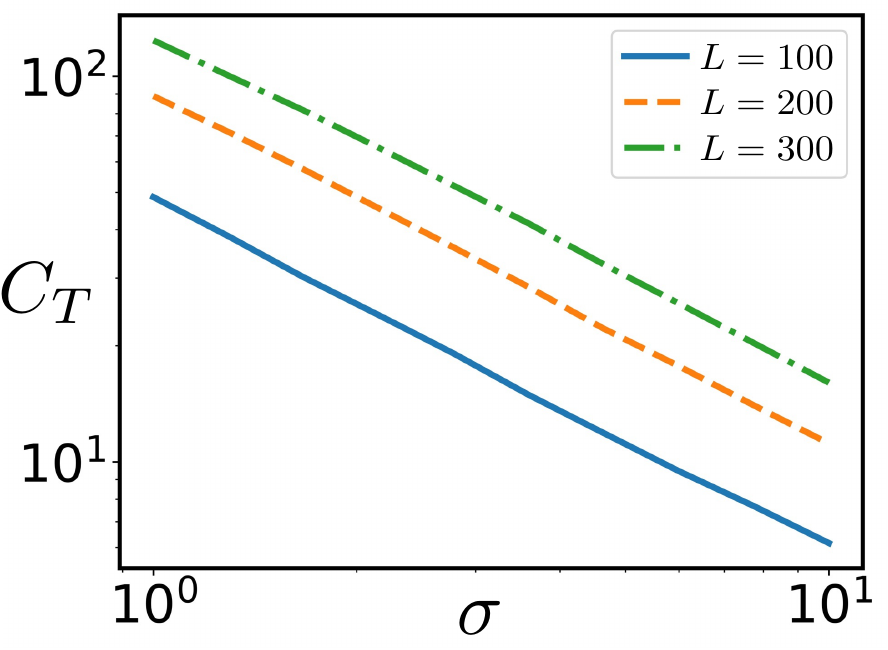}
\caption{\yk{Total REC $C_T$ for direct imaging of a general source extending beyond the Rayleigh limit, plotted as a function of the PSF width $\sigma$ for different source sizes $L=100, 200, 300$. Both axes use logarithmic scales. The sample number is $S = 10^4$.}} 
\label{general_source_sigma}
\end{center}
\end{figure}

\begin{figure}[!tb]
\begin{center}
\includegraphics[width=1\columnwidth]{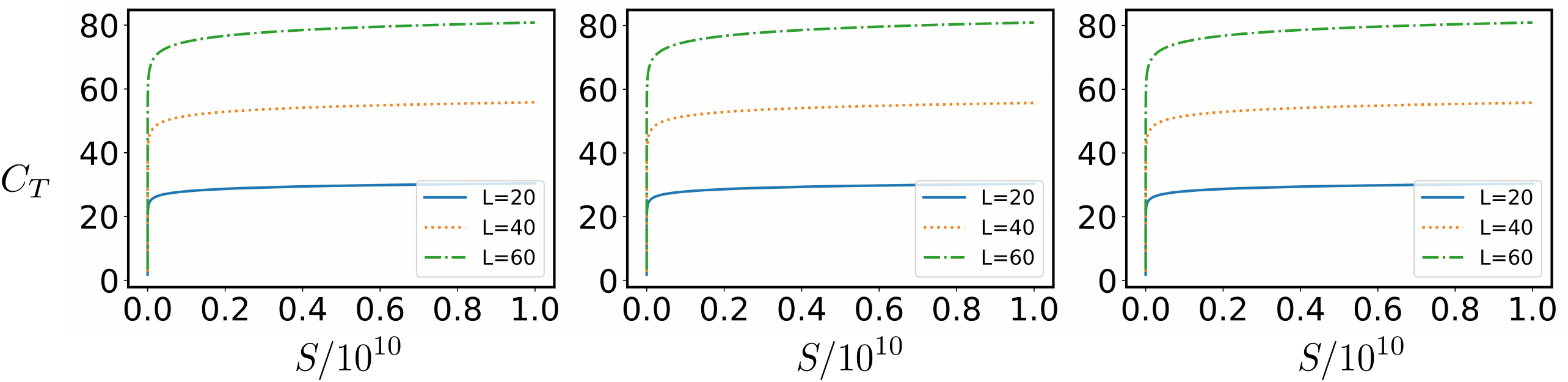}
\caption{Three consecutive executions of the calculation of total REC $C_T$ for direct imaging of general source extended outside the Rayleigh limit as a function of $S$ for source with size $L=20,40,60$. $\sigma=1$.} 
\label{general_source_robust}
\end{center}
\end{figure}

\section{Imaging multiple compact sources}\label{SI:multiple_compact_source}

The state from multiple compact sources is given by
\begin{equation}
\rho=\sum_{q=1}^Q\int du du_1du_2I_q(u)\psi(u-u_1)\ket{u_1}\bra{u_2}\psi^*(u-u_2),
\end{equation}
where $Q$ is the number of compact sources, $I_q(u)$ is the intensity distribution for $q$th compact source. We can expand the PSF near the centroid $u_q$ of $q$th source
\begin{equation}\begin{aligned}
\psi(u_1-u)&=\sum_{n=0}^\infty \left(\frac{\partial^n \psi(u_1-u)}{\partial u^n}\bigg|_{u=u_q}\frac{L_q^n}{n!}\right)\left(\frac{u-u_q}{L_q}\right)^n\\
&=\sum_{n=0}^\infty \psi^{(n)}_q(u_1)\left(\frac{u-u_q}{L_q}\right)^n,
\end{aligned}\end{equation}
where $L_q$ is the size of $q$th source.

\begin{equation}\begin{aligned}
&\rho=\sum_{q=1}^Q \sum_{m,n=0}^\infty x_{m+n,q}\ket{\psi_q^{(m)}}\bra{\psi_q^{(n)}},\\
&\ket{\psi_q^{(m)}}=\int du \psi_q^{(m)}(u)\ket{u}, 
\end{aligned}\end{equation}
where $x_{n,q}=\int du I_q(u)\left(\frac{u-u_q}{L_q}\right)^n $ is the $n$th moment for the $q$th source.

\subsection{Direct imaging}\label{SI:multiple_compact_direct}
If we have $Q$ compact sources, direct imaging will have the probability distribution in the following form
\begin{equation}\label{multiple_direct_P}
P_n=\sum_{q=1}^Q\sum_{m=0}^\infty c_{nmq}x_{mq}\alpha_q^m.
\end{equation}
Again, we measure a discretized version of direct imaging for both analytical and numerical discussions $M_n = \int_{X_n - l/2}^{X_n + l/2} du \ket{u}\bra{u}$. Notice that $\alpha_q= L_q/\sigma$ can be different. But in the following, we will redefine the coefficients $c_{nmq}$ involves the difference between $L_q$ for different sources and have $\alpha_q=\alpha$. We will numerically explore the case where $\alpha_q$ differs significantly and cannot be incorporated into the coefficients $c_{nmq}$ below.

\begin{equation}\begin{aligned}
&D=\sum_{n=0}^\infty D^{(n)}=\sum_{n=0}^\infty \alpha^n\sum_{q=1}^Q d_{nq} \mathcal{C}_{nq},\quad \mathcal{C}_{nq}=\text{diag}(C_{nq}),\quad C_{nq}= [c_{0nq},c_{1nq},c_{2nq},\cdots\cdots]^T, \\
&G=\sum_{n=0}^\infty G^{(n)}=\sum_{n=0}^\infty\alpha^n\sum_{n_1+n_2=n}\sum_{q_1,q_2=1}^Qg_{n_1q_1,n_2q_2}  C_{n_1q_1} C_{n_2q_2}^T,\\
&g_{n_1q_1,n_2q_2}=\int d\vec{x}p(\vec{x})x_{n_1q_1}x_{n_2q_2}, \quad d_{nq}=\int d\vec{x}p(\vec{x})x_{nq}.\\
\end{aligned}\end{equation}
where $\text{diag}(C_{nq})$ denotes a diagonal matrix with the diagonal elements being the elements of $C_{nq}$. If we choose a set of basis $z_n$ as the Gram-Schmidt orthogonalization of ${D_0^{-1/2}C_{mq}}$, we can then express the matrix $W = D^{-1/2} G D^{-1/2}$ in the basis of $z_n$ and divide it into $Q \times Q$ blocks. In this block structure, the $(m,n)$th block scales as $\Theta(\alpha^{m+n})$. We can then examine the structure of the characteristic polynomial $|\lambda I - W| = \sum_{k=0}^M (-1)^k e_k \lambda^{M-k}$, where the coefficients $e_k$ now take the form of
\begin{equation}\begin{aligned}
&e_1=\Theta(\alpha^0),\cdots, e_Q=\Theta(\alpha^0), e_{Q+1}=\Theta(\alpha^{0+2}),\cdots, e_{2Q}=\Theta(\alpha^{0+2}), e_{2Q+1}=\Theta(\alpha^{0+2+4}),\cdots, e_{3Q}=\Theta(\alpha^{0+2+4}),\cdots\cdots 
\end{aligned}
\end{equation}
So, the scaling of eigenvalue is
\begin{equation}\begin{aligned}
&\beta_{0\leq i\leq Q-1}^2=\Theta(1),\quad\beta_{Q\leq i\leq 2Q-1}^2=\Theta(\alpha^{-2}),\quad\beta_{2Q\leq i\leq 3Q-1}^2=\Theta(\alpha^{-4}),\cdots
\end{aligned}\end{equation}

\subsection{Separate SPADE method }\label{SI:multiple_compact_superresolution}

If we naively apply the SPADE method \cite{zhou2019modern} to each compact source, which will be referred as the separate SPADE method. In other word, we use the POVM
\begin{equation}\label{eq:POVM_separate}
\left\{\frac{1}{2Q}\ket{\phi_{q,m,\pm}}\bra{\phi_{q,m,\pm}},\frac{1}{2Q}\ket{b_{q,0}}\bra{b_{q,0}}\right\}_{q=1,2,\cdots,Q,\, m=0,1,2,\cdots,\infty},     
\end{equation}
where $\{\ket{b_{q,m}}\}_{m=0,1,2,\cdots,\infty}$ is given by the Gram-Schmidt orthogonalization of set $\{\ket{\psi_q^{(m)}}\}_{m=0,1,2,\cdots,\infty}$, and $\ket{\phi_{q,m,\pm}}=\frac{1}{\sqrt{2}}(\ket{b_{q,m}}\pm\ket{b_{q,m+1}})$. We add the normalization factor $1/Q$ to our POVM.

An immediate issue arises: nearby compact sources can disrupt the construction of separate SPADE method. This occurs because the designed measurement relies on canceling out lower-order terms of $O(\alpha^n)$ through careful design. However, if other compact sources are too close, they may contribute terms that are much larger than $O(\alpha^n)$, undermining the effectiveness of the superresolution technique.

In general, when all compact sources are sufficiently far apart, superresolution can still be achieved. However, when the sources are too close together, the probability distribution of separate SPADE method will resemble that of direct imaging, as described in Eq.~\ref{multiple_direct_P}.
To be more specific, if we choose the PSF as $\psi(u) = {\exp(-u^2/4\sigma^2)}/{(2\pi\sigma^2)^{1/4}}$, the threshold for achieving superresolution is determined by whether $\alpha^n \gg \exp(-L^2/\sigma^2)$, where $n$ is determined by the order of the measured moments, and $\exp(-L^2/\sigma^2)$ represents the contribution from nearby compact sources.


\begin{figure}[!tb]
\begin{center}
\includegraphics[width=1\columnwidth]{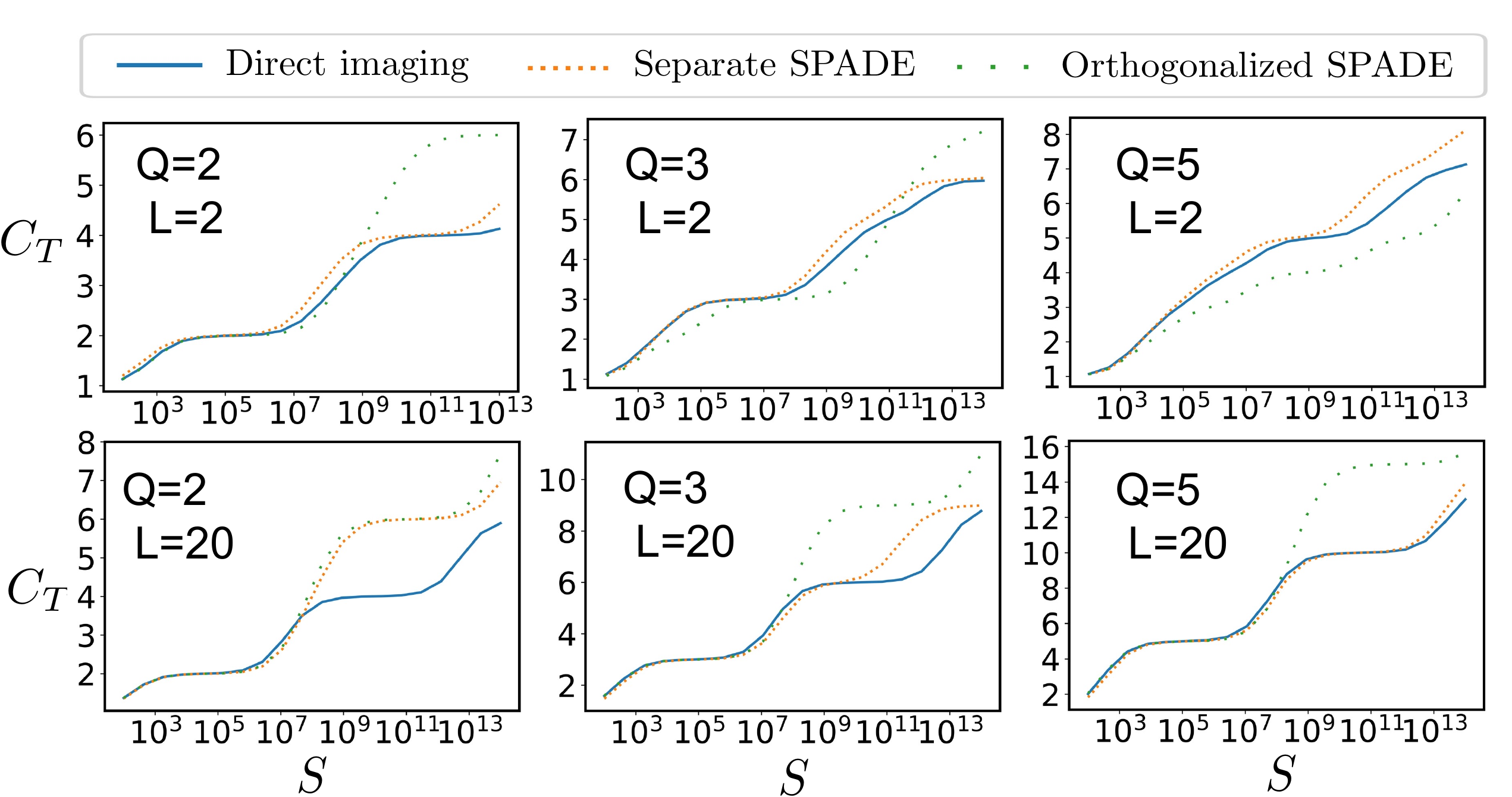}
\caption{Total REC $C_T$ of imaging $Q$ compact sources for direct imaging, separate SPADE method and orthogonalized SPADE method as a function of $S$. $\alpha=10^{-2}$, $\sigma=1$. We consider two different  $L$ and different number of compact source $Q$. 
} 
\label{new_method_Q}
\end{center}
\end{figure}

\begin{figure}[!tb]
\begin{center}
\includegraphics[width=1\columnwidth]{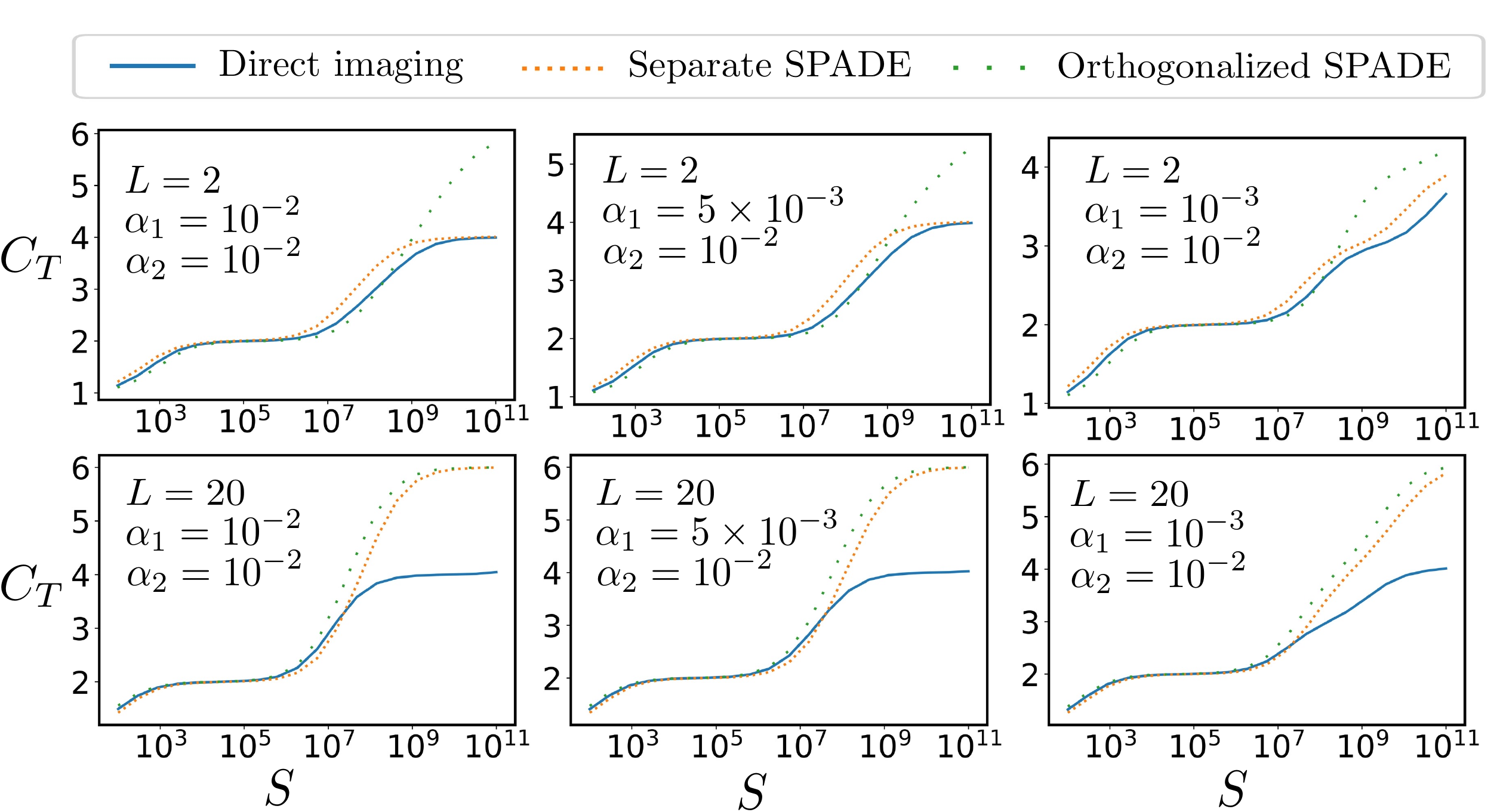}
\caption{Total REC $C_T$ for direct imaging, separate SPADE method and orthogonalized SPADE method as a function of $S$. Assume we are imaging two compact sources $Q=2$, $\sigma=1$. We consider two different  $L$ and assume the size of the two compact source is $\alpha_1,\alpha_2$, which can be different. } 
\label{new_method_different_alpha}
\end{center}
\end{figure}

\subsection{Orthogonalized SPADE method}\label{SI:new_method_superresolution}

With the construction of orthonormal basis $\ket{b_j^{(l)}}$ from the Gram-Schmidt procedure such that 
\begin{equation}\label{eq:bjl_SI}
a_{mk,lj}=\bra{\psi^{(m)}_k}\ket{b^{(l)}_j}\left\{
\begin{array}{ccc}
=0 & \quad   &  m\leq l-1\\
=0 &\quad   & m=l \,\,\& \,\, k\leq j-1\\
\neq 0  & \quad   &  \text{otherwise}
\end{array}\right.
\end{equation}
Choose POVM as 
\begin{equation}\label{eq:POVM_orthogonalized}
\left\{\frac{1}{2}\ket{\phi_{j,\pm}^{(l)}}\bra{\phi_{j,\pm}^{(l)}},\frac{1}{2}\ket{b_{j}^{(0)}}\bra{b_{j}^{(0)}}\right\}_{j=1,2,\cdots,Q,\, l=0,1,2,\cdots,\infty},   
\end{equation}
where
\begin{equation}
\ket{\phi_{j\pm}^{(l)}}=\frac{1}{\sqrt{2}}\left(\ket{b_j^{(l)}}\pm\ket{b_j^{(l+1)}}\right),\quad j=1,2,3,\cdots,Q,\quad l=0,1,2,\cdots,\infty
\end{equation}
Then, the probability distribution
\begin{align}\label{eq:oSPADE_P}
&P_{0,j}=\frac{1}{2}\bra{b_{j}^{(0)}}\rho\ket{b_{j}^{(0)}}=\sum_{q=1}^Q\sum_{n=0}^\infty c_{jnq} x_{nq}\alpha^n,\\
&P_{l,j,\pm}=\frac{1}{2}\bra{\phi_{j\pm}^{(l)}}\rho\ket{\phi_{j\pm}^{(l)}}=\sum_{q=1}^Q\sum_{n=2l}^\infty c_{ljnq}^\pm x_{nq}\alpha^n,\\
&c_{ljnq}^\pm=\frac{1}{4}\sum_{m=0}^n 
\left(a_{pq,lj}a^*_{mq,lj}+a_{pq,(l+1)j}a^*_{mq,(l+1)j}\pm a_{pq,(l+1)j}a^*_{mq,lj}\pm a_{pq, lj}a^*_{mq,(l+1)j}\right)/\alpha^n,\quad p=n-m.
\end{align}

We set up the REC calculation as 
\begin{align}
&D=\sum_{n=0}^\infty D^{(n)}=\sum_{n=0}^\infty \alpha^n\sum_{q=1}^Q d_{nq} \mathcal{C}_{nq},\quad \mathcal{C}_{nq}=\Pi_{nq}\,\mathcal{C}_{nq}^0\,\Pi_{nq},\\
&\mathcal{C}_{nq}^0=\text{diag}(C_{nq}^0),\quad d_{nq}=\int d\vec{x}p(\vec{x})x_{nq},\\
&G=\sum_{n=0}^\infty G^{(n)}=\sum_{n=0}^\infty\alpha^n\sum_{n_1+n_2=n}\sum_{q_1,q_2=1}^Qg_{n_1q_1,n_2q_2}  C_{n_1q_1} C_{n_2q_2}^T, \\
&C_{nq}=\Pi_{nq}C^0_{nq},\quad g_{n_1q_1,n_2q_2}=\int d\vec{x}p(\vec{x})x_{n_1q_1}x_{n_2q_2},\\
&C^0_{nq}= [c_{1nq},c_{2nq},\cdots,c_{Qnq},c_{01nq}^+,c_{01nq}^-,c_{02nq}^+,c_{02nq}^-,\cdots,c_{0Qnq}^+,c_{0Qnq}^-, c_{11nq}^+,c_{11nq}^-,c_{12nq}^+,c_{12nq}^-,\cdots,c_{1Qnq}^+,c_{1Qnq}^-, \cdots\cdots]^T, \\
&\Pi_{nq}=\left[\begin{matrix}
I_{2\lfloor n/2\rfloor+1}\otimes I_Q & 0 & 0\\
0 & I_2\otimes Q_q & 0\\
0 & 0 & 0
\end{matrix}\right], \quad 
Q_q=\left[\begin{matrix}
I_q & 0\\
0 &0\\
\end{matrix}\right].
\end{align}
We can again give an intuitive argument as in Sec.~\ref{SI:one_source_direct} and Sec.~\ref{SI:one_source_superresolution} about the scaling of eigenvalues $\beta_k^2$ based on characteristic polynomial. 
Define
$P_{m}=\Pi_{2m,Q}-\Pi_{2(m-1),Q}$, $\Pi_{m<0,k}=0$. In the $2Q$ dimensional subspace in the support of $P_m$, we have vectors $P_mC_{2m,q}, P_mC_{2m+1,q}$, $q=1,2,3,\cdots,Q$, $m=0,1,2,\cdots,\infty$. We will define the  set of orthonormal basis $z_{2m+1,q},z_{2m+2,q}$ for convenience  such that $z_{2m+1,q}$ is the Gram-Schmidt orthorgonalization of all of the $\{P_mC_{2m,q}\}_{q=1,2,3,\cdots,Q}$ and $z_{2m+2,q}$ is the Gram-Schmidt orthorgonalization of all of the $\{P_mC_{2m+1,q}\}_{q=1,2,3,\cdots,Q}$ while keeping orthogonal to $\{P_mC_{2m,q}\}_{q=1,2,3,\cdots,Q}$. Then, it should be clear that in this basis, for  the ($m_1$,$m_2$)th block $W_{m_1,m_2}=P_{m_1}WP_{m_2}$, the scaling of  over $\alpha$ is given by
\begin{equation}
\left[
\begin{matrix}
W_{m_1+m_2} & W_{m_1+m_2+1}\\
W_{m_1+m_2+1} & W_{m_1+m_2+2}
\end{matrix}\right],
\end{equation}
where each block $W_{n}$ is a $Q\times Q$ matrix scaling as $\Theta(\alpha^{n})$.
To find the scaling of eigenvalues , we again try to observe the structure of the character polynomial $|\lambda I-W|=\sum_{k=0}^M(-1)^k e_k\lambda^{M-k}$, its coefficients now has the form of 
\begin{equation}\begin{aligned}
&e_1=\Theta(\alpha^0),\cdots, e_Q=\Theta(\alpha^0), \\ &e_{Q+1}=\Theta(\alpha^{0+2}),\cdots, e_{3Q}=\Theta(\alpha^{0+4Q}), \\ &e_{3Q+1}=\Theta(\alpha^{0+4Q+4}),\cdots, e_{5Q}=\Theta(\alpha^{0+4Q+8Q}),\cdots\cdots 
\end{aligned}
\end{equation}
So, the scaling of eigenvalue is
\begin{equation}\begin{aligned}
&\beta_{0\leq i\leq Q-1}^2=\Theta(1),\quad\beta_{Q\leq i\leq 3Q-1}^2=\Theta(\alpha^{-2}),\quad\beta_{3Q\leq i\leq 5Q-1}^2=\Theta(\alpha^{-4}),\cdots
\end{aligned}\end{equation}

In Fig.~\ref{new_method_Q}, we present the case where $\alpha = 10^{-2}$ is fixed, and different numbers of compact sources $Q$ are considered, with the centroids of these $Q$ sources evenly distributed over $[-L/2, L/2]$. The figure shows that when the number of sources $Q$ exceeds roughly $L/Q$, the stepwise increase in $C_T$ becomes smoothed out. Note that when we fix $L=20$, the two sources in the case of $Q=2$ are sufficiently far apart. However, as $Q$ increases to 3 or 5, the distance between the sources decreases, making the orthogonalization method more advantageous in these cases. In Fig.~\ref{new_method_different_alpha}, we consider the case where $Q=2$, but $\alpha_1$ and $\alpha_2$ take different values for the two compact sources. It is evident that the difference in $\alpha_i$ between the compact sources also smooths out the stepwise increase in $C_T$ by introducing smaller steps with varying thresholds. The prior knowledge about the moments are again chosen by randomly generating a set of images as in detailed in Sec.~\ref{SI:numerics}.

\subsection{Implementation of the orthogonalized SPADE method}\label{appendix:povm}

\yk{In this section, we discuss how to implement the orthogonalized SPADE method using a Hermite-Gaussian mode sorter,  beam splitters, and postselection. The central idea is that the Hermite-Gaussian mode sorter enables a transformation from Hermite-Gaussian modes into transverse momentum states. In our implementation of the orthogonalized SPADE method,   the basis functions are linear combinations of the Hermite-Gaussian modes. Therefore, the Hermite-Gaussian mode sorter enables the transformation of these orthogonalized basis functions into transverse momentum eigenstates (up to additional processing), thereby allowing full implementation of the desired measurement using only the mode sorter, beam splitter, and postselection. We illustrate this implementation with a representative example: imaging two compact sources.}

When there are two compact sources located in $u_1$ and $u_2$ respectively and  the PSF is given by $\psi(u)=\exp(-u^2/4\sigma^2)/(2\pi\sigma^2)^{1/4}$, for the separate SPADE measurement, the POVM as in Eq.~\ref{eq:POVM_separate} is $\left\{\frac{1}{2Q}\ket{\phi_{q,m,\pm}}\bra{\phi_{q,m,\pm}},\frac{1}{2Q}\ket{b_{q,0}}\bra{b_{q,0}}\right\}_{q=1,2,\cdots,Q,\, m=0,1,2,\cdots,\infty}$, where $\ket{\phi_{q,m,\pm}}=\frac{1}{\sqrt{2}}(\ket{b_{q,m}}\pm\ket{b_{q,m+1}})$. We can calculate the wavefunction $\ket{b_{q,m}}=\int du b_{q,m}(u)\ket{u}$ as 
\begin{equation}
\begin{aligned}\label{eq:b_q01}
b_{q,0}(u)=\frac{\exp(-(u-u_q)^2/4\sigma^2)}{(2\pi\sigma^2)^{1/4}},\quad b_{q,1}(u)=\frac{(u-u_q)\exp(-(u-u_q)^2/4\sigma^2)}{(2\pi)^{1/4}\sigma^{3/2}},  \quad q=1,2\,,
\end{aligned}
\end{equation}
\yk{
These basis functions, constructed under the assumption of a Gaussian point spread function, are the Hermite-Gaussian modes. In this subsection, we denote $b_{q,m}(u)$ as the $m$th Hermite-Gaussian mode centered at the centroid of the $q$th compact source.
For POVM of the orthogonalized SPADE method given in Eq.~\ref{eq:POVM_orthogonalized},
$\left\{\frac{1}{2}\ket{\phi_{j,\pm}^{(l)}}\bra{\phi_{j,\pm}^{(l)}},\frac{1}{2}\ket{b_{j}^{(0)}}\bra{b_{j}^{(0)}}\right\}_{j=1,2,\cdots,Q,\, l=0,1,2,\cdots,\infty},$ 
where
$\ket{\phi_{j\pm}^{(l)}}=\frac{1}{\sqrt{2}}\left(\ket{b_j^{(l)}}\pm\ket{b_j^{(l+1)}}\right)$, the basis functions $b_j^{(l)}(u)$ can be computed as linear combinations of the Hermite-Gaussian modes}
\begin{equation}\label{eq:bjl_combination}
\begin{aligned}
&b_{1}^{(0)}(u)=b_{1,0}(u),\\
&b_{2}^{(0)}(u)=p_1b_{1,0}(u)+p_2b_{2,0}(u), \quad p_1=-\frac{e^{-(u_1-u_2)^2/8\sigma^2}}{\sqrt{1-e^{-(u_1-u_2)^2/4\sigma^2}}}, \quad p_2=\frac{1}{\sqrt{1-e^{-(u_1-u_2)^2/4\sigma^2}}},\\
&b_{1}^{(1)}(u)=p_3b_{1,0}(u)+p_4b_{2,0}(u)+p_5b_{1,1}(u),\quad p_3=\frac{
    \sqrt{4 \sigma^2 + \frac{e^{\frac{u_1 u_2}{2 \sigma^2}} (u_1 - u_2)^2}{e^{\frac{u_1 u_2}{2 \sigma^2}} - e^{\frac{u_1^2 + u_2^2}{4 \sigma^2}}}} (u_1 - u_2)
}{
    4 \left(-1 + e^{\frac{(u_1 - u_2)^2}{4 \sigma^2}}\right) \sigma^2 - (u_1 - u_2)^2
},\\
&p_4=\frac{
    e^{\frac{(u_1 - u_2)^2}{8 \sigma^2}} 
    \sqrt{4 \sigma^2 + \frac{e^{\frac{u_1 u_2}{2 \sigma^2}} (u_1 - u_2)^2}{e^{\frac{u_1 u_2}{2 \sigma^2}} - e^{\frac{u_1^2 + u_2^2}{4 \sigma^2}}}} 
    (-u_1 + u_2)
}{
    4 \left(-1 + e^{\frac{(u_1 - u_2)^2}{4 \sigma^2}}\right) \sigma^2 - (u_1 - u_2)^2
}
,\quad p_5=-\frac{
    2 \left(-1 + e^{\frac{(u_1 - u_2)^2}{4 \sigma^2}}\right) \sigma 
    \sqrt{4 \sigma^2 + \frac{e^{\frac{u_1 u_2}{2 \sigma^2}} (u_1 - u_2)^2}{e^{\frac{u_1 u_2}{2 \sigma^2}} - e^{\frac{u_1^2 + u_2^2}{4 \sigma^2}}}}
}{
    4 \left(-1 + e^{\frac{(u_1 - u_2)^2}{4 \sigma^2}}\right) \sigma^2 - (u_1 - u_2)^2
}
.
\end{aligned}
\end{equation}
Higher-order basis states $\ket{b_q^{(l)}}$ can be constructed similarly, though their explicit forms become too complex to present here.

There has been significant discussion regarding Hermite-Gaussian mode sorting \cite{lavery2012refractive,beijersbergen1993astigmatic,ionicioiu2016sorting,zhou2017sorting,zhou2018hermite}, which enables the conversion of Hermite-Gaussian modes, with $m = 0, 1, 2, \ldots, \infty$ for a fixed $q$, into transverse momentum states. Note that when the PSF is Gaussian, $b_{q,m}(u)$ in Eq.~\ref{eq:b_q01} corresponds to the Hermite-Gaussian mode. The transformation from Hermite-Gaussian modes to momentum states can be expressed as $\ket{m} = F_q\ket{b_{q,m}}$. Since the orthogonalized SPADE basis $b_j^{(l)}(u)$ can be written as a linear combination of these Hermite-Gaussian modes $b_{q,m}(u)$, the mode converter $F_q$, allows for transforming the basis $b_j^{(l)}(u)$ into simple transverse momentum states with some addition steps, thereby enabling complete measurement.

Fig.~\ref{implementation} illustrates the process of generating superpositions of nonorthogonal spatial modes $\ket{\psi_{a,b,c}}$, where each spatial mode $\ket{\psi_{a,b,c}}$ can be created using the mode converter $F_{a,b,c}^\dagger$ from a single photon in momentum states. Examples are shown for the cases of two and three distinct spatial modes, which can be easily extended to include additional modes. For the implementation of $\ket{\phi_{j\pm}^{(l)}}$, we will need the spatial modes $b_{q,m}(u)$ for $q=1,2,\cdots,Q$ when $m\leq l$, and $q\leq j$ when $m=l+1$. 
This procedure can be reversed to detect spatial modes as linear combinations of $\ket{b_{q,m}}$. Since both $\ket{\phi_{j,\pm}^{(l)}}$ and $\ket{b_j^{(l)}}$ can be expressed as linear combinations of $\ket{b_{q,m}}$, the orthogonalized SPADE method can be implemented using the Hermite-Gaussian mode converter with some additional steps.

The basis $\ket{b_j^{(l)}}$ can be constructed differently from Eq.~\eqref{eq:bjl_SI} while preserving the key feature of the orthogonalized SPADE method---namely, ensuring that lower-order terms vanish in the probability distribution when estimating higher-order moments. This design enhances the performance of moment estimation. For instance, one could redefine the basis $\ket{\psi'^{(m)}_k}$ as a linear combination of $\ket{\psi^{(m)}_{j=1,2,\dots,Q}}$ in the Eq.~\eqref{eq:bjl_SI} and then construct $\ket{b_j^{(l)}}$ using the Gram-Schmidt procedure applied to the set of basis states ${\ket{\psi'^{(m)}_k}}_{k=1,2,\dots,Q;, m=0,1,2,\dots}$. 
This approach allows for the exploration of alternative constructions of $\ket{b_j^{(l)}}$. As the current construction is not necessarily the optimal choice for experimental implementation, the flexibility in constructing the basis $\ket{b_j^{(l)}}$ could potentially lead to simpler or more practical implementations of the orthogonalized SPADE method.

The implementation of the measurement must be sufficiently accurate to ensure that any deviations in the probability distribution due to imprecision remain smaller than the lowest-order terms of $\alpha$ for each outcome. Consequently, the precision requirements for the measurement become increasingly stringent for higher-order moments.

\begin{figure}[!tb]
\begin{center}
\includegraphics[width=0.9\columnwidth]{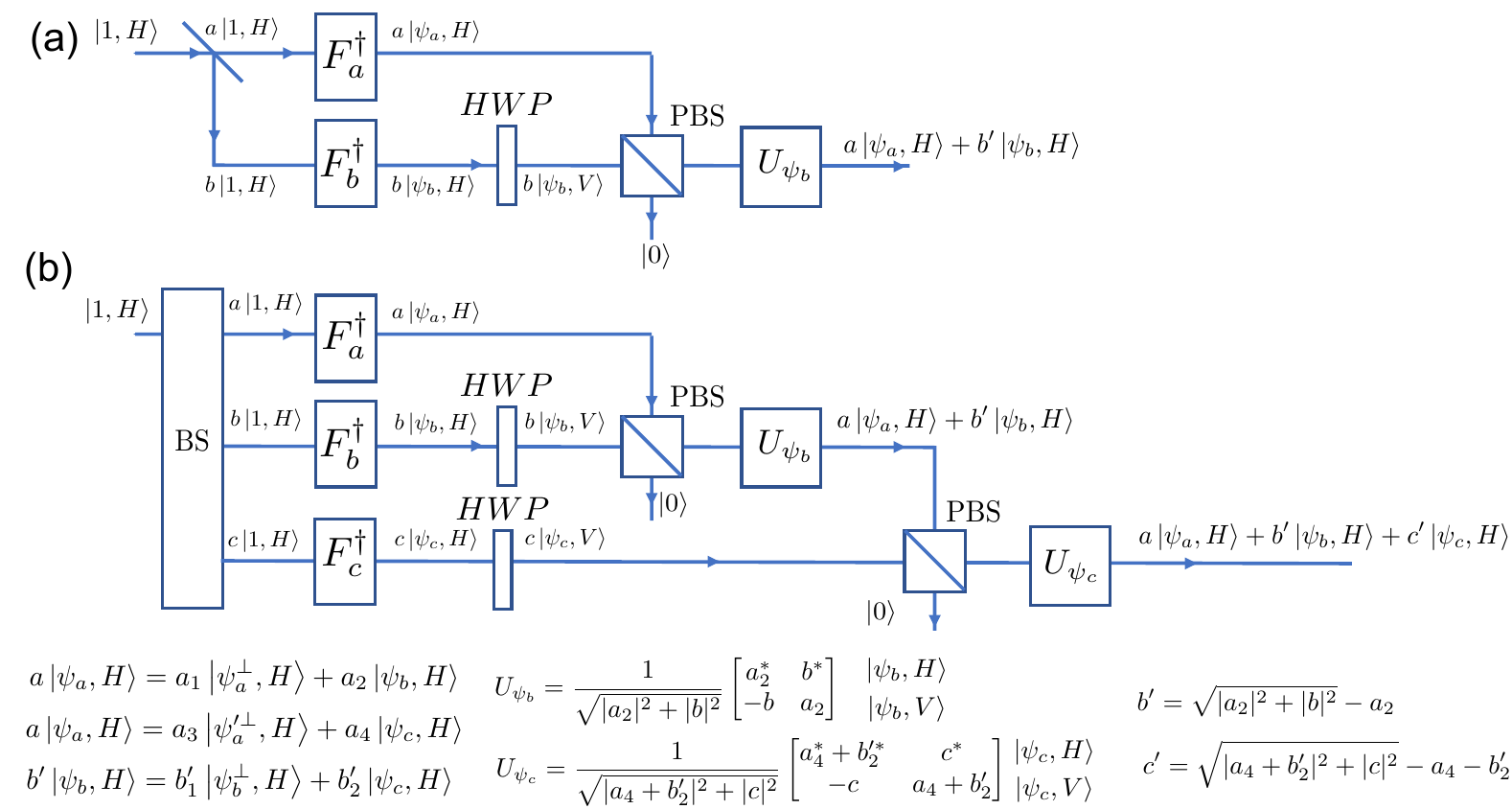}
\caption{Generation of a superposition of (a) two and (b) three nonorthogonal spatial modes using the mode converter $F_{a,b,c}$, which transforms the nonorthogonal spatial modes $\ket{\psi_{a,b,c}}$ into transverse momentum states. The labels $H$ and $V$ denote the two polarizations. The setup includes  polarizing beam splitter (PBS),  beam splitter (BS), and  half-wave plate (HWP). The operation $U_{\psi_{b,c}}$ is applied to the two polarizations of the spatial mode $\ket{\psi_{b,c}}$ to achieve the desired transformation. }
\label{implementation}
\end{center}
\end{figure}

\subsection{Eigenvectors and eigenvalues for multiple compact sources with varying position distributions}\label{SI:eigentask_multiple}

\begin{figure}[!tb]
\begin{center}
\includegraphics[width=1\columnwidth]{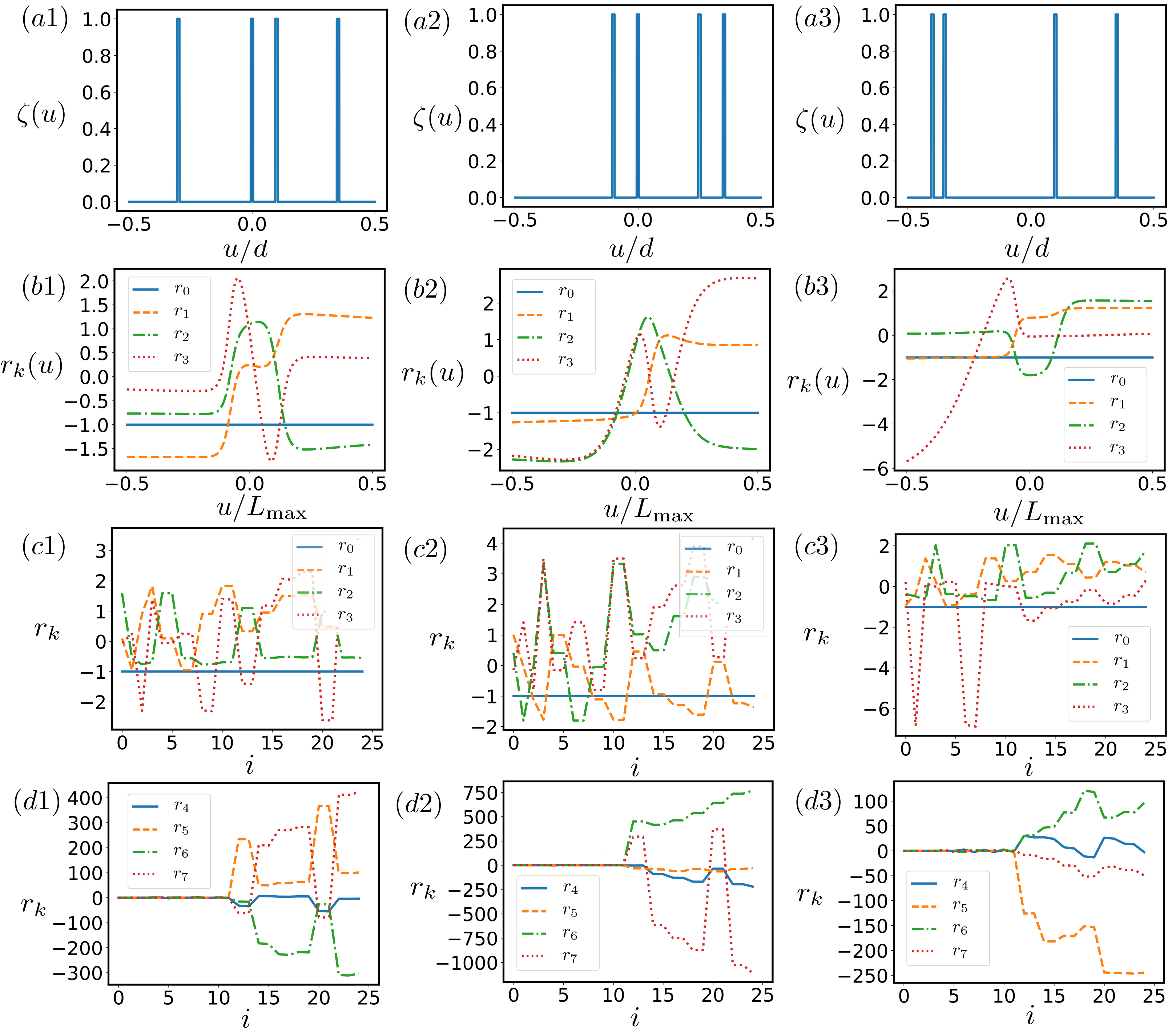}
\caption{\yk{(a1)–(a3) Different positions of the four compact sources, indicated by $\zeta(u)$. The $x$-axis indicates the position $u$ on the image plane.
(b1)–(b3) The first four eigenvectors $r_k$ for direct imaging. The $x$-axis indicates the position $u$ on the detection plane. 
(c1)–(c3), (d1)–(d3) The first eight eigenvectors $r_k$ for the orthogonalized SPADE method. The outcomes for the orthogonalized SPADE method are labeled in the following order: $P_{01}, \cdots, P_{04}, P_{0,1,+}, P_{0,1,-}, \cdots, P_{0,4,+}, P_{0,4,-}, P_{1,1,+}, P_{1,1,-}, \cdots, P_{1,4,+}, P_{1,4,-}, \dots$, where the meaning of each subscript in $P_{l,j,\pm}$ is defined in Eq.~\ref{eq:oSPADE_P}.
We set $\alpha = 10^{-1}$ for these figures. The centroids of all compact sources are confined within $[-d/2, d/2]$ with $d = 10$, the width of the PSF is $\sigma = 1$, and the size of each source is $\alpha\sigma = 0.1$. The region $[-L_{\text{max}}/2, L_{\text{max}}/2]$ on the image plane is detected with $L_{\text{max}} = d + 10\sigma$.
 } }
\label{fig:mutiple_eigenvector}
\end{center}
\end{figure}

\begin{figure}[!tb]
\begin{center}
\includegraphics[width=0.4\columnwidth]{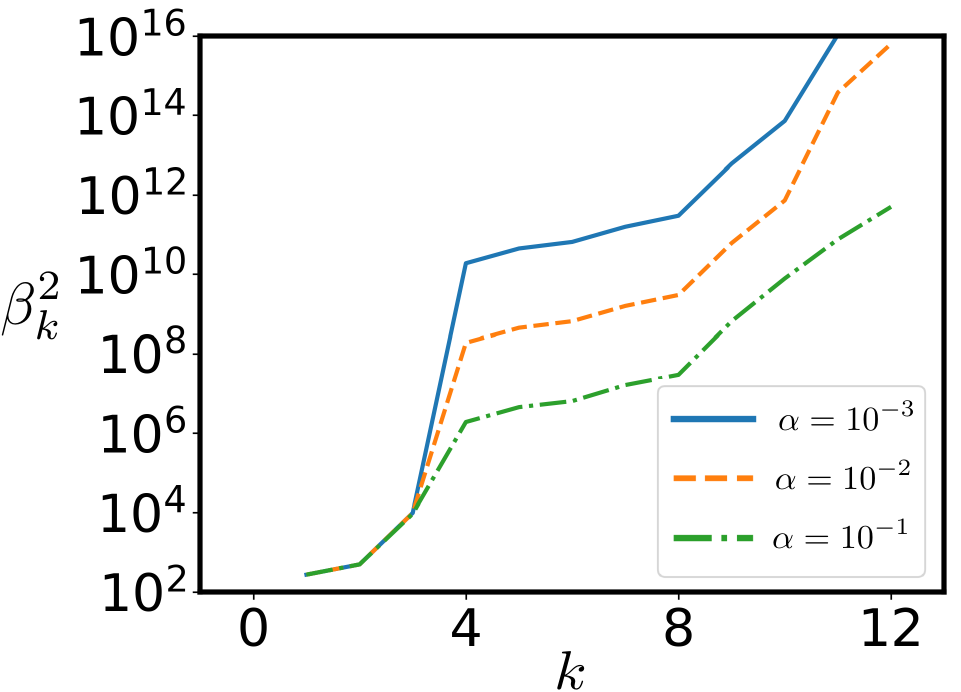}
\caption{\yk{Eigenvalues $\beta_k^2$ for various $k$ and different $\alpha$ using the orthogonalized SPADE method to image four compact sources. The simulation setup is identical to that in Fig.~\ref{fig:mutiple_eigenvector}(a3). The value $\beta_0^2 = 0$ is omitted from the plot.}
 } 
\label{fig:mutiple_eigenvalue}
\end{center}
\end{figure}

\yk{

In the case of imaging multiple compact sources, the eigentasks generally take on very diverse forms in practical scenarios. In fact, the form of the eigentasks is strongly influenced by the practical imaging model and the prior knowledge, such as the positions of the individual sources. As previously noted, this is an important aspect of the problem: when the number of samples is finite, it is unclear which features can be accurately estimated, and identifying these features is nontrivial. This is crucial because, in downstream analysis, it is essential to know which measured features can be reliably used.

We present the eigentasks corresponding to three different configurations of four compact sources in Fig.  \ref{fig:mutiple_eigenvector}(a1)-(a3). The numerical simulation procedure follows that described in Sec. \ref{SI:numerics} of the Supplemental Material: prior information is obtained by randomly generating a set of source distributions, which are then used to compute the sampled values of $d_n$ and $g_{mn}$. We assume the presence of four compact sources, each confined to a region indicated by the function $\zeta(u)$ in the figure. Note that the intensity distributions within each $\zeta(u) = 1$ region are randomly generated; we are not imaging four simple sources with flat intensity profiles. Instead, each region where $\zeta(u) = 1$ may exhibit more complex internal structure, while any region with $\zeta(u) = 0$ has zero intensity.

Let us first consider the direct imaging method, as shown in Fig.~\ref{fig:mutiple_eigenvector}(b1)–(b3). The eigenvectors $r_k$, obtained by solving the generalized eigenvalue equation $V r_k = \beta_k^2 G r_k$ in Eq. \ref{SI_eq:eigenproblem}, are displayed in the figure. In the direct imaging method, each measurement outcome corresponds to the probability of detecting a photon at a position $u$ on the image plane.  The $k$-th eigentask is given by $\sum_i r_k(u_i) P(u_i)$, where $P(u_i)$ is the probability of detecting a photon at position $u_i$. It is evident that the eigentasks are strongly influenced by the positions of the compact sources.

We next examine the eigentasks of the orthogonalized SPADE method. From Fig.~\ref{fig:mutiple_eigenvector}(c1)–(c3), we observe the first four eigenvectors of the orthogonalized SPADE method. As discussed in the main text, for four compact sources, we expect four eigentasks corresponding to the Rayleigh-resolvable features. In general, these eigentasks incorporate all measurement outcomes of the constructed orthogonalized SPADE method, particularly including the first few outcomes that scale as $\Theta(1)$.
From Fig.~\ref{fig:mutiple_eigenvector}(d1)–(d3), we observe the next four eigenvectors of the orthogonalized SPADE method, which correspond to the sub-Rayleigh features. Note that the eigenvectors $r_k$ are normalized such that $r_{k_1}^T D r_{k_2} = \delta_{k_1k_2}$. Consequently, some components of $r_k$ can take very large values because certain elements of $D$ in the orthogonalized SPADE method can be quite small, owing to the $\alpha^{2n}$ scaling in $D$. A clear characteristic of these eigentasks is that they have vanishing components in the first few outcomes. This is expected because the construction of these eigentasks explicitly cancels the lower-order terms in $\alpha$, as shown in Eq.~\ref{eq:oSPADE_P}, and we therefore anticipate corresponding structures in the eigenvectors $r_k$. Direct computation of the eigenvalue problem confirms this expectation.
Clearly, the eigentasks exhibit complex structures that are strongly influenced by the positions of the compact sources. This complexity becomes increasingly important as the number of compact sources grows, since the prefactors accompanying the $\alpha$-dependent terms also affect the thresholds, especially when $\alpha$ is not very small. Selecting eigentasks along with their corresponding thresholds is therefore crucial for the downstream analysis of features extracted from the images.

As shown in Fig.~\ref{fig:mutiple_eigenvalue}, we plot the eigenvalues $\beta_k^2$ obtained using the orthogonalized SPADE method to image four compact sources. These results reveal that although, when $\alpha$ is sufficiently small, the $\alpha$-scaling dominates the behavior of $\beta_k^2$, for the multiple compact source problem the prefactors can also lead to an increase in $\beta_k^2$ at higher-order eigenvalues. See Sec.\ref{SI:prefactor} of the Supplemental Material for a more systematic study of these prefactors.
In practice, we view the multiple compact source model as an approximation of more general sources exhibiting clustered structures, or as a potential approach to tackle general imaging problems after further development of the orthogonalized SPADE method, which aims to resolve fine details of the source even when its overall extent exceeds the Rayleigh limit. This suggests that the number of compact sources can be large in practical tasks.
It is therefore essential to identify the eigentasks that can be reliably estimated given a finite sample size in such scenarios involving multiple compact sources. In general, these eigentasks incorporate contributions from all the compact sources, determined by the prior information.

}

\section{Details of numerical calculations}\label{SI:numerics}

\begin{figure}[!b]
\begin{center}
\includegraphics[width=1\columnwidth]{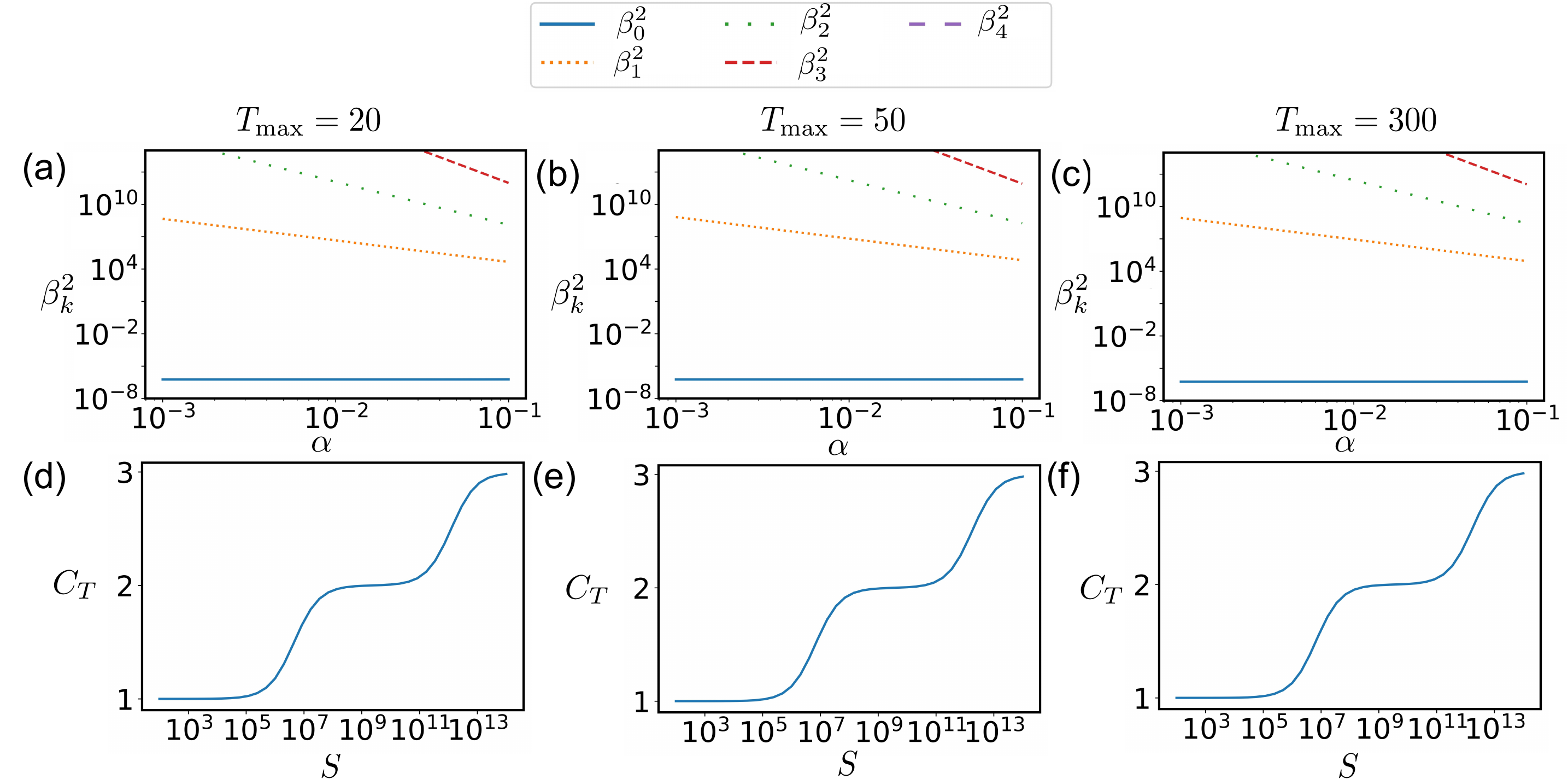}
\caption{Robustness of the numerical calculation for direct imaging with different discretization $T_{\text{max}}$. The calculation is done for one compact source, and for (d)-(e) we choose $\alpha=10^{-2}$. (Note in this subsection, we include $\beta_0^2$ in the  plot for $\beta_k^2$ as a function of $\alpha$, which numerically appears slightly above zero due to finite numerical precision. This discrepancy could be minimized by increasing the calculation's precision. Additionally, we do not interpolate values of $\beta_k^2$ above $10^{15}$ with a straight line as in the main text; instead, we introduce a direct cutoff at $10^{15}$, beyond which numerical calculations lack sufficient precision.) } 
\label{robustness_direct_imaging}
\end{center}
\end{figure}

\begin{figure}[!tb]
\begin{center}
\includegraphics[width=1\columnwidth]{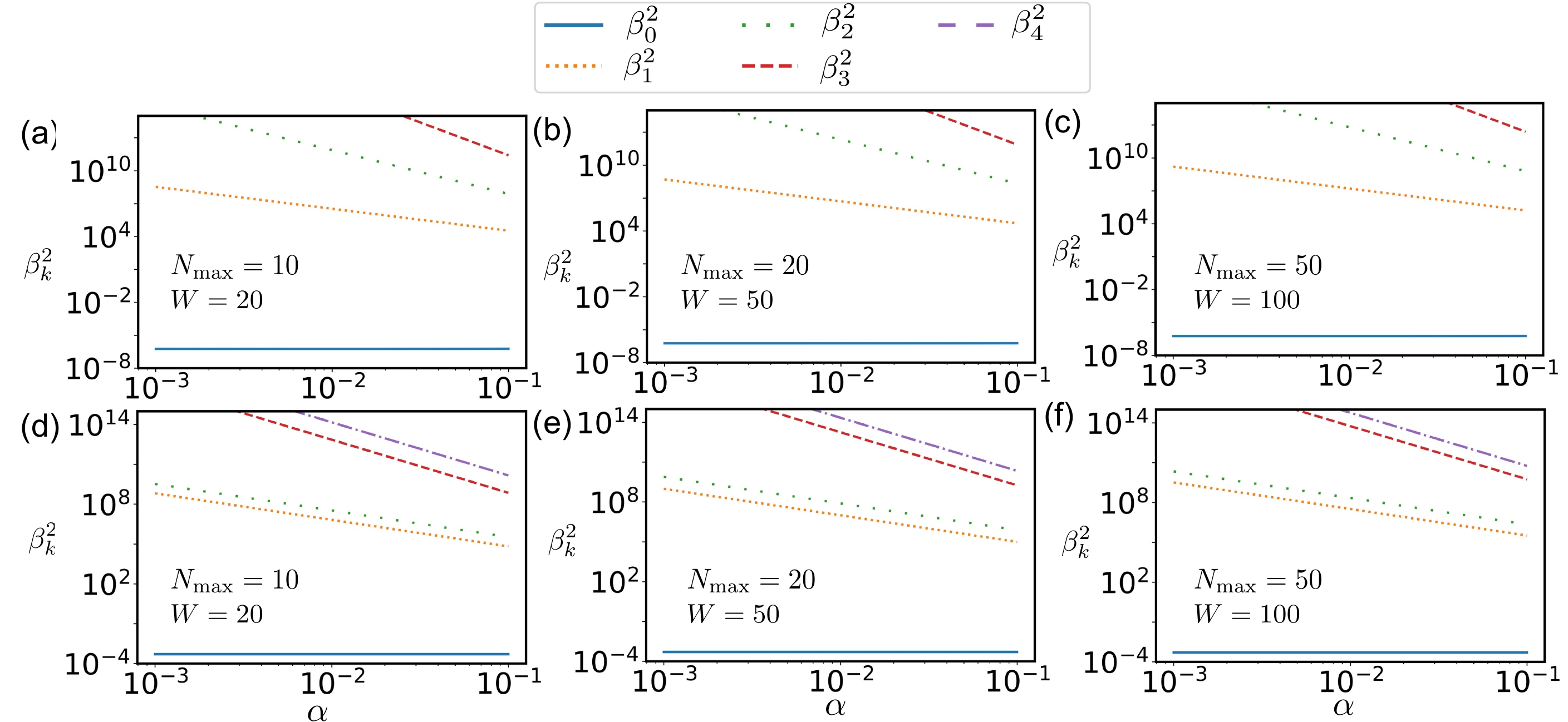}
\caption{
The numerical calculation of $\beta_k^2$  performed using various randomly generated sets of prior information, with different values for $N_{\text{max}}$ and $W$. We analyze the imaging of a single compact source using direct imaging and the SPADE method. 
$\beta_4^2$ exceeds the cutoff at $10^{15}$ for (a)-(c). In the main text, this line is interpolated using calculations for $\alpha > 0.1$.}
\label{gd_generation_beta_k}
\end{center}
\end{figure}

\begin{figure}[!tb]
\begin{center}
\includegraphics[width=1\columnwidth]{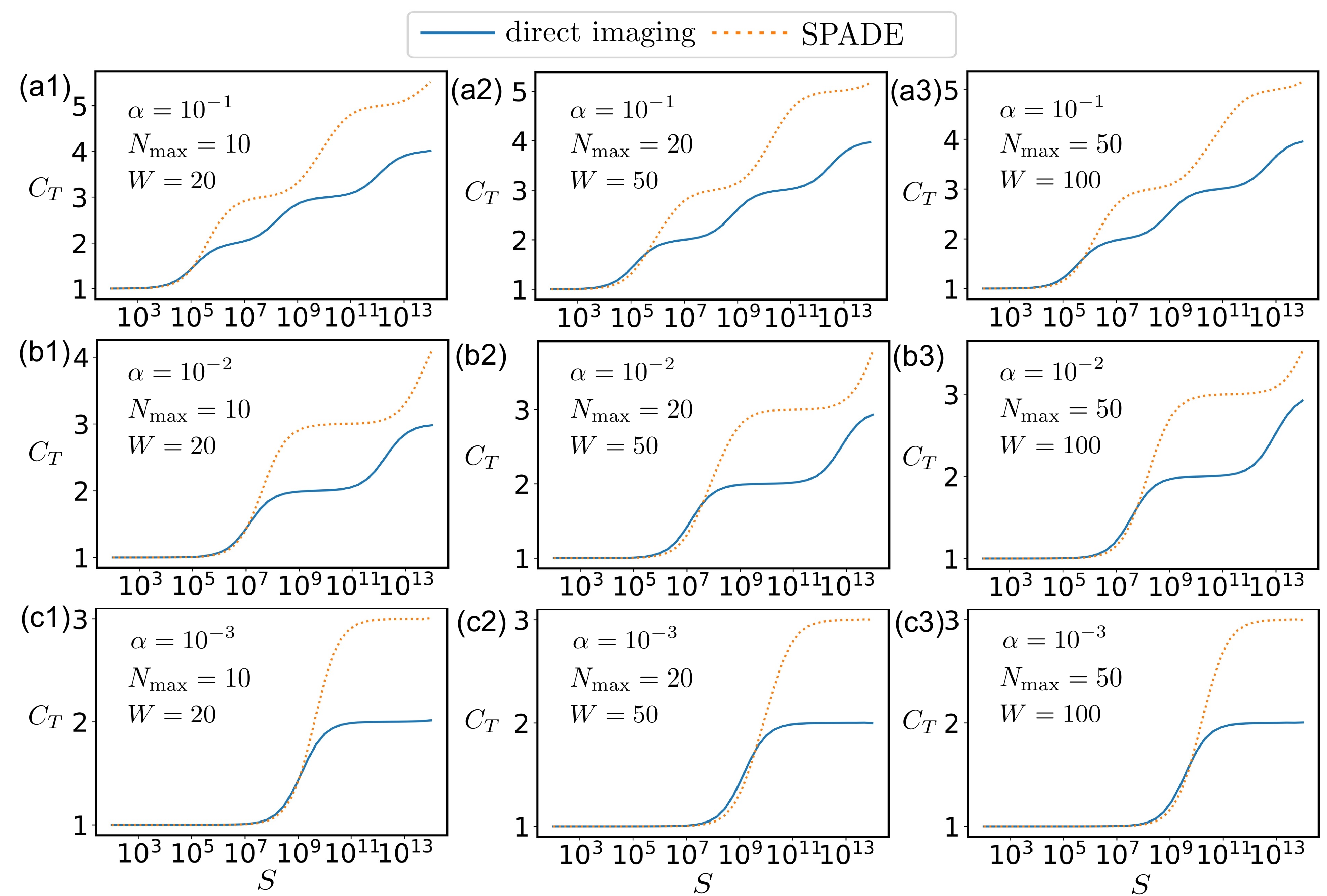}
\caption{The numerical calculation of REC $C_T$ as a function of $S$  performed using various randomly generated sets of prior information, with different values for $N_{\text{max}}$ and $W$. We analyze the imaging of a single compact source using direct imaging and the SPADE method.  } 
\label{gd_generation_CT}
\end{center}
\end{figure}

\begin{figure}[!htb]
\begin{center}
\includegraphics[width=1\columnwidth]{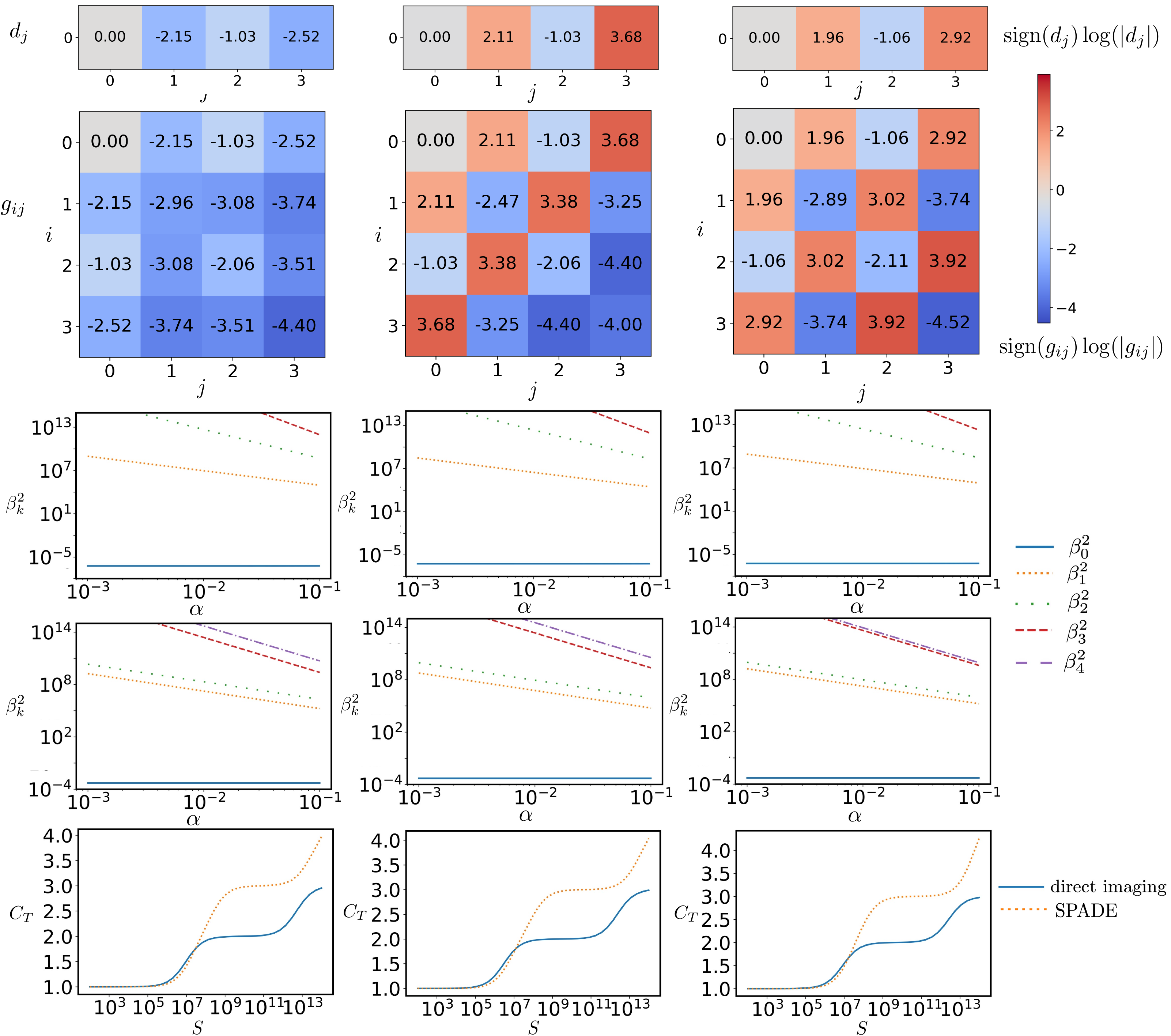}
\caption{The numerical calculation of REC $C_T$ as a function of $S$ when $\alpha=10^{-2}$ and $\beta_k^2$ as a function of $\alpha$. The calculation is performed using three different prior information randomly generated when $N_{\text{max}}=10$, $W=20$. The first few elements of $d,g$ for these three different prior information are given. We analyze the imaging of a single compact source using direct imaging and the SPADE method.  } 
\label{different_gd}
\end{center}
\end{figure}

In our numerical calculations for the direct imaging case, we discretize the spatial coordinate for analysis, which is reasonable since, in practice, direct imaging involves pixels of finite size. Assuming the centroid of all compact sources or point sources is within the interval $[-d/2, d/2]$, we extend this range to $[-d/2 - 5\sigma, d/2 + 5\sigma]$ for the PSF given by $\psi(u) = \frac{\exp(-u^2/4\sigma^2)}{(2\pi\sigma^2)^{1/4}}$. Because of the exponential decay of the PSF, this cutoff is expected to have a negligible impact on the results. We then select $T_{\text{max}}$ points at $X_1,X_2,\cdots,X_{T_{\text{max}}}$ uniformly distributed across this interval, forming the set of POVMs $\{M_i\}_{i=1,2,\dots,T{\text{max}}}$, where each $M_i = \int_{X_i - l/2}^{X_i + l/2} dx \ket{x}\bra{x}$ and $l = (d + 10\sigma)/T_{\text{max}}$. In the calculations presented in the main text, we use $T_{\text{max}}=50$, as further increasing $T_{\text{max}}$ does not show obvious difference. A plot illustrating the convergence of our calculations for different $T_{\text{max}}$ values confirms the stability of the results. Using the imaging of a single compact source as an example (Fig.~\ref{robustness_direct_imaging}), we show that varying $T_{\text{max}}$ does not affect the scaling behavior of $\beta_k^2$ with $\alpha$ or the stepwise increase of $C_T$ with $S$. In theory, $C_T$ should increase slightly with larger $T_{\text{max}}$ values due to the greater number of independent pixels; however, this effect is minimal in the numerical calculation. This is likely because, for compact source imaging, the information contained in the nearby pixel is largely similar.


\begin{table}[bt]
\centering
\begin{tabular}{c|ccccc}
  \hline
  $m\backslash l$ & 0 & 1 & 2 & 3 & 4 \\
  \hline
  0 & 1 & 0 & 0 & 0 & 0 \\
  1 & 0 & $\frac{1}{2}\alpha$ & 0 & 0 & 0  \\
  2 & $-\frac{1}{8}\alpha^2$ & 0 & $\frac{1}{4\sqrt{2}}\alpha^2$ & 0 & 0  \\
  3 & 0 & $-\frac{1}{16}\alpha^3$ & 0 & $\frac{1}{8\sqrt{6}}\alpha^3$ & 0 \\
  4 & $\frac{1}{128}\alpha^4$ & 0 & $-\frac{1}{32\sqrt{2}}\alpha^4$ & 0 & $\frac{1}{32\sqrt{6}}\alpha^4$ \\
  \hline
\end{tabular}
\caption{$a_{ml}=\bra{\psi^{(m)}}\ket{b_l}$ for the coefficient Eq.~\ref{eq:aml}.} 
\label{aml} 
\end{table}

For imaging of one compact source using the SPADE method \cite{zhou2019modern}, assume the PSF is again given by $\psi(u)={\exp(-u^2/4\sigma^2)}/{(2\pi\sigma^2)^{1/4}}$. 
Using $\rho=\sum_{m,n=0}^\infty x_{m+n}\ket{\psi^{(m)}}\bra{\psi^{(n)}}$,
we will construct the orthonormal measurement basis $\{\ket{b_l}\}_l$ such that
\begin{equation}\label{eq:aml}
a_{ml}=\bra{\psi^{(m)}}\ket{b_l}=\left\{
\begin{array}{cc}
=0    &  m\leq l-1\\
\neq 0     &  m\geq l
\end{array}\right.
\end{equation}
We can explicitly calculate the coefficients $a_{ml}$, the first few coefficients is shown in Table~\ref{aml} as examples.
If we define the measurement as  in Eq.~\ref{ref5_POVM}, the probability distribution of getting $\ket{b_0}$, $\ket{\phi_{k\pm}}$, $k=0,1,2,\cdots$ is given by
\begin{equation}
\begin{aligned}
&P_0=\frac{1}{2}\sum_{m,n=0}^\infty x_{m+n} a_{m,0}^*a_{n,0},\\
&P_{k,+}=\frac{1}{4}\sum_{m,n=0}^\infty x_{m+n}(a_{m,k}^*a_{n,k}+a_{m,k+1}^*a_{n,k+1}+a_{m,k+1}^*a_{n,k}+a_{m,k}^*a_{n,k+1}),\\
&P_{k,-}=\frac{1}{4}\sum_{m,n=0}^\infty x_{m+n}(a_{m,k}^*a_{n,k}+a_{m,k+1}^*a_{n,k+1}-a_{m,k+1}^*a_{n,k}-a_{m,k}^*a_{n,k+1}).
\end{aligned}\end{equation}
By comparing this concrete form of probability with the abstract form used earlier, $P_0=\sum_{m=0}^\infty c_m x_m\alpha^m$, $P_{k,\pm} = \sum_{m \geq 2k} c_{km}^{\pm} x_m \alpha^m$ as in Eq.~\ref{eq:P_superresolution2}, we can determine the individual coefficients $c_{nm}^\pm$.

\yk{
To perform the numerical calculations for both single and multiple compact sources, we randomly generate prior knowledge $p(\vec{x})$ over the input moment vectors $\boldsymbol{\theta} = \vec{x} = [x_0, x_1, x_2, \dots]$. This is done by first generating a set of $W$ one-dimensional images, each of which yields a corresponding moment vector $\vec{x}(\omega)$ that serves as a  sample of input parameters, where $\omega$ indexes each generated vector. The prior distribution $p(\vec{x})$ is then defined by the empirical statistics of this ensemble.
This prior is used to compute the quantities $d_n = \int d\vec{x} p(\vec{x}) x_n$ and $g_{n_1n_2} = \int d\vec{x} p(\vec{x}) x_{n_1} x_{n_2}$, which enter subsequent stages of the calculation. Numerically, these matrix elements are estimated by averaging over the sampled moment vectors:
$d_n = \frac{1}{W} \sum_w x_n(w)$ and $g_{n_1n_2} = \frac{1}{W} \sum_w x_{n_1}(w) x_{n_2}(w)$.
This sampling-based approach enables efficient incorporation of general prior knowledge into the numerical evaluation of the quantities required for our analysis.

To generate these images, we assume that the $q$th compact source is confined to a region of size $L_q$, which is divided into $N_{\text{max}}$ equal segments. A point source is placed at the center of each segment, with an assigned intensity $I_m$ drawn uniformly from the interval $[0, 1]$. As $N_{\text{max}}$ becomes sufficiently large, this setup provides high degrees of freedom to model generally distributed compact sources. After assigning random intensities $I_m$, we normalize them by replacing each $I_m$ with $I_m / (\sum_m I_m)$ so that the total intensity sums to one. Because each point source has a known coordinate, we can compute the corresponding moment vector $\vec{x}(w)$ for each generation.
By repeating this procedure, we obtain a set of $W$ moment vectors $\vec{x}(w)$ for use in statistical averaging. In the calculations presented in the main text, we set $N_{\text{max}} = 20$ and $W = 50$, meaning each compact source is modeled by 20 evenly spaced pixels, and 50 different moment vectors are generated to compute $d_n$ and $g_{n_1n_2}$.
For scenarios involving multiple compact sources whose centroids are separated by a distance $L$ and individual sizes $L_q \ll L$, we generate nonzero intensity only within each compact region of size $L_q$. The intensities are normalized such that the total intensity across all compact sources sums to one.}

We emphasize that these randomly generated sources serve as a simple example to simulate prior knowledge and to demonstrate our approach. However, the discussion in the main text applies to any general prior information about the moments $\vec{x}$. It is important to note that a set of images is needed to generate the matrices $g_{n_1n_2}$ and $d_n$; using only one image would make $g_{n_1n_2}$ a rank-one matrix, implying that the moment of the image is already known, hence no additional information can be obtained through measurement.

We examine the impact of varying $M$ and $W$ on both the scaling of $\beta_k^2$, shown in Fig.~\ref{gd_generation_beta_k}, and the $C_T$ vs. $S$ plot, shown in Fig.~\ref{gd_generation_CT}. As we increase $N_{\text{max}}$ and $W$, the scaling of $\beta_k^2$ with $\alpha$  and the stepwise increase of $C_T$ as a function of $S$ remains the same. However, a careful reader may observe that as $N_{\text{max}}$ and $W$ grow,  $\beta_k^2$ exhibit a slight increase, as can be seen by comparing Fig.~\ref{gd_generation_beta_k}(a) with Fig.~\ref{gd_generation_beta_k}(c). This trend is also evident in Fig.~\ref{gd_generation_CT}(c1) and Fig.~\ref{gd_generation_CT}(c3), where the threshold for each stepwise increment rises.
This effect occurs because the prior knowledge influences the prefactor of the eigenvalues $\beta_k^2$. Numerically, we observe the prefactors of $\beta_k^2$ tend to increase as the eigenvalues of $g_{n_1n_2}$ decrease. If $g_{n_1n_2}$ is full rank, the prior information is not expected to influence the scaling of $\beta_k^2$.

The method described above generates a unique random prior knowledge for each execution, resulting in a different set of $g$ and $d$ for every run. 
In all the discussions in this work, the plots are generated using a single instance of randomly generated prior knowledge.
To verify that our method works universally for any prior information, we run the code multiple times and observe that each instance of prior knowledge produces a similar scaling of $\beta_k^2$ with $\alpha$, as expected. Consequently, $C_T$ consistently displays a stepwise increase as a function of $S$. In Fig.~\ref{different_gd}, we present results for three randomly generated instances of prior knowledge using this method. In these three instances, $g$ and $d$ differ due to the difference of prior information.
While the eigenvalues and $C_T$ vary with the prior knowledge, the scaling behavior of $\beta_k^2$ remains consistent, ensuring the stepwise increase of $C_T$ is always observed.

\section{Discussion about the prefactors}\label{SI:prefactor}
\yk{
In this section, we discuss the prefactors of $\beta_k^2$, complementing the main text, which primarily focuses on the scaling of $\beta_k^2$ with respect to $\alpha$ in different scenarios. For convenience, we summarize the eigenvalue problems from Eq.~\ref{SI_eq:eigenproblem}, Eq.~\ref{SI_eq:eigentask}, and Eq.~\ref{SI_eq:eigentask_SPADE} below
\begin{equation}\begin{aligned}
&V=D-G,\quad Vr_{k}=\beta_k^2 G r_{k},\quad Gr_{k}=\lambda_kDr_{k},\quad \lambda_k=(1+\beta_k^2)^{-1},\\
&D=\sum_{n=0}^\infty d_n\alpha^n \mathcal{C}_n,\quad G=\sum_{n=0}^\infty\alpha^n\sum_{i+j=n}g_{ij}  C_i C_j^T.
\end{aligned}\end{equation}
In general, it is difficult to extract the exact values of the prefactors analytically. However, we can analyze a special case that offers clear intuition and, in some situations, serves as a good approximation to the actual prefactor values in $\beta_k^2$. We redefine
\begin{equation}\label{SI_eq:mathcal_C}
\begin{aligned}
&W=D^{-1/2}GD^{-1/2},\quad Wy_{k}=\lambda_ky_{k},\quad y_k=D^{1/2}r_k,\\
&W=\sum_{n=0}^\infty\alpha^n\sum_{i+j=n}g_{ij}  (D^{-1/2}C_i)(D^{-1/2}C_j)^T=\mathbb{C}^Tg\mathbb{C},\\
&\mathbb{C}=[D^{-1/2}C_0,\alpha D^{-1/2}C_1,\alpha^2 D^{-1/2}C_2,\cdots]\\
&\mathbf{C}gv_k=\lambda_kv_k,\quad \mathbf{C}=\mathbb{C}\mathbb{C}^T
\end{aligned}
\end{equation}
where we use the fact that the eigenvalues of $AB$ and $BA$ are equal \cite{horn2012matrix}. Therefore, the analysis of the eigenvalues $\beta_k^2$ reduces to examining the spectrum of the product of the two matrices $\mathbf{C}g$. In general, this is still difficult to compute. However, in the special case where $\mathbf{C}$ and $g$ commute, we can choose $v_k$ to be the common eigenvectors of both $\mathbf{C}$ and $g$. In this case, the eigenvalues satisfy $\lambda_k = \lambda_k(\mathbf{C}) \lambda_k(g)$, so the spectrum is directly determined by the eigenvalues of $\mathbf{C}$ and $g$.

Of course, in general, $\mathbf{C}$ and $g$ do not commute with each other. However, in the following, we numerically show that the product of the eigenvalues of $\mathbf{C}$ and $g$ provides a good approximation to the actual values of $\lambda_k$ in some examples.
As shown in Fig.~\ref{fig:prefactor_sample}, we plot the eigenvalues of $g$ and $\mathbf{C}$ for both the direct imaging and SPADE methods. The simulation consists of $10$ successive runs of such simulations, where each run uses a randomly generated prior, as explained in Sec.~\ref{SI:numerics} of the Supplemental Material. For each run, we randomly generate a set of $W=5$ one-dimensional images and use the corresponding moment vectors to calculate $d$ and $g$ by averaging, as described in Sec.~\ref{SI:numerics} of the Supplemental Material.
It is clear that the eigenvalues $\lambda_k(g)$ of the matrix $g$ fluctuate randomly. The eigenvalues $\lambda_k(\mathbf{C})$ of the matrix $\mathbf{C}$ for both the direct imaging and SPADE methods used for imaging a single compact source are shown for each run. We observe that the eigenvalues $\lambda_k(\mathbf{C})$ do not vary much under the randomly generated priors considered here.
We also plot the prefactor $\tau_k = \lambda_k / \alpha^{n_k}$, where $n_k = [0, -2, -4, -6, \cdots]$ for direct imaging and $n_k = [0, -2, -2, -4, -4, \cdots]$ for SPADE method. We find that the variation of $\tau_k$ is mainly determined by $\lambda_k(g)$ for a given approach, since $\lambda_k(\mathbf{C})$ does not vary much with the prior considered here; however, $\tau_k$ can differ significantly between approaches due to differences in $\lambda_k(\mathbf{C})$.
Clearly, the variation of $\tau_k$ computed using the approximate $\hat{\lambda}_k = \lambda_k(\mathbf{C}) \lambda_k(g)$ for both direct imaging and SPADE closely follows that obtained using the actual eigenvalues $\lambda_k$. We find that this approximation is reasonably good and can serve as an intuitive guide for understanding the spectrum of $\beta_k^2$ in the examples considered here.
Note that $\lambda_k = (1 + \beta_k^2)^{-1}$, so when $\alpha \ll 1$, $\tau_{k \geq 1}^{-1}$ is roughly the prefactor of $\beta_k^2$. We choose to plot the prefactors of $\lambda_k$ because $\tau_k$ exhibits the same trend as the eigenvalues $\lambda_k(g)$, while the prefactors of $\beta_k^2$ show the inverse trend. Thus, it is easier to observe the behavior of the prefactors using $\lambda_k$.

Note $g_{ij} = \int d\vec{x}\, p(\vec{x}) x_i x_j$ is the matrix that directly quantifies our prior knowledge of the moments of the images. The matrix $\mathbf{C}$ also incorporates some prior information, but further includes $\alpha$ and the vectors $C_i$, which are directly related to the measurement strategy. In other words, different approaches, such as direct imaging and SPADE methods, have the same $g$ but different $\mathbf{C}$. In this sense, we can intuitively think of $g$ as quantifying the prior information, while $\mathbf{C}$ quantifies the measurement strategy.
Note that $\mathbf{C}$ also depends on the prior information through its dependence on $D$, but numerical analysis above suggests that this dependence may not be very strong in some cases. Thus, this analysis offers an intuitive and heuristic way to understanding the behavior of the prefactors, highlighting the two main factors that influence $\beta_k^2$: the measurement strategy, which primarily determines $\mathbf{C}$, and the prior information, which determines $g$.
Note that the vectors $C_i$, which constitute $\mathbf{C}$, are the coefficients of each order term in the series expansion of the probability distribution in powers of $\alpha$. Intuitively, larger eigenvalues of $\mathbf{C}$ indicate that the $C_i$ are more independent, meaning that the probability distributions exhibit more distinct behavior in each coefficient. Consequently, the measurement outcomes may contain more information, as reflected by the larger eigenvalues of $\mathbf{C}$.
Lastly, we emphasize that the approximate value $\hat{\lambda}_k = \lambda_k(g) \lambda_k(\mathbf{C})$ is meant solely as an intuitive tool to understand the behavior of the prefactors of $\beta_k^2$, inspired by a special case that $g$ and $\mathbf{C}$ commute with each other. We do not claim that this approximation is universally valid for all prior information and measurement strategies.

}

\begin{figure}[!tb]
\begin{center}
\includegraphics[width=1\columnwidth]{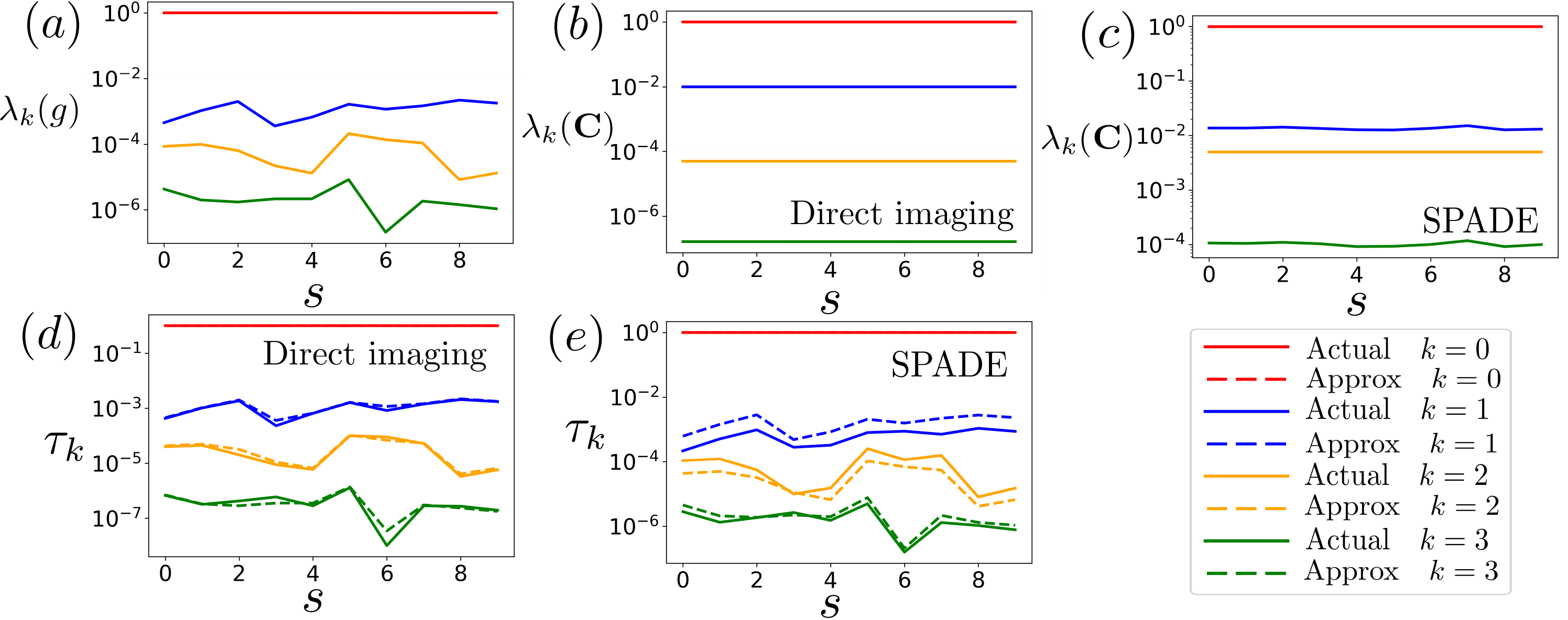}
\caption{\yk{(a) Eigenvalues of $g$. (b)–(c) Eigenvalues of $\mathbf{C}$ for direct (b) and SPADE (c) imaging of a compact source. (d) Prefactors of $\lambda_k$ for direct imaging, defined as $\tau_k = \lambda_k / \alpha^{n_k}$, where $n_k = [0, -2, -4, -6, \cdots]$. (e) Prefactors of $\lambda_k$ for the SPADE method, defined as $\tau_k = \lambda_k / \alpha^{n_k}$ with $n_k = [0, -2, -2, -4, -4, \cdots]$.
We fix $\alpha = 0.1$. All plots are computed over 10 successive runs of the simulation with randomly generated prior distributions, labeled by $s$ (the $s$-th run). Here, $k$ labels the $k$-th eigenvector. The dashed lines represent the prefactors computed using the approximate value $\hat{\lambda}_k = \lambda_k(g) \lambda_k(\mathbf{C})$, while the solid lines represent the prefactors computed using the actual eigenvalues $\lambda_k$.}
} 
\label{fig:prefactor_sample}
\end{center}
\end{figure}

\section{Demonstrative example}\label{SI:demonstrative_example}

In this section, we provide the simulation details of the demonstrative examples that illustrate the use of the quantum learning approach and our new orthogonalized SPADE method to perform a distinguishing task with a highly complex structure, inspired by the pattern of a QR code and representative of challenges encountered in practical applications. We begin with the imaging of a general source beyond the Rayleigh limit, then discuss the imaging of multiple compact sources, and finally highlight the advantages of the quantum learning approach in addressing such complex imaging problems.


\begin{algorithm}[H]
\caption{Demonstrative Example Procedure}
\begin{algorithmic}\label{Algorithm:general}
\STATE \textbf{Input:} Number of pixels $N$, interval size $L$, activation probability $p_0$, number of training samples per class $W_\text{train}$, number of test samples per class $W_\text{test}$, total detected photons per test image $S$, truncation order $\mathcal{K}$
\STATE Define indicator function $\eta(u) \in \{0,1\}$ over $[-L/2, L/2]$ to specify active regions.

\STATE \textbf{Training:}
\FOR{each class (Case 1: active; Case 2: inactive)}
    \FOR{$i = 1$ to $W_\text{train}$}
        \FOR{each pixel $j=1$ to $N$}
            \IF{Case 1 and $\eta(u_j) = 1$}
                \STATE Set $I_i(u_j) = 1$ with probability $p_0$, else $0$
            \ELSIF{Case 2 and $\eta(u_j) = 0$}
                \STATE Set $I_i(u_j) = 1$ with probability $p_0$, else $0$
            \ELSE
                \STATE Set $I_i(u_j) = 0$
            \ENDIF
        \ENDFOR
        \STATE Normalize $I_i(u)$ such that $\sum_{j=1}^N\, I_i(u_j) = 1$, which determines the  nonzero value of $I_i(u_j)$ as $i_0$.
    \ENDFOR
\ENDFOR
\STATE Compute prior matrices $d_k$ and $g_{ij}$ from all training samples, taking the prior to be the uniform distribution over the training set.
\STATE Solve eigenvalue problem (Eq.~\ref{SI_eq:eigenproblem}) to obtain eigenvectors $r_k$ and eigenvalues $\beta_k^2$.
\FOR{each training sample $i$ in both classes}
    \STATE Compute theoretical detection probabilities $P_i(u_j)$.
    \STATE Compute eigentask components $\xi_{ki} = \sum_j r_{kj} P_i(u_j)$.
    \STATE Form truncated eigentask vector $\vec{\xi}_i = [\xi_{0i}, \dots, \xi_{\mathcal{K}i}]$.
\ENDFOR
\STATE Train a logistic regression classifier using $\vec{\xi}_i$ as inputs and class labels as targets.

\STATE \textbf{Testing:}
\FOR{each test sample $i=1$ to $W_{\text{test}}$ in both classes}
    \STATE Generate intensity distribution as in the training stage.
    \STATE Simulate $S$ detection events to obtain empirical probabilities $\hat{P}_i(u_j)$.
    \STATE Compute empirical eigentask vector $\hat{\vec{\xi}}_i$.
    \STATE Classify $\hat{\vec{\xi}}_i$ using the trained classifier.
\ENDFOR
\STATE Repeat testing for multiple independent sets of test samples.
\STATE Compute mean and range of success probability $P_\text{succ}$ as a function of $\mathcal{K}$ and $S$.
\end{algorithmic}
\end{algorithm}

\begin{figure}[!tb]
\begin{center}
\includegraphics[width=0.8\columnwidth]{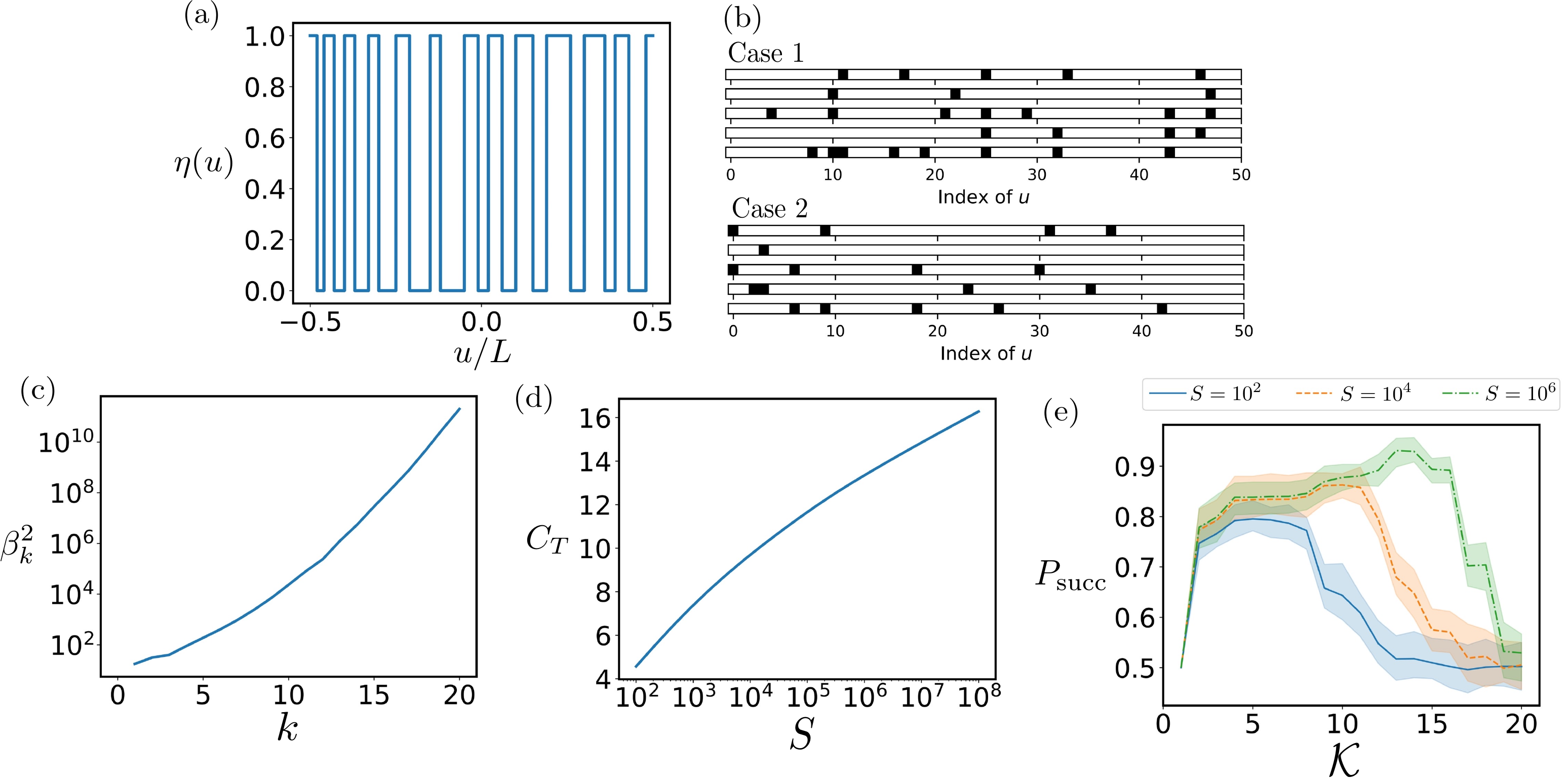}
\caption{\yk{(a) Indicator function $\eta(u)$ used for the simulation here. (b) Examples of intensity distributions randomly generated in the two cases; for demonstration purposes, we choose the number of pixels $N=50$. We set the probability of $I(u) = 1$ to be $p_0 = 0.2$ for each pixel where $\eta(u) = 1$. (c) The eigenvalues $\beta_k^2$ for different orders $k$. In the calculation, we use $N = 200$ pixels and set $L = 10$. A total of $2W_{\text{train}} = 400$ different intensity distributions are generated to compute the prior distribution and the corresponding $g$ and $d$. (d) The total REC $C_T$ as a function of the sample number $S$. 
(e) The success probability $P_{\text{succ}}$ for different sample sizes $S$ as a function of the highest order $\mathcal{K}$ of the eigentask included in the training and testing procedure. We use a test set of $2W_{\text{test}} = 600$ different intensity distributions, with $W_{\text{test}} = 300$ distributions for each case, to calculate the success probability of distinguishing between the two cases. This calculation is repeated 50 times using  $2W_{\text{test}} = 600$ independently generated intensity distributions each time. The line in the figure represents the mean success probability, and the shaded region indicates the maximum and minimum values across the 50 repetitions.}
} 
\label{simulation_example}
\end{center}
\end{figure}

\subsection{Simulation for imaging a source outside the Rayleigh limit}

We consider a class of distinguishing tasks using direct imaging. This serves as an additional simulation example not included in the main text; the simulation presented in the main text is detailed in the next subsection. The pseudocode for this simulation is provided in Algorithm \ref{Algorithm:general} and is further explained below.
The spatial intensity distribution $I(u)$ is defined over the interval $[-L/2, L/2]$ and takes binary values, $I(u) \in \{0,i_0\}$, where $i_0$ is fixed by the normalization condition $\int du\, I(u) = 1$, resembling the structure of a QR code. The density matrix of the state on the imaging plane is given in Eq.~\ref{eq:rho_basic}, which is reproduced below for convenience
\begin{equation}
\rho=\int du du_1 du_2 I(u)\psi(u-u_1)\ket{u_1}\bra{u_2}\psi(u-u_2),
\end{equation}
where the PSF is chosen as $\psi(u) = \exp(-u^2/4\sigma^2)/(2\pi\sigma^2)^{1/4}$ for all simulations in this section, with $\sigma = 1$. We consider the imaging of a source beyond the Rayleigh limit, i.e., $L $ can be greater than $ \sigma$.  For simplicity, the interval is discretized into $N$ equally spaced pixels, with the intensity assumed to be constant within each pixel. 


To generate random instances of $I(u)$, we introduce an indicator function $\eta(u) \in \{0,1\}$ defined over $[-L/2, L/2]$, one example is shown in the Fig.~\ref{simulation_example}(a). The function $\eta(u)$ specifies which spatial regions are ``active" (where $\eta(u) = 1$) and which are ``inactive" (where $\eta(u) = 0$). The generation of $I(u)$ differs depending on the following two cases:

Case 1 (active case): Intensity is allowed to appear only in the regions where $\eta(u) = 1$. That is, for each pixel whose centroid lies in a region where $\eta(u) = 1$, the corresponding value of $I(u)$ is independently set to 1 with activation probability $p_0$, and 0 otherwise. For pixels where $\eta(u) = 0$, we set $I(u) = 0$ deterministically. This simulates a sparse intensity pattern supported only within the $\eta(u)=1$ region.

Case 2 (inactive case): Intensity is allowed only in the complementary region where $\eta(u) = 0$. For each pixel with $\eta(u) = 0$ at its centroid, $I(u)$ is randomly set to 1 with probability $p_0$, and 0 otherwise. For pixels where $\eta(u) = 1$, we set $I(u) = 0$. This case represents an alternative intensity pattern that avoids the $\eta(u) = 1$ regions.

We provide example intensity distributions generated in the two cases in Fig.~\ref{simulation_example}(b).
The goal of the distinguishing task is to determine, based on the observed data on the image plane, whether a given realization of $I(u)$ was generated according to Case 1 or Case 2.


We first generate a training set of $2W_{\text{train}}$ intensity distributions according to the formalism described above, with $W_{\text{train}}$ samples for each case. The corresponding theoretical measurement outcomes are computed without sample noise. To model the task as a learning problem \cite{hu2023tackling}, we incorporate prior knowledge by treating the intensity values at each pixel as the input parameters for each distribution. Using these inputs, we construct the matrices $d_k$ and $g_{ij}$ by averaging over the all the $2W_{\text{train}}$ intensity distributions in the training set. This enables us to compute the eigentasks and their associated eigenvalues $\beta_k^2$.


Once the eigenvectors $r_k$, as defined in Eq.~\ref{SI_eq:eigentask}, are obtained with the prior quantified by $d_i$ and $g_{ij}$, and the pixel-wise detection probability distribution $P_m(u_j)$ for the $m$th intensity distribution is known theoretically, we construct the $k$th eigentask as $\xi_{km} = \sum_j r_{kj} P_m(u_j)$. This yields, for each intensity distribution indexed by $m$, an eigentask vector $\vec{\xi}_m = [\xi_{0m}, \xi_{1m}, \xi_{2m}, \dots, \xi_{\mathcal{K}m}]$. We truncate the eigentasks at order $\mathcal{K}$, which is determined based on the number of detection events $S$—a choice we will justify later. For each of the two classes, we generate $W_{\text{train}}$ such eigentask vectors, forming two datasets that will be used as training inputs for the logistic regression classifier.

We use logistic regression, implemented via the standard \texttt{LogisticRegression} class from the widely used Python package \texttt{scikit-learn}, to perform binary classification. Logistic regression models the probability that a sample belongs to class 1 as $P(y = 1 \mid \vec{x}) = {1}/({1 + \exp(-\vec{w}^T \vec{x})})$, where $\vec{x}$ is the input feature vector and $\vec{w}$ is the weight vector learned during training. A classification decision is made by thresholding this probability at 0.5. In our case, the input features $\vec{x}$ correspond to the eigentask vector $\vec{\xi}_m$, with each component normalized by dividing by its mean absolute value across the training set before being passed to the classifier. Since we know which class each eigentask vector $\vec{\xi}_m$ corresponds to during training, we use this standard python package to fit the weight vector $\vec{w}$ accordingly.

During testing, we again generate intensity distributions from the two classes. For each intensity distribution, we fix the total number of detected photons (sample number) to be $S$. By counting the number of photons detected in each pixel on the image plane, we obtain an empirical probability distribution $\hat{P}_m(u_j)$ for the $m$th intensity distribution. Using this distribution, we compute the corresponding empirical eigentask vector $\hat{\vec{\xi}}_m = [\hat{\xi}_{0m}, \hat{\xi}_{1m}, \hat{\xi}_{2m}, \dots, \hat{\xi}_{\mathcal{K}m}]$.
The choice of truncation order $\mathcal{K}$ depends on the precision with which each eigentask can be estimated. The eigenvalues $\beta_k^2$ characterize the sample size threshold required for accurately estimating each eigentask component; accurate estimation becomes possible once $S \gg \beta_k^2$. In Fig.~\ref{simulation_example}(c), we show the eigenvalues $\beta_k^2$ for the model considered here. A suitable choice of $\mathcal{K}$ satisfies $\beta_\mathcal{K}^2 \sim S$, ensuring that all retained eigentask components can be estimated reliably. The resulting eigentask vectors $\hat{\vec{\xi}}_m$ are then used for classification using the trained logistic regression model.

As in Fig.~\ref{simulation_example}(d), the success probability of distinguishing the two cases is plotted as a function of the number of included eigentasks $\mathcal{K}$ for various sample sizes $S$. As observed, the success probability initially increases with $\mathcal{K}$ and then decreases. This behavior is intuitive: increasing $\mathcal{K}$ allows more eigentasks to be included, capturing more information about the intensity distribution and thus improving classification performance. However, beyond a certain point, the higher-order eigentasks are poorly estimated due to limited sample size $S$. Naively including these noisy components degrades the classifier's performance and ultimately reduces the success probability. As shown in Fig.~\ref{simulation_example}(c), when $S = 10^2$, $10^4$, and $10^6$, we can roughly reliably estimate approximately the first 5, 10, and 14 eigentasks, respectively. As expected, Fig.~\ref{simulation_example}(d) shows that when we naively include the first $\mathcal{K}$ eigentasks, the success probability $P_{\text{succ}}$ peaks at around $\mathcal{K} = 5$, $10$, and $14$, respectively.

Fig.~\ref{different_prior_eigentask} shows that the prior distribution can significantly influence the form of the eigentasks, which in general do not have a simple form. 

\begin{figure}[!tb]
\begin{center}
\includegraphics[width=1\columnwidth]{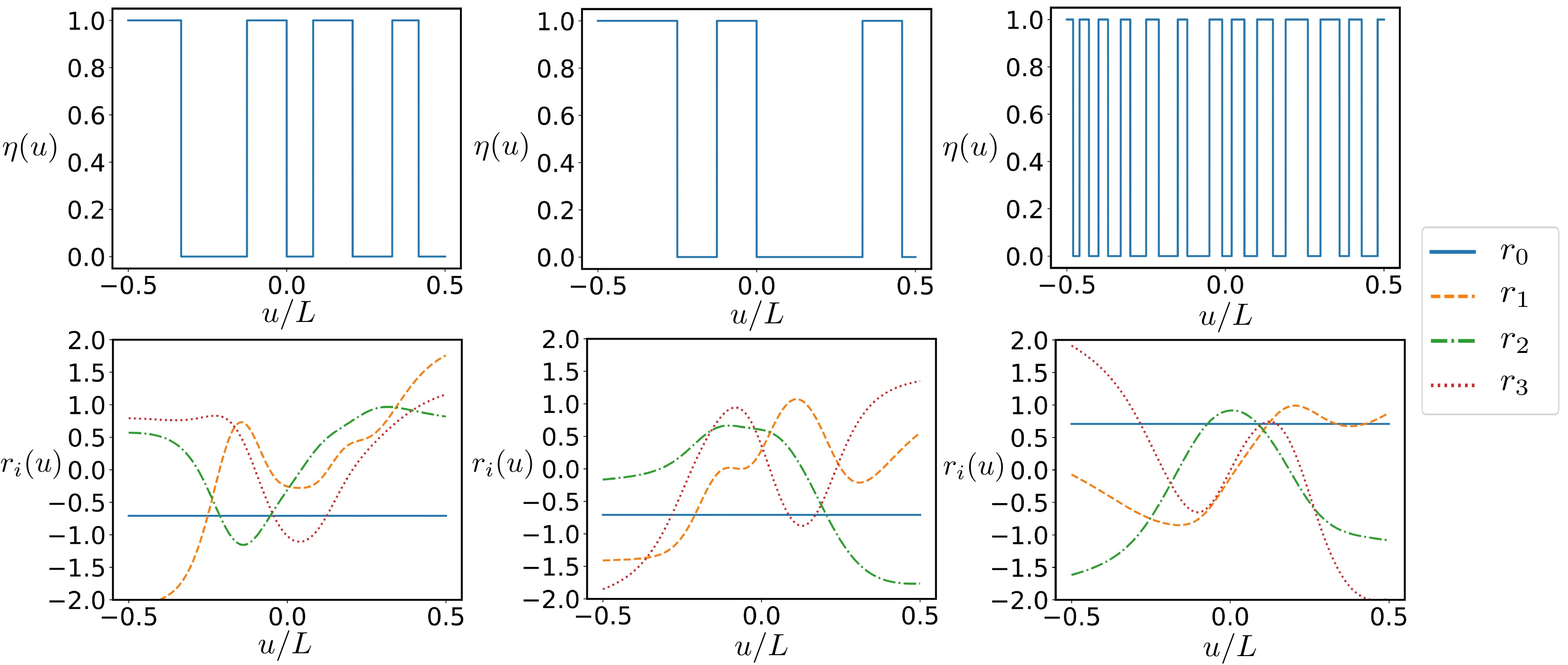}
\caption{\yk{Different indicator functions $\eta(u)$ lead to different eigentasks $r_k$, which are obtained by solving the generalized eigenvalue equation $G r_k = \lambda_k D r_k$ in Eq.~\ref{SI_eq:eigentask}. In this example, the direct imaging measurement yields the probability distribution of detecting a photon at each position $u$. The $k$th eigentask is given by $\sum_i r_k(u_i) P(u_i)$, where $P(u_i)$ is the probability of detecting a photon at position $u_i$. $L=10$, $\sigma=1$.}
} 
\label{different_prior_eigentask}
\end{center}
\end{figure}

\subsection{Imaging multiple compact sources}

\begin{algorithm}[H]
\caption{Distinguishing multiple compact sources using eigentasks}
\begin{algorithmic} \label{Algorithm:compact}
\STATE \textbf{Input:} Number of pixels $N$, total region size $\mathbb{L}$, activation probability $p_0$, number of compact sources $Q$, number of training samples per class $W_\text{train}$, number of test samples per class $W_\text{test}$, total detected photons per test image $S$, truncation order $\mathcal{K}$
\STATE Define the position function $\zeta(u)\in\{0,1\}$ to specify the spatial locations of the $Q$ compact sources.
\STATE For each compact source $S_\ell$ ($\ell=1,\dots,Q$), define an indicator function $\eta_\ell(u)\in\{0,1\}$ specifying the active region within that source.

\STATE \textbf{Training:}
\FOR{each class (Case 1: active; Case 2: inactive)}
    \FOR{$i=1$ to $W_\text{train}$}
        \STATE Initialize $I_i(u_j)=0$ for all pixels $j$.
        \FOR{each compact source $\ell=1,\dots,Q$}
            \FOR{each pixel $j$ where $\zeta(u_j)=1$ (i.e.\ within a compact source)}
                \IF{Case 1 and $\eta_\ell(u_j)=1$}
                    \STATE Set $I_i(u_j)=1$ with probability $p_0$, else $0$.
                \ELSIF{Case 2 and $\eta_\ell(u_j)=0$}
                    \STATE Set $I_i(u_j)=1$ with probability $p_0$, else $0$.
                \ENDIF
            \ENDFOR
        \ENDFOR
        \STATE Normalize $I_i(u)$ so that $\sum_{j=1}^N I_i(u_j)=1$, where each nonzero $I_i(u_j)$ takes value $i_0$.
    \ENDFOR
\ENDFOR

\STATE Compute prior matrices $d_k$ and $g_{ij}$ from all training samples, taking the prior to be the uniform distribution over the training set.

\STATE Solve the eigenvalue problem (Eq.~\ref{SI_eq:eigenproblem}) to obtain eigenvectors $r_k$ and eigenvalues $\beta_k^2$.

\FOR{each training sample $i$ in both classes}
    \STATE Compute theoretical detection probabilities $P_i(j)$ for the chosen measurement method (direct imaging, separate SPADE, orthogonalized SPADE).
    \STATE Compute eigentask components $\xi_{k,i} = \sum_j r_{kj} P_i(j)$.
    \STATE Form the truncated eigentask vector $\vec{\xi}_i = [\xi_{0i},\dots,\xi_{\mathcal{K}i}]$.
\ENDFOR

\STATE Train a logistic regression classifier using $\vec{\xi}_i$ as inputs and the class labels as targets.

\STATE \textbf{Testing:}
\FOR{each test sample $i=1$ to $W_{\text{test}}$ in both classes}
    \STATE Generate a new intensity distribution $I_i(u)$ as in the training phase.
    \STATE Simulate $S$ detection events to obtain empirical probabilities $\hat{P}_i(j)$.
    \STATE Compute empirical eigentask vector $\hat{\vec{\xi}}_i = [\hat{\xi}_{0i},\dots,\hat{\xi}_{\mathcal{K}i}]$.
    \STATE Classify $\hat{\vec{\xi}}_i$ using the trained logistic regression model.
\ENDFOR

\STATE Repeat the testing process for multiple independent sets of test samples.
\STATE Compute the mean and range of the success probability $P_\text{succ}$ as a function of $\mathcal{K}$ and $S$.
\end{algorithmic}
\end{algorithm}

We now consider a similar task: distinguishing multiple compact sources using the direct imaging method, the separate SPADE method, and the orthogonalized SPADE method. 
This corresponds to the simulation example presented in the main text.
The pseudocode for this simulation is provided in Algorithm \ref{Algorithm:compact} and is further explained below.

As before, we randomly generate images according to two different cases specified by the indicator function $\eta(u)$. In this setting of imaging multiple compact sources, we assume that the intensity is nonzero only at the locations of a few compact sources where the position function $\zeta(u)=1$, with sizes within the Rayleigh limit, i.e., $\alpha \ll 1$. As illustrated in Fig.~\ref{fig:intensity_examples}, we present an example with four compact sources defined by $\zeta(u)$. For each compact source, the corresponding indicator function $\eta(u)$ within its region is shown, along with several examples of the randomly generated intensity distributions for both cases and for each source.

As in the previous simulation of imaging a general source, we use $2W_{\text{train}}$ different intensity distributions to generate the corresponding $g_{ij}$ and $d_i$. This allows us to compute the eigentasks and the associated eigenvalues $\beta_k^2$. Once the eigenvectors $r_k$ are obtained, and the theoretical probability distribution $P_m(j)$ — corresponding to the $m$th intensity distribution and the $j$th outcome — is known, we construct the $k$th eigentask as
$\xi_{km} = \sum_j r_{kj} P_m(j)$.
Note that, for the direct imaging method, $P_m(j)$ represents the probability of detecting photons at position $u_j$. For the separate SPADE method, the outcomes are 
\begin{equation}\begin{aligned}\label{SI_eq:simulation_sSPADE}
&\frac{1}{2Q}\ket{b_{1,0}}\bra{b_{1,0}},\cdots, \frac{1}{2Q}\ket{b_{4,0}}\bra{b_{4,0}},\frac{1}{2Q}\ket{\phi_{1,0,\pm}}\bra{\phi_{1,0,\pm}},\cdots,\frac{1}{2Q}\ket{\phi_{4,0,\pm}}\bra{\phi_{4,0,\pm}},\\
&\frac{1}{2Q}\ket{\phi_{1,1,\pm}}\bra{\phi_{1,1,\pm}},\cdots,\frac{1}{2Q}\ket{\phi_{4,1,\pm}}\bra{\phi_{4,1,\pm}},\frac{1}{2Q}\ket{\phi_{1,2,\pm}}\bra{\phi_{1,2,\pm}},\cdots,\frac{1}{2Q}\ket{\phi_{4,2,\pm}}\bra{\phi_{4,2,\pm}},\cdots  
\end{aligned}
\end{equation}
which are defined in Eq. \ref{eq:POVM_separate}. And the orthogonalized SPADE method has outcomes in the following order
\begin{equation}\begin{aligned}\label{SI_eq:simulation_oSPADE}
&\frac{1}{2}\ket{b_{1}^{(0)}}\bra{b_{1}^{(0)}},\cdots,\frac{1}{2}\ket{b_{4}^{(0)}}\bra{b_{4}^{(0)}},\frac{1}{2}\ket{\phi_{1,\pm}^{(0)}}\bra{\phi_{1,\pm}^{(0)}},\cdots,\frac{1}{2}\ket{\phi_{4,\pm}^{(0)}}\bra{\phi_{4,\pm}^{(0)}},\\
&\frac{1}{2}\ket{\phi_{1,\pm}^{(1)}}\bra{\phi_{1,\pm}^{(1)}},\cdots,\frac{1}{2}\ket{\phi_{4,\pm}^{(1)}}\bra{\phi_{4,\pm}^{(1)}},\frac{1}{2}\ket{\phi_{1,\pm}^{(2)}}\bra{\phi_{1,\pm}^{(2)}},\cdots,\frac{1}{2}\ket{\phi_{4,\pm}^{(2)}}\bra{\phi_{4,\pm}^{(2)}},\cdots
\end{aligned}\end{equation} 
which are defined in Eq. \ref{eq:POVM_orthogonalized}.
This yields, for each intensity distribution indexed by $m$, an eigentask vector $\vec{\xi}_m = [\xi_{0m}, \xi_{1m}, \xi_{2m}, \dots, \xi_{\mathcal{K}m}]$. We truncate the eigentasks at order $\mathcal{K}$, which is determined by the number of detection events $S$—a choice we will justify later. For each of the two classes, we get $W_{\text{train}}$ such eigentask vectors, forming two datasets that serve as training inputs for the logistic regression classifier.

During testing, we generate $2W_{\text{test}}$ intensity distributions from the two cases. For each intensity distribution, we fix the total number of detected photons (sample number) to be $S$. By counting the number of occurrences of each outcome, we obtain an empirical probability distribution $\hat{P}_m(j)$ for the $m$th intensity distribution. Using this distribution, we compute the corresponding empirical eigentask vector $\hat{\vec{\xi}}_m = [\hat{\xi}_{0m}, \hat{\xi}_{1m}, \hat{\xi}_{2m}, \dots, \hat{\xi}_{\mathcal{K}m}]$. The resulting eigentask vectors $\hat{\vec{\xi}}_m$ are then used as input for classification with the trained logistic regression model.

We then show the performance of distinguishing the two cases using the three approaches in Fig.~\ref{fig:simulation_example_multiple}. In this calculation, we use the position function $\zeta(u)$ and the indicator function $\eta(u)$, as illustrated in Fig.~\ref{fig:intensity_examples}, to generate the intensity distributions for the two cases. We observe that the success probability $P_{\text{succ}}$ first increases with $\mathcal{K}$ and then decreases to around $50\%$ in Fig.~\ref{fig:simulation_example_multiple}(d1)-(d3). This behavior is expected because we should only include the eigentasks that can be reliably estimated given the sample number $S$.

Comparing the success probability $P_{\text{succ}}$ across the three approaches, since the distance between the compact sources is not much larger than the width of the PSF $\sigma$, we expect the separate SPADE method to perform comparably to the direct imaging case, while our orthogonalized SPADE method should achieve better performance. This is explicitly demonstrated in Fig.~\ref{fig:simulation_example_multiple}(d1)-(d3), where the peak success probability $P_{\text{succ}}$ for the orthogonalized SPADE method is clearly higher than that of the other two approaches when $S=10^{10}$. 

This simulation also illustrates the operational meaning of the total REC $C_T$: it roughly indicates the number of eigentasks that should be included in the fitting. This is because $\beta_k^2$ serves as an approximate threshold on the sample number required to reliably estimate the $k$th eigentask, and $C_T = \sum_k 1/(1+\beta_k^2/S)$, where each term $1/(1+\beta_k^2/S)$ approaches 1 when $S \gg \beta_k^2$. Hence, $C_T$ roughly counts how many eigentasks can be reliably estimated.
When the number of samples is $S=10^4$ or $S=10^7$, the total REC $C_T$ is roughly the same across the three approaches, and we observe similar behavior of $P_{\text{succ}}$ as a function of $\mathcal{K}$. However, at $S=10^{10}$, the total REC $C_T$ for the orthogonalized SPADE method becomes noticeably larger than that of the other two approaches, reaching roughly 13. This aligns precisely with the peak position in the $P_{\text{succ}}$ plot for the orthogonalized SPADE method. In comparison, for direct imaging and the separate SPADE method, $C_T$ is about 10, which is also consistent with the peak positions in their respective $P_{\text{succ}}$ plots. These observations clearly demonstrate the operational meaning of $C_T$, which reflects the number of eigentasks that should be included in the downstream analysis. Being able to include more reliably estimated features can improve the performance of the subsequent distinguishing task, whereas including eigentasks that are not well estimated can degrade the performance. A larger $C_T$ indicates that more eigentasks can be included in the downstream analysis, reflecting a better-performing imaging method.

The eigentasks here exhibit complex structures that depend on the positions of the compact sources $\zeta(u)$ and the indicator function $\eta(u)$. The first four eigenvectors for the orthogonalized SPADE method are shown in Fig.~\ref{fig:simulation_example_multiple}(c). This demonstrates the strength of the quantum learning approach in identifying which eigentasks should be included and determining up to which order they should be incorporated in the logistic regression for downstream analysis. From Fig.~\ref{fig:simulation_example_multiple}(a), we observe that the eigenvalues increase almost smoothly across all approaches considered here, rather than showing the stepwise increase we previously described. This is because the value of $\alpha=0.1$ in Fig.~\ref{fig:simulation_example_multiple}; if $\alpha$ were even smaller, the $\beta_k^2$ would be dominated by the scaling over $\alpha$. These aspects — the eigentask structure and the eigenvalue behavior — are discussed in detail in Sec.~\ref{SI:eigentask_multiple} of the Supplemental Material.

In the calculation shown in Fig.~\ref{fig:simulation_example_multiple}, we use $N_{\text{max}} = 20$ pixels for each compact source. A total of $2W_{\text{train}} = 400$ intensity distributions — with $W_{\text{train}} = 200$ per case — are generated to compute the prior distribution and the corresponding $g$ and $d$ for the eigenvalues $\beta_k^2$ and the total REC $C_T$. These $2W_{\text{train}} = 400$ intensity distributions, together with the corresponding theoretical eigentasks, are also used as the training data for logistic regression. A test set of $2W_{\text{test}} = 2000$ intensity distributions — with $W_{\text{test}} = 1000$ per case — is used to evaluate the success probability of distinguishing the two cases. This calculation is repeated 20 times, each with an independently generated test set of $2W_{\text{test}} = 2000$ intensity distributions.

\begin{figure}[!tb]
\begin{center}
\includegraphics[width=1\columnwidth]{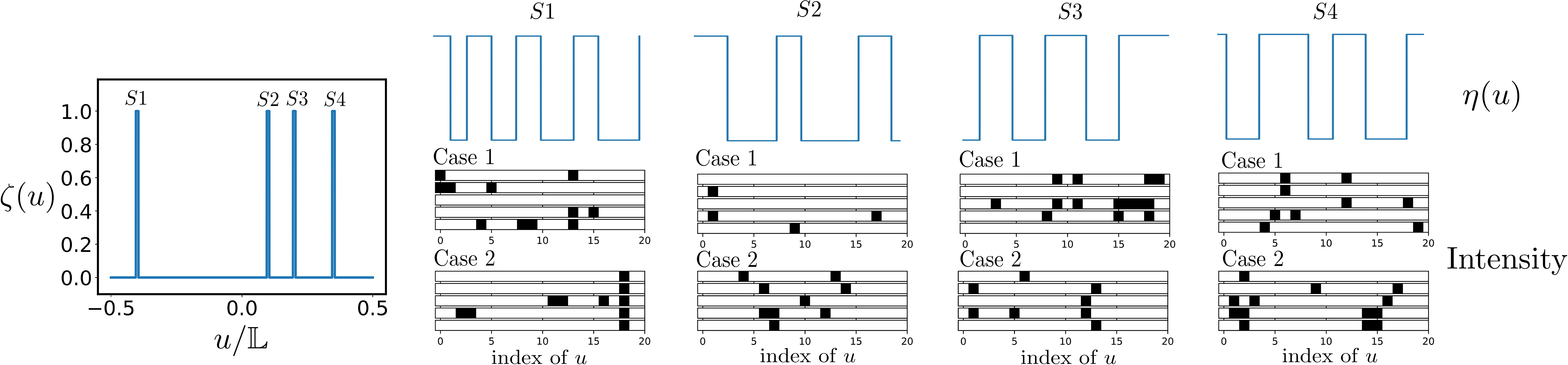}
\caption{\yk{The position function $\zeta(u)$ specifies the positions of four compact sources, labeled $S1$, $S2$, $S3$, and $S4$. For each compact source, the intensity distribution is randomly generated based on another indicator function $\eta(u)$. In case 1, the intensity is nonzero only where $\eta(u) = 1$, whereas in case 2, it is nonzero only where $\eta(u) = 0$.  All four compact sources are confined within the interval $[-\mathbb{L}/2, \mathbb{L}/2]$ with $\mathbb{L} = 10$, and each source has a size of $0.1$, which is smaller than the PSF width $\sigma = 1$. The intensity distributions are randomly generated, and we present five examples of both cases for each compact source above. Note that the intensity of each pixel is assigned based on the position of its centroid, so edge pixels may appear slightly misaligned with the indicator function $\eta(u)$; however, this effect becomes negligible as the number of pixels increases.}
} 
\label{fig:intensity_examples}
\end{center}
\end{figure}

\begin{figure}[!tb]
\begin{center}
\includegraphics[width=1\columnwidth]{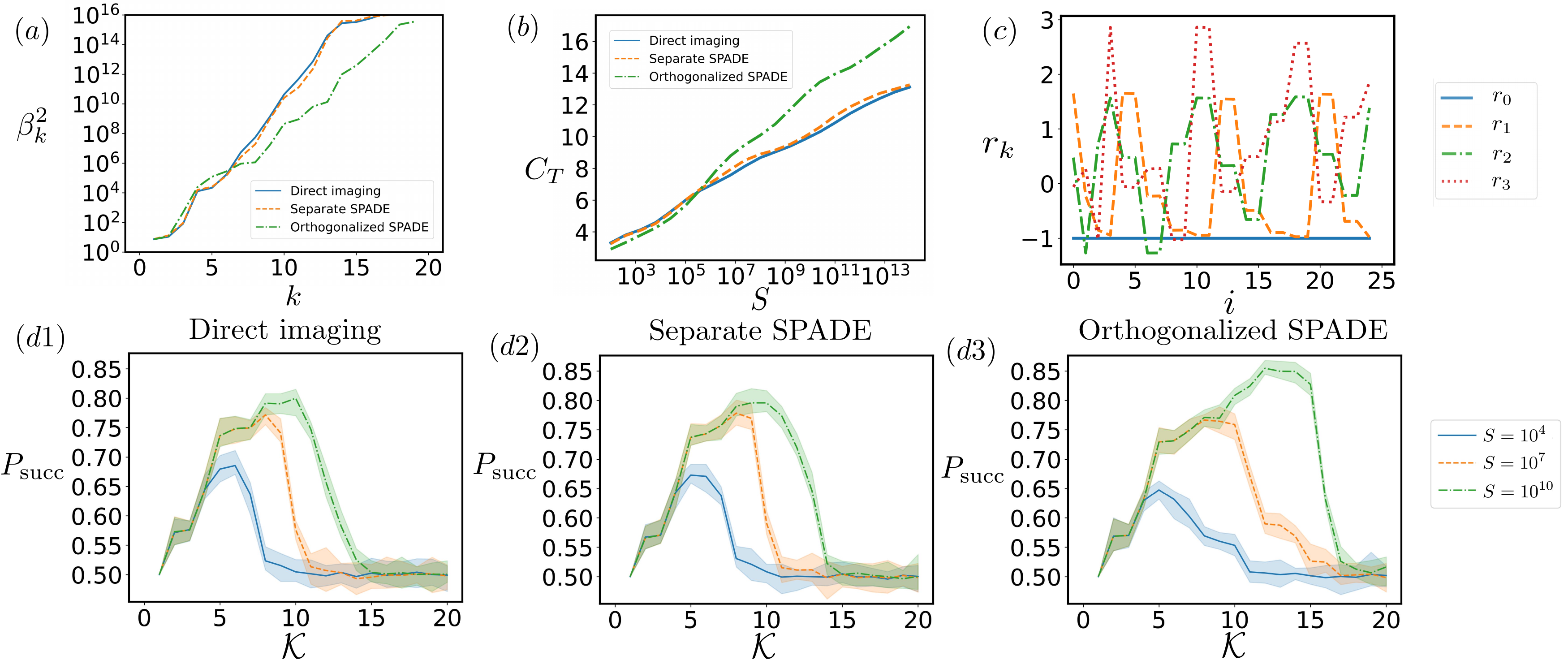}
\caption{\yk{(a) The eigenvalues $\beta_k^2$ for different orders $k$ obtained using the three approaches to image four compact sources.
(b) The total REC $C_T$ as a function of the sample number $S$ for the three approaches.
(c) The first four eigenvectors $r_k$ for the orthogonalized SPADE method, where $i$ labels the $i$th outcome. 
(d1)–(d3) The success probability $P_{\text{succ}}$ as a function of the highest order $\mathcal{K}$ of the eigentasks included in the training and testing procedures for classifying the images in the two cases, evaluated at different sample sizes $S$. The solid line represents the mean success probability, and the shaded region shows the maximum and minimum values across the 20 repetitions. We set $\alpha = 0.1$ in this figure.}
} 
\label{fig:simulation_example_multiple}
\end{center}
\end{figure}

\subsection{Face recognition}\label{SI:face}

\begin{figure}[!tb]
\begin{center}
\includegraphics[width=0.7\columnwidth]{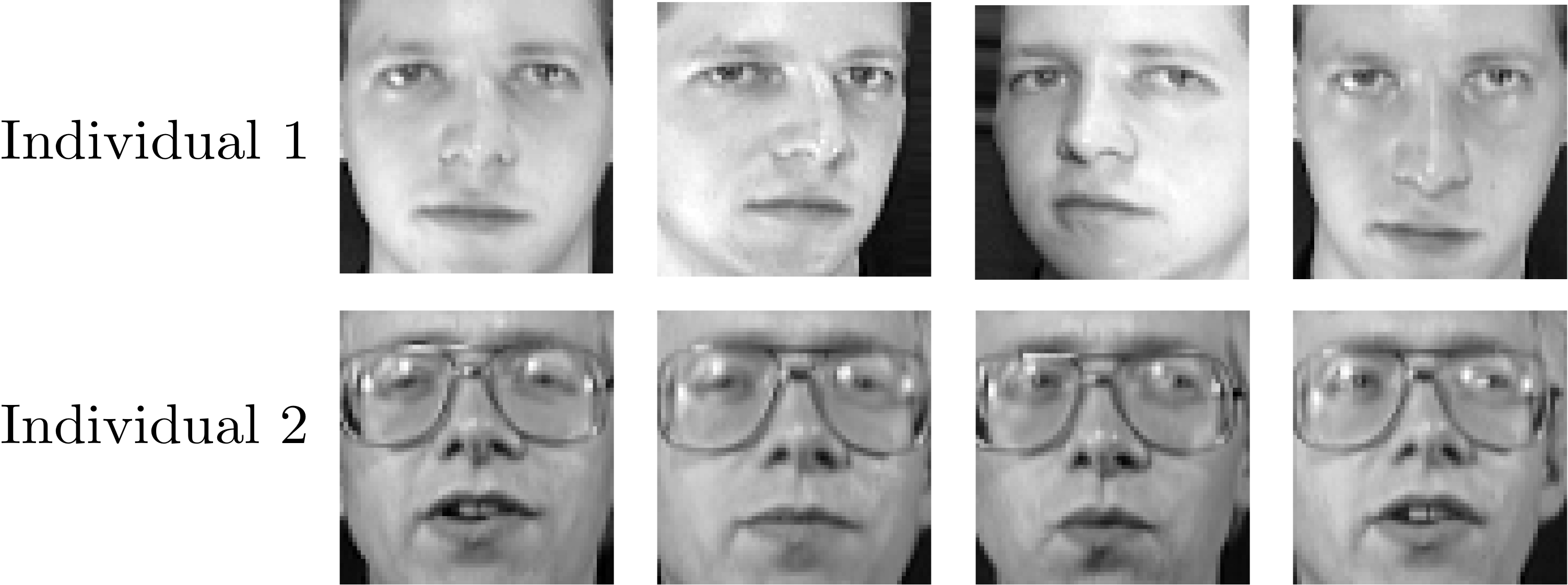}
\caption{\ykwang{Example images of two individuals from the Olivetti Faces dataset. Each row shows several images belonging to the same individual, illustrating natural variations in expression, pose, and illumination. These samples form the input set for the face-recognition task considered in this work.
} }
\label{fig:face_example}
\end{center}
\end{figure}

\begin{figure}[!tb]
\begin{center}
\includegraphics[width=1\columnwidth]{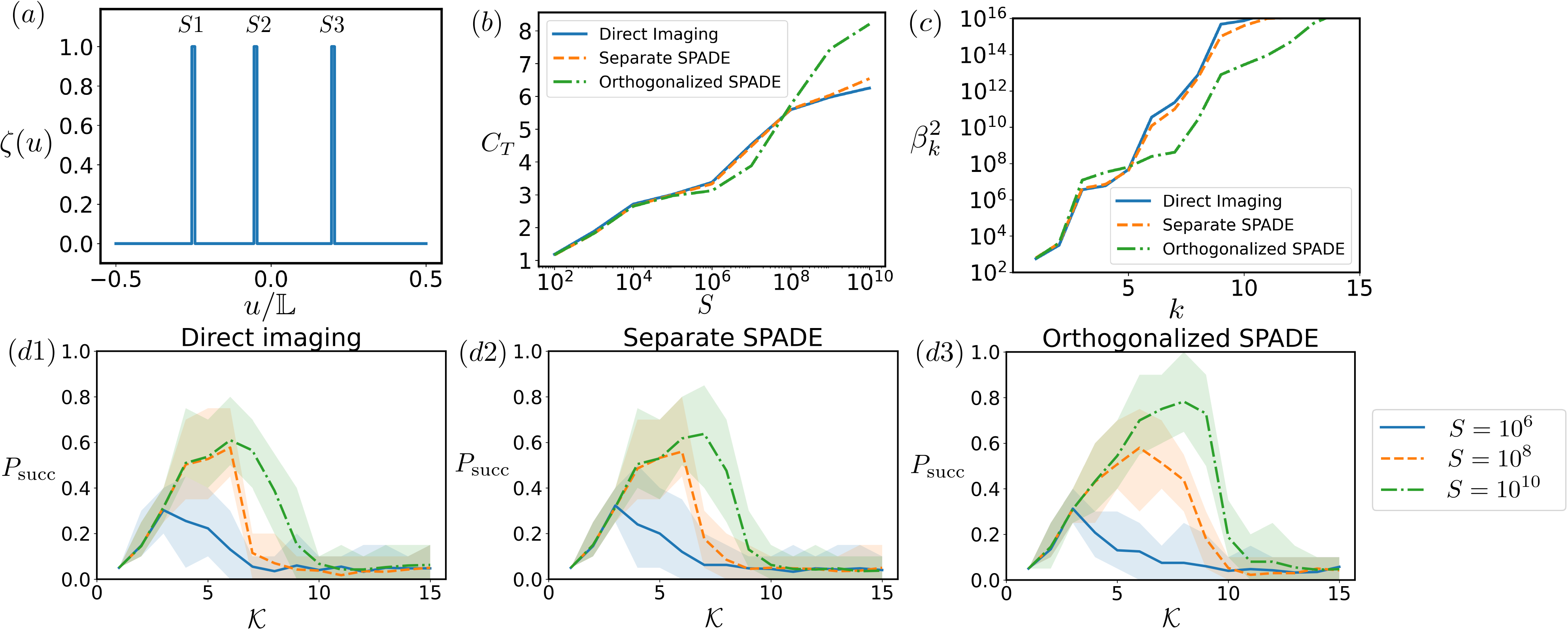}
\caption{\ykwang{Simulation for face recognition. (a) The position function $\zeta(u)$ specifies the positions of three compact sources, labeled $S1$, $S2$, and $S3$. All three compact sources are confined within the interval $[-\mathbb{L}/2, \mathbb{L}/2]$ with $\mathbb{L} = 10$, and each source has a size of $0.1$, which is smaller than the PSF width $\sigma = 1$. (b) The eigenvalues $\beta_k^2$ for different orders $k$ obtained using the three approaches.
(c) The total REC $C_T$ as a function of the sample number $S$ for the three approaches.
(d1)–(d3) The success probability $P_{\text{succ}}$ as a function of the highest order $\mathcal{K}$ of the eigentasks included in the training and testing procedures for face recognition, evaluated at different sample sizes $S$. The lines represent the mean success probability, and the shaded region shows the maximum and minimum values across the 20 repetitions. We set $\alpha = 0.1$ in this figure.
} }
\label{fig:face_recognition}
\end{center}
\end{figure}

\begin{algorithm}[H]
\caption{Face-recognition simulation using eigentasks}
\begin{algorithmic} \label{Algorithm:face_recognition}

\STATE \textbf{Input:} Number of individuals $N_{\mathrm{person}}$, training images per person $w_{\mathrm{train}}$, test images per person $w_{\mathrm{test}}$, number of compact sources $M$, total region size $\mathbb{L}$, PSF width $\sigma$, total detected photons per test image $S$, truncation order $\mathcal{K}$, measurement type (direct imaging, separate SPADE, orthogonalized SPADE)

\STATE \textbf{Preprocessing:}
\FOR{each face image}
    \STATE Load the $64\times 64$ grayscale face image.
    \STATE Vectorize the image into a 1D array of $N_{\mathrm{pix}}=4096$ pixels.
    \STATE Split the array into $M$ consecutive segments of equal length.
    \STATE Assign segment $m$ to compact source $S_m$ using the position function $\zeta(u)$.
    \STATE Construct the 1D intensity distribution $I(u)$ and normalize it so $\sum_j I(u_j)=1$.
\ENDFOR

\STATE \textbf{Training:}
\STATE Randomly select $w_{\mathrm{train}}$ training images for each of the $N_{\mathrm{person}}$ individuals.
\FOR{each training image $i$}
    \STATE Construct $I_i(u)$ following the preprocessing steps.
\ENDFOR

\STATE Compute prior matrices $d_k$ and $g_{ij}$ from all training intensity distributions, taking the prior to be uniform.

\STATE Solve the eigenvalue problem (Eq.~\ref{SI_eq:eigenproblem}) to obtain eigenvectors $r_k$ and eigenvalues $\beta_k^2$.

\FOR{each training image $i$}
    \STATE Compute theoretical detection probabilities $P_i(j)$ for the chosen measurement method.
    \STATE Compute eigentask components $\xi_{k,i} = \sum_j r_{kj} P_i(j)$ for $k=0,\dots,\mathcal{K}$.
    \STATE Form the truncated eigentask vector $\vec{\xi}_i = [\xi_{0i},\dots,\xi_{\mathcal{K}i}]$.
    \STATE Assign the class label corresponding to the identity of the person.
\ENDFOR

\STATE Train a logistic regression classifier using the eigentask vectors $\vec{\xi}_i$ as inputs and person identities as labels.

\vspace{0.5em}
\STATE \textbf{Testing:}
\STATE Select $w_{\mathrm{test}}$ test images per person.

\FOR{each test image $i$}
    \STATE Construct a new 1D intensity distribution $I_i(u)$ as in the training phase.
    \STATE Simulate $S$ detection events based on $I_i(u)$ and the chosen measurement to obtain empirical probabilities $\hat{P}_i(j)$.
    \STATE Compute empirical eigentask components $\hat{\xi}_{k,i} = \sum_j r_{kj} \hat{P}_i(j)$.
    \STATE Form the empirical eigentask vector $\hat{\vec{\xi}}_i = [\hat{\xi}_{0i},\dots,\hat{\xi}_{\mathcal{K}i}]$.
    \STATE Classify $\hat{\vec{\xi}}_i$ using the trained logistic regression model.
\ENDFOR

\STATE Repeat the testing process using randomly selected training and test sets.
\STATE Compute the mean and range of the success probability $P_{\mathrm{succ}}$ as functions of $\mathcal{K}$ and $S$.

\end{algorithmic}
\end{algorithm}

\ykwang{
In this subsection, we present a simulation example of face recognition in the context of imaging below the Rayleigh limit, using all three measurement strategies: the direct imaging method, the separate SPADE method, and the orthogonalized SPADE method. The pseudocode for this simulation is provided in Algorithm \ref{Algorithm:face_recognition} and is further explained below.

As shown in Fig.~\ref{fig:face_example}, the face images used in our simulation are taken from the Olivetti Faces dataset provided by AT\&T Laboratories Cambridge and distributed through \texttt{scikit-learn}. Each individual has ten $64\times 64$ grayscale images with different facial expressions and small variations in pose. We use these images as the training and testing data for the face-recognition task considered in this simulation.
Our theoretical analysis is developed for one-dimensional intensity distribution, whereas the face images are inherently two-dimensional. While it should be possible to extend our framework to the fully two-dimensional setting, we leave this to future work, as most superresolution techniques generalize naturally to higher dimensions without introducing additional qualitative differences \cite{zhou2019modern,ang2017quantum,lupo2020quantum,yu2018quantum}. Reducing a two-dimensional image to a one-dimensional signal through rasterization is also a standard dimensionality-reduction technique in machine learning and image processing.
To embed the data into our one-dimensional model, each $64\times 64$ face image is converted into a one-dimensional signal by rasterizing the image: all pixel rows are concatenated sequentially to form a vector of length $4096$. We partition this vector into $M=3$ segments. These segments play the role of compact sources and are placed at the positions specified by the position function $\zeta(u)$, as illustrated in Fig.~\ref{fig:face_example}(a). In this way, each two-dimensional face image is embedded into the one-dimensional framework required for our superresolution analysis.

We use $w_{\text{train}}=9$ training images per person and consider a total of $N_{\text{person}}=20$ individuals. As in the simulation of the previous subsections, this gives $W_{\text{train}}=w_{\text{train}}N_{\text{person}}=180$ training face images corresponding to different people. These training data are used to generate the quantities $g_{ij}$ and $d_i$, which in turn allow us to compute the eigentasks and the associated eigenvalues $\beta_k^2$. Once the eigenvectors $r_k$ are obtained, and the theoretical probability distribution $P_m(j)$—corresponding to the $m$th intensity distribution and the $j$th outcome—is known, we construct the $k$th eigentask as
$\xi_{km} = \sum_j r_{kj} P_m(j)$.
For the direct imaging method, $P_m(j)$ represents the probability of detecting photons at position $u_j$. For the separate SPADE method, the outcomes are given in Eq.~\ref{SI_eq:simulation_sSPADE}. The orthogonalized SPADE method uses the outcomes in the order specified in Eq.~\ref{SI_eq:simulation_oSPADE}.
This procedure yields, for each intensity distribution indexed by $m$, an eigentask vector $\vec{\xi}_m = [\xi_{0m}, \xi_{1m}, \xi_{2m}, \dots, \xi_{\mathcal{K}m}]$. We truncate the eigentasks at order $\mathcal{K}$, which is set by the number of detection events $S$—a choice we will justify later. For each individual, we obtain $w_{\text{train}}$ such eigentask vectors, forming  datasets that serve as the training inputs for the logistic regression classifier.


During testing, we use $w_{\text{test}}=1$ image per person and consider the same $N_{\text{person}}=20$ individuals, giving a total of $W_{\text{test}}=w_{\text{test}}N_{\text{person}}=20$ test face images corresponding to different people. For each intensity distribution, we fix the total number of detected photons (sample number) to be $S$. By counting the occurrences of each outcome, we obtain an empirical probability distribution $\hat{P}_m(j)$ for the $m$th intensity distribution. Using this distribution, we then compute the corresponding empirical eigentask vector $\hat{\vec{\xi}}_m = [\hat{\xi}_{0m}, \hat{\xi}_{1m}, \hat{\xi}_{2m}, \dots, \hat{\xi}_{\mathcal{K}m}]$, which is also truncated at order $\mathcal{K}$. These empirical eigentask vectors $\hat{\vec{\xi}}_m$ are used as inputs to the trained logistic regression classifier.


We then evaluate the performance of identifying these $N_{\text{person}}=20$ individuals in the context of imaging below the Rayleigh limit using the three approaches, as shown in Fig.~\ref{fig:face_recognition}. We observe that the success probability $P_{\text{succ}}$ first increases with $\mathcal{K}$ and then decreases, as illustrated in Fig.~\ref{fig:face_recognition}(d1)–(d3). This trend is expected: only eigentasks that can be reliably estimated given the sample number $S$ should be included.
Comparing the success probability $P_{\text{succ}}$ across the three approaches, since the distance between the compact sources is not much larger than the width of the PSF $\sigma$, we expect the separate SPADE method to perform comparably to the direct imaging case, while our orthogonalized SPADE method should achieve better performance. This is explicitly demonstrated in Fig.~\ref{fig:face_recognition}(d1)-(d3), where the peak success probability $P_{\text{succ}}$ for the orthogonalized SPADE method is clearly higher than that of the other two approaches when $S=10^{10}$.

This simulation also illustrates the operational meaning of the total REC $C_T$: it approximately indicates how many eigentasks should be included in the fitting. The reason is that $\beta_k^2$ serves as a threshold on the sample number required to reliably estimate the $k$th eigentask, and $C_T = \sum_k 1/(1+\beta_k^2/S)$, where each term $1/(1+\beta_k^2/S)$ approaches 1 when $S \gg \beta_k^2$. Thus, $C_T$ effectively counts the number of eigentasks that can be reliably estimated. From Fig.~\ref{fig:face_recognition}(c), we observe that the eigenvalues $\beta_k^2$ for the orthogonalized SPADE method are significantly smaller than those of the other two methods for approximately $ k\geq 6$.
When the number of samples is $S=10^6$ or $S=10^8$, the total REC $C_T$ is roughly the same across all three approaches, and the resulting dependence of $P_{\text{succ}}$ on $\mathcal{K}$ is correspondingly similar. However, at $S=10^{10}$, the total REC $C_T$ for the orthogonalized SPADE method increases noticeably, reaching a value of roughly 8, while the corresponding values for direct imaging and separate SPADE remain around 6. This behavior matches the peak positions of $P_{\text{succ}}$ in Fig.~\ref{fig:face_recognition}(d1)–(d3): the orthogonalized SPADE method peaks near $\mathcal{K}\approx 8$, while the other two methods peak near $\mathcal{K}\approx 6$.
These observations clearly demonstrate the operational meaning of $C_T$: it reflects the number of eigentasks that should be included in the downstream analysis. Being able to incorporate more reliably estimated features improves the performance of the face recognition, whereas including poorly estimated eigentasks degrades it. A larger $C_T$ therefore indicates that more eigentasks can be used effectively, corresponding to a better-performing imaging method.

In the calculation shown in Fig.~\ref{fig:face_recognition}(d1)–(d3), a total of $W_{\text{train}} = 180$ intensity distributions, with $w_{\text{train}} = 9$ per person, are generated to compute the prior distribution and the corresponding $g_{ij}$ and $d_i$ for determining the eigenvalues $\beta_k^2$ and the total REC $C_T$. These $W_{\text{train}} = 180$ intensity distributions, together with their theoretical eigentasks, are also used as the training data for logistic regression. A test set of $W_{\text{test}} = 20$ intensity distributions, with $w_{\text{test}} = 1$ per person, is used to evaluate the success probability of identifying each individual.
This calculation is repeated 20 times. In each repetition, the $w_{\text{test}} + w_{\text{train}} = 10$ images available for each person are randomly split so that 9 images are used for training and 1 image is used for testing. Repeating the simulation with different random choices of test images prevents the result from depending on any particular image and provides a more reliable estimate of the overall performance.
}

\subsection{Discussion}

Here, we want to discuss the key advantages of applying the quantum learning formalism \cite{hu2023tackling} to imaging problems, in comparison to traditional quantum parameter estimation approaches based on the Cramer-Rao bound. A central insight is that practical imaging tasks involve intrinsic ambiguities: (i) parameterization—it is unclear which set of parameters best captures the task-relevant information, and different choices can lead to markedly different performance; (ii) finite-sample limitations—with limited data, only a subset of parameters can be reliably estimated, making it crucial to identify and retain well-estimated ones, which is especially challenging in the inherently infinite-dimensional setting of imaging. While quantum parameter estimation remains a powerful tool when the goal is to estimate well-defined parameters encoded in a system, this quantum learning approach is better suited for more general tasks that involve prior structure, model uncertainty, and finite sample limitations. We believe this approach opens the door to a range of practically important applications where traditional estimation strategies may fall short.

A core difficulty in analyzing imaging problems is that imaging inherently involve a large number of degrees of freedom. Conventional parameter estimation provides no clear guidance on which parameters to extract or how many to include in downstream analyses, such as logistic regression in our simulation. In this simulation, our objective is to distinguish between two cases, rather than to estimate a well-defined set of parameters. As such, it is unclear which quantities should be treated as parameters, and if so, up to what order. For example, one might consider the intensity values at each pixel as unknown parameters, or alternatively use the Fourier components or the moments of the intensity distribution—all of which, in principle, contain the full information and can be incorporated into the Fisher information framework. Furthermore, if moments are chosen as features, there remains ambiguity regarding how many moments to include. If we rely solely on Fisher information or traditional parameter estimation approaches, such questions must be addressed carefully on a case-by-case basis, lacking both conceptual simplicity and a unified theoretical foundation.
In fact, we view this ambiguity as a significant drawback in earlier discussions of superresolution based on Fisher information \cite{pirandola2018advances,sorelli2021optimal,grace2020approaching,tsang2019quantum,tsang2017subdiffraction,zhou2019modern,wang2021superresolution,nair2016far,lupo2016ultimate,napoli2019towards,yu2018quantum,ang2017quantum,yang2016farfield,tang2016fault,paur2016achieving,tham2017beating,parniak2018beating,zanforlin2022optical,santamaria2024single,tan2023quantum,rouviere2024ultra,tan2023quantum}, where a fixed set of parameters is assumed from the outset without a clear justification for their selection.


By contrast, the quantum learning framework identifies a natural basis—the eigentasks—derived from the measurement model and prior knowledge, which serve as meaningful and physically motivated features for downstream tasks such as logistic regression. While parameters such as pixel intensities, Fourier components, or moments can be chosen as inputs, the features that can actually be estimated should be determined by the measurement system itself, together with the known prior distribution over the inputs. These are precisely the eigentasks. Importantly, the total REC and the eigentasks are invariant under reparameterization, as explicitly proven in Proposition~\ref{prop:reparameterization} of the Supplemental Material. 
The eigentasks span the space of functions expressible as linear combinations of the measured features. If a subset of them can be estimated with sufficient accuracy, then any learnable function must lie within their span, making them a principled and well-justified choice for regression.
Crucially, the number of detection events $S$ determines how many eigentasks can be reliably estimated, and thus which features should be retained for classification. This is particularly important in superresolution problems, where the performance exhibits step-like behavior closely tied to the size of the sources. As a result, precision quantified solely by Fisher information is inadequate for guiding practical decisions—particularly when the prior is complex, the sample size is limited, or the task is not explicitly parameter estimation.

\section{Comparison with the Chernoff bound}\label{SI:Chernoff}

\begin{figure}[!tb]
\begin{center}
\includegraphics[width=0.9\columnwidth]{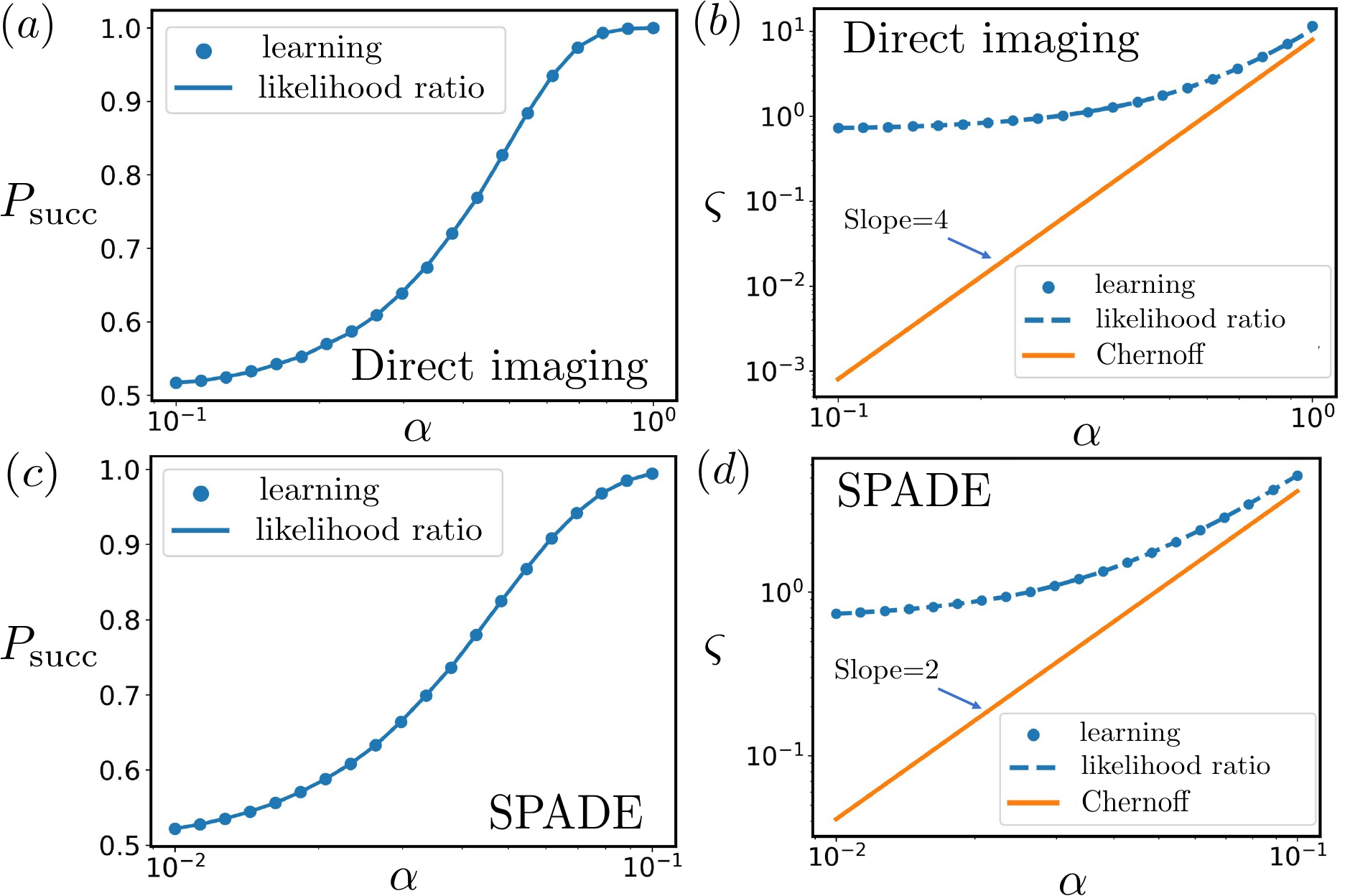}
\caption{\ykwang{Comparison of the learning approach, likelihood-ratio method, and Chernoff bound for discriminating a single point source from two point sources.
(a), (b) Direct imaging with sample number $S = 10^7$. Panel (a) shows the success probability $P_{\text{succ}}$ for the three approaches, and panel (b) shows the corresponding exponential index $\varsigma = -\log(1 - P_{\text{succ}})$.
(c), (d) SPADE measurements with sample number $S = 10^6$. Panel (c) shows the success probability $P_{\text{succ}}$, and panel (d) shows the exponential index $\varsigma = -\log(1 - P_{\text{succ}})$.
The legend ``learning'' denotes the learning-based approach introduced in this work, ``likelihood ratio'' denotes the likelihood-ratio test, and ``Chernoff'' shows the theoretical Chernoff bound. In all panels, the single-source and two-source hypotheses occur with equal prior probability, and the two point sources have equal brightness with separation $\alpha d_{\text{sep}}$, where $d_{\text{sep}} = 0.6$ and the PSF width is $\sigma = 1$. $\alpha$ and $\varsigma$ are plotted in log scale, while $P_{\text{succ}}$ is plotted in linear scale.
} }
\label{fig:chernoff}
\end{center}
\end{figure}

\begin{figure}[!tb]
\begin{center}
\includegraphics[width=1\columnwidth]{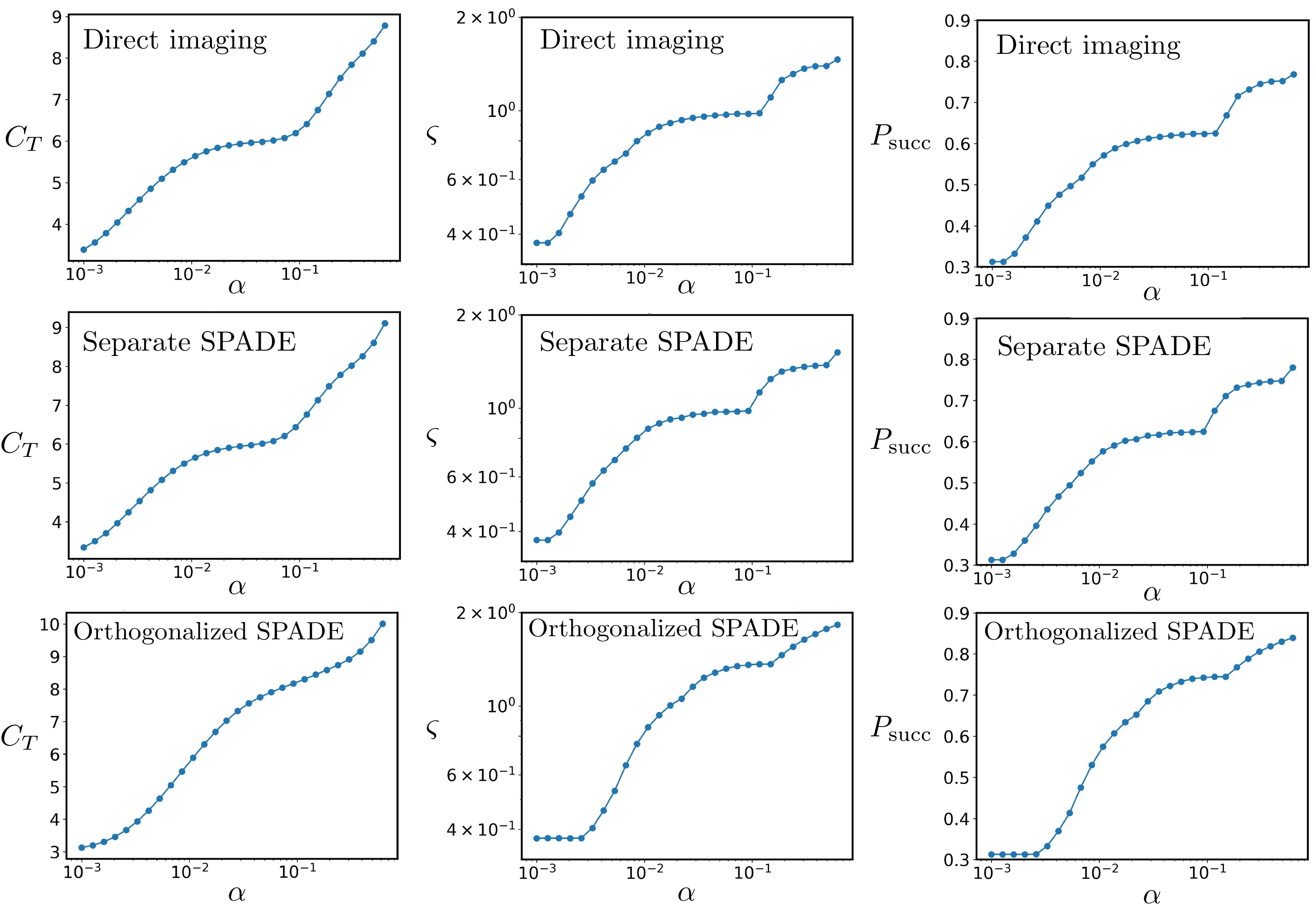}
\caption{\ykwang{Simulation of face-recognition performance as a function of the separation parameter $\alpha$ using the the learning-based approach, comparing three measurement strategies. The task involves identifying an individual from a set of 20 people, where each person has 9 training images and 1 additional image reserved for testing. Each face image is split into three compact pieces, consistent with the simulation setup of the previous section. All panels use $S = 10^{10}$ samples for each method. For every method, we report the total REC $C_T$, the success probability $P_{\text{succ}}$, and the corresponding exponential index $\varsigma = -\log(1 - P_{\text{succ}})$. The quantities $P_{\text{succ}}$ and $\varsigma$ are obtained by running $100$ repetitions and averaging the resulting success probabilities. } }
\label{fig:rate_face}
\end{center}
\end{figure}

\ykwang{
Many machine-learning tasks can be viewed as discrimination problems. It is therefore natural to compare learning-based discrimination with the predictions obtained from more traditional statistical tools. In this section, we compare the learning approach with the likelihood-ratio method and with the corresponding Chernoff bounds.


Before going into details, we first emphasize that although many machine-learning tasks can be formulated as discrimination problems, this does not trivialize the learning-based approach. A key advantage of learning methods is their ability to handle highly complex image structures. In contrast, the likelihood-ratio method typically requires an accurate statistical model for the objects being imaged. For example, in a face-recognition task involving tens of individuals, it is essentially impossible to write down a closed-form likelihood function for the measured data that would allow an exact likelihood-ratio calculation. Learning-based approaches naturally excel in such high-dimensional, structured tasks, as is evident from their widespread use in everyday face-recognition systems, even though such tasks are, in principle, solvable by an optimal likelihood-ratio test.
Our quantum learning method inherits these advantages while also being able to exploit the underlying structure of quantum measurements and quantum states, allowing it to operate effectively in the quantum-imaging setting considered here.

It is difficult to directly compare the likelihood-ratio method, the learning-based method, and the corresponding Chernoff bounds for the complex face-recognition example presented in the main text. We therefore begin by carrying out this comparison in a simpler and well-studied setting: distinguishing between a single point source and two point sources, a central problem in the theory  of superresolution \cite{zanforlin2022optical,huang2021quantum,zhang2020super,lu2018quantum,grace2022identifying}. We then use the insights gained from this simple example to discuss the behavior of the quantum learning–based method in the face-recognition task. 
For this elementary discrimination problem, previous works have predicted that for direct imaging, the error probability scales as $P_{\text{err}}\sim \exp(-\alpha^4 S)$ up to constant prefactors, whereas the SPADE measurement achieves $P_{\text{err}}\sim \exp(-\alpha^2 S)$, where $S$ is the sample number. The error probability is simply related to the success probability by $P_{\text{err}}=1-P_{\text{succ}}$. These scalings follow from the Chernoff bound \cite{audenaert2007discriminating,chernoff1952measure,cover1999elements}, which upper-bounds the decay rate of the error probability via $P_{\text{err}}\sim\exp(-S C)$ with $C = -\log\left(\min_{0\le t\le 1}\sum_i p_0(i)^{t} p_1(i)^{1-t}\right)$,  where $p_{0,1}(i)$ denote the probabilities of obtaining outcome $i$ under hypotheses 0 and 1, respectively.
As shown in Fig.~\ref{fig:chernoff}, our numerically calculated Chernoff bound exhibits the same scaling over $\alpha$ as reported in previous work \cite{zanforlin2022optical,huang2021quantum,zhang2020super,lu2018quantum,grace2022identifying}.

We further evaluate the performance of the likelihood ratio method, which compares how likely the observed data are under two competing hypotheses and selects the hypothesis with the higher likelihood. The likelihood ratio method decides between two hypotheses $H_0$ and $H_1$ by computing
$\Lambda(\vec{x}) = {P(\vec{x} \mid H_1)}/{P(\vec{x} \mid H_0)}$,
and choosing $H_1$ if $\Lambda(\vec{x}) > 1$, and $H_0$ otherwise, where $\vec{x}$ denotes the outcomes of $S$ samples. To apply this method, we first construct the success-probability distributions for imaging a single point source and two point sources, corresponding to the two hypotheses in each measurement method, which yields explicit formulas for $P(\vec{x} \mid H_0)$ and $P(\vec{x} \mid H_1)$. Based on any observed sample outcome $\vec{x}$, we then compute $\Lambda(\vec{x})$ to make the prediction.

We also provide the performance of distinguishing the two hypotheses using the learning approach developed in our work in Fig. \ref{fig:chernoff}. In this calculation, we assume prior information that the single–point–source case and the two–point–source case occur with equal probability, which allows us to compute $D$, $G$, and the corresponding eigenvectors and eigenvalues based on Eq. \ref{SI_eq:eigenproblem}. Following Sec. \ref{SI:demonstrative_example} of the Supplemental Material, these eigenvectors are then used to construct the eigentasks $\xi_{ik} = \sum_j r_{kj} {P}_i(j)$, where ${P}_i(j)$ is the probability distribution for the $j$th outcome obtained from the $i$th test instance, which may correspond to either a single point source or two point sources.
We use these truncated eigentask vectors $\vec{\xi}_i = [\xi_{0i}, \dots, \xi_{\mathcal{K}i}]$ as the input features for training the logistic regression model. To avoid the instability caused by having very few real samples in this simple example, we employ model-based data augmentation that generates synthetic eigentask features around each class’s prototype model, yielding a sufficiently large and robust training set.
Later, for each test instance, we sample the $S$ outcomes to obtain the empirical probability distribution $\hat{P}_i(j)$ and compute the corresponding empirical eigentask $\hat{\xi}_{ik} = \sum_j r_{kj} \hat{P}_i(j)$, which is then used as the input to the trained classifier for evaluation. Sec. \ref{SI:demonstrative_example} in the Supplemental Material has shown that the success probability depends on the truncation order $\mathcal{K}$. Hence, we vary $\mathcal{K}$ to compute the corresponding success probability $P_{\text{succ}}$ and select the peak value across all truncation orders. The plot reports the peak success probability achieved by this learning-based approach.

It is clear from Fig. \ref{fig:chernoff} that the learning approach and the likelihood-ratio method perform similarly for this simple discrimination task. However, the likelihood ratio quickly becomes difficult to implement for practical problems such as face recognition, where the structure of the data is highly complex. In contrast, our learning-based approach naturally captures this structure through training, and therefore remains effective even in such complicated settings. Importantly, the learning approach avoids the need to construct the likelihood function $P(\vec{x} \mid H_{0,1})$, which is a significant advantage in practical applications.

As $\alpha$ increases, or equivalently as the separation between the two point sources becomes larger, both the learning approach and the likelihood-ratio method eventually approach the performance predicted by the Chernoff bound. We note that in a binary discrimination task, the minimal possible success probability $P_{\text{succ}}$ is $50\%$, corresponding to random guessing. In contrast, the Chernoff bound predicts $P_{\text{succ}} = 1 - \exp(-\text{constant} \times S)$, which can be smaller than $50\%$—for instance, by formally taking the sample number $S=0$. This is expected because the Chernoff bound is an asymptotic bound that applies only when $S$ is sufficiently large. This asymptotic nature explains why both the learning approach and the likelihood-ratio method initially outperform the Chernoff prediction for small $\alpha$, but converge to the Chernoff scaling as $\alpha$ becomes larger.

In principle, the learning-based approach should not outperform the limit predicted by the Chernoff bound, since the Chernoff bound characterizes the fundamental performance limit of a discrimination task. However, we emphasize that the Chernoff bound is not the complete story for quantifying performance in many discrimination problems. Its limitations already appear in this simple example of distinguishing a single point source from two point sources. To approach the scaling of the exponential index $\varsigma = -\log(1 - P_{\text{succ}})$ predicted by the Chernoff bound, we increase $\alpha$ while keeping the sample number $S$ fixed. When the predicted scaling of $\varsigma = -\log(1 - P_{\text{succ}})$ is finally reached, the success probability $P_{\text{succ}}$ is already above $99\%$. In other words, the Chernoff bound accurately reflects performance only in the regime where the task is almost certain to succeed in this example. In practical applications, however, our attention is not on how quickly $P_{\text{succ}}$ increases once it is already above $99\%$, but on when we can reach a moderate level of performance such as $70\%$ and what strategy should be adopted to achieve it. Such nonasymptotic behavior is clearly not captured by the Chernoff bound in our simulation, and the Chernoff bound also does not indicate what strategy should be used to reach the desired performance when we face a complex task for which the likelihood ratio method is not easily applicable. By contrast, our learning-based approach provides the total REC as a meaningful figure of merit in the finite-sample regime, and, more importantly, offers a principled strategy for tackling discrimination tasks when sample number $S$ is finite.

For more complex tasks such as the face recognition example in the main text, it becomes difficult to make a direct comparison with the likelihood ratio method. We do not know how to construct the likelihood function,  which is not feasible for complex imaging tasks in practice. In addition, we are not working in the regime where the success probability approaches $100\%$, so the Chernoff bound would not provide a meaningful benchmark. Instead, we show in Fig. \ref{fig:rate_face} the performance obtained from the learning-based approach. In the plot, the exponential index $\varsigma = -\log(1 - P_{\text{succ}})$ is shown as a function of $\alpha$ for a fixed sample number, using exactly the same setting as in the main text and in Sec. \ref{SI:face} of the Supplemental Material. As the success probability depends on the truncation order $\mathcal{K}$, we vary $\mathcal{K}$ and take the peak success probability in this calculation.
The calculation shows a step-like structure. Because $C_T$ varies in discrete steps as a function of $\alpha$, the number of measurable features also grows in a stepwise manner. Correspondingly, both the success probability $P_{\text{succ}}$ and the exponential index $\varsigma = -\log(1 - P_{\text{succ}})$ display this step-like behavior as $\alpha$ increases. Importantly, the success probability in this regime is far from $100\%$, and the learning approach enables us to probe performance well before the asymptotic regime predicted by the Chernoff bound is reached, while also providing a principled strategy for handling complex discrimination tasks in the finite sample regime.
}

\end{document}